\documentclass{article}
\pdfoutput=1
\usepackage[toc,page,header]{appendix}
\usepackage{minitoc}
\usepackage[utf8]{inputenc}

\usepackage{ifthen}
\newboolean{useTemplateICML}

\setboolean{useTemplateICML}{true}

\ifthenelse{\boolean{useTemplateICML}}
{
  \usepackage[accepted]{icml2023}

}
{
  \usepackage{arxiv}
}

\usepackage{microtype}
\usepackage{graphicx}
\usepackage{subfigure}
\usepackage{booktabs} 
\usepackage{array, makecell}
\usepackage{afterpage}

\usepackage{float}
\usepackage[hidelinks]{hyperref}

\usepackage{amsmath}
\usepackage{amssymb}
\usepackage{mathtools}
\usepackage{amsthm}
\usepackage{bm}
\usepackage[table]{colortbl}
\usepackage{thm-restate}

\usepackage[capitalize,noabbrev]{cleveref}

\theoremstyle{plain}
\newtheorem{theorem}{Theorem}[section]
\newtheorem{proposition}[theorem]{Proposition}
\newtheorem{lemma}[theorem]{Lemma}

\theoremstyle{definition}
\newtheorem{definition}[theorem]{Definition}

\theoremstyle{remark}

\newcommand{\xtest}{X_{\mathrm{te}}}
\newcommand{\ytest}{Y_{\mathrm{te}}}

\newcommand{\bfx}{\mathbf{x}}
\newcommand{\bfy}{\mathbf{y}}

\newcommand{\R}{\mathbb{R}}
\newcommand{\E}{\mathbb{E}}
\newcommand{\GP}{\mathcal{GP}}

\newcommand{\featurespace}{\R^2}
\newcommand{\bfeps}{\boldsymbol{\epsilon}}
\newcommand{\khelm}{k_{\mathrm{Helm}}}
\newcommand{\kvel}{k_{\mathrm{vel}}}
\newcommand{\obsvar}{\sigma^2_{\mathrm{obs}}}
\newcommand{\Xtrain}{\mathbf{X}_{\mathrm{tr}}}
\newcommand{\Xtest}{\mathbf{X}_{\mathrm{te}}}
\newcommand{\Ytrain}{\mathbf{Y}_{\mathrm{tr}}}
\newcommand{\Ytest}{\mathbf{Y}_{\mathrm{te}}}
\newcommand{\xtrain}{X_{\mathrm{tr}}}
\newcommand{\ytrain}{Y_{\mathrm{tr}}}
\newcommand{\Khtrtr}{K_{\mathrm{trtr}}} 
\newcommand{\Khtetr}{K_{\mathrm{tetr}}} 
\newcommand{\Khtete}{K_{\mathrm{tete}}} 
\newcommand{\sigmaFu}{\sigma_{1}}
\newcommand{\sigmaFv}{\sigma_{2}}
\newcommand{\sehelmgp}{\text{SE-Helmholtz GP}}
\newcommand{\sevelgp}{\text{SE-velocity GP}}

\newcommand{\bk}{\ensuremath {\boldsymbol {k}}}
\newcommand{\bi}{\ensuremath {\boldsymbol {i}}}
\newcommand{\bj}{\ensuremath {\boldsymbol {j}}}

\newcommand{\Var}{\mathrm{Var}}
\newcommand{\Cov}{\mathrm{Cov}}
\newcommand{\dt}[2]{#1^{(#2)}}

\newcommand{\lightcell}{\cellcolor{green!40!white}} 
\newcommand{\darkcell}{\cellcolor{green!80!black}}

\graphicspath{{images/}}
\usepackage{blindtext}
\usepackage{subfiles}

\usepackage[textsize=tiny]{todonotes}
\begin{document}

\ifthenelse{\boolean{useTemplateICML}}
{
    \icmltitlerunning{Gaussian Processes at the Helm(holtz)}
    \twocolumn[
    \icmltitle{Gaussian \nolinebreak Processes \nolinebreak at \nolinebreak the \nolinebreak Helm(holtz): \nolinebreak A \nolinebreak More \nolinebreak Fluid \nolinebreak Model \nolinebreak for \nolinebreak Ocean \nolinebreak Currents}
    \icmlsetsymbol{equal}{*}

    
    

    \begin{icmlauthorlist}
    \icmlauthor{Renato Berlinghieri}{mit}
    \icmlauthor{Brian L.~Trippe}{col}
    \icmlauthor{David R.~Burt}{mit}
    \icmlauthor{Ryan Giordano}{mit}
    \icmlauthor{Kaushik Srinivasan}{ucla}
    \icmlauthor{Tamay \"Ozg\"okmen}{miami}
    \icmlauthor{Junfei Xia}{miami}
    \icmlauthor{Tamara Broderick}{mit}
    \end{icmlauthorlist}
    
    \icmlaffiliation{mit}{Department of Electrical Engineering and Computer Science, Massachusetts Institute of Technology}
    \icmlaffiliation{col}{Department of Statistics, Columbia University}
    \icmlaffiliation{ucla}{Department of Atmospheric and Oceanic Sciences, UCLA}
    \icmlaffiliation{miami}{Department of Ocean Sciences, Rosenstiel School of Marine and Atmospheric Science, University of Miami}
    
    \icmlcorrespondingauthor{Renato Berlinghieri}{renb@mit.edu}

    \icmlkeywords{Gaussian Processes, Helmholtz, currents, divergence, vorticity}

    
    \vskip 0.3in
    ]


    \printAffiliationsAndNotice{} 
}
{
    \title{Gaussian processes at the Helm(holtz): \\ A more fluid model for ocean currents}
    \setshorttitle{Gaussian processes at the Helm(holtz): A more fluid model for ocean currents} 
    \author{Renato Berlinghieri$^1$, Brian L.~Trippe$^2$, David R.~Burt$^1$, Ryan Giordano$^1$, \\ \textbf{Kaushik Srinivasan}$^3$, \textbf{Tamay \"{O}zg\"{o}kmen}$^4$, \textbf{Junfei Xia}$^4$, \textbf{Tamara Broderick}$^1$}
    \date{
    $^1$Department of Electrical Engineering and Computer Science, Massachusetts Institute of Technology;\\%
    $^2$Department of Statistics, Columbia University;\\
    $^3$Department of Atmospheric and Oceanic Sciences, UCLA;\\
    $^4$Department of Ocean Sciences, Rosenstiel School of Marine and Atmospheric Science, University of Miami.
    }
    \maketitle  
}

\doparttoc 
\faketableofcontents 

\begin{abstract}
Given sparse observations of buoy velocities, oceanographers are interested in reconstructing ocean currents away from the buoys and identifying divergences in a current vector field.
As a first and modular step, we focus on the time-stationary case -- for instance, by restricting to short time periods.
Since we expect current velocity to be a continuous but highly non-linear function of spatial location, Gaussian processes (GPs) offer an attractive model.
But we show that applying a GP with a standard stationary kernel directly to buoy data can struggle at both current reconstruction and divergence identification, due to some physically unrealistic prior assumptions.
To better reflect known physical properties of currents, we propose to instead put a standard stationary kernel on the divergence and curl-free components of a vector field obtained through a Helmholtz decomposition.
We show that, because this decomposition relates to the original vector field just via mixed partial derivatives, we can still perform inference given the original data with only a small constant multiple of additional computational expense.
We illustrate the benefits of our method with theory and experiments on synthetic and real ocean data. 



\end{abstract}

\ifthenelse{\boolean{useTemplateICML}}
{
}
{
}

\section{Introduction}
Ocean currents are key to the global distribution of water, heat, and nutrients. To better understand ocean currents, scientists are interested in two tasks: (1) reconstructing ocean currents at different locations and (2) identifying divergences 
in the current vector field.
Reconstructing ocean currents accurately can facilitate weather forecasting, maritime navigation, and forecasting of oil spill dispersion.
And current divergences are important to identify since they are responsible for the transport of biomass, carbon, and nutrients -- with implications for ecosystem management, climate, and the fishing industry \citep{dasaro2018ocean}. 
With these tasks in mind, researchers release and track GPS-tagged buoys in the ocean \citep{ozgokmen2012carthe,dasaroLASERdata2017}.

It remains to choose an appropriate method to reconstruct currents and their divergences from buoy data. \citet{goncalves2019naive} and \citet{lodise2020investigating} proposed modeling buoy velocities in the latitude and longitude directions according to independent Gaussian processes (GPs) with standard spatiotemporal kernels (e.g., squared exponential kernels). In our work, we focus on the spatial aspects of this task and assume the velocity field is stationary in time. Even under this simplification, an independent spatial GP prior on the velocities is a natural choice due to its ability to handle the sparsity of buoy observations on the ocean surface and its assumption that currents vary continuously but in a nonlinear fashion. We call such a model the \emph{velocity GP}.


However, in what follows, we show that there remains substantial room for improvement. In simulated cases where we have access to ground truth, we observe that the velocity GP approach can fail to complete vortices or fail to connect currents when buoys are observed sparsely. And while we show how to derive divergence estimates for the velocity GP, we also find that these estimates often fail to capture the true divergence when it is known in simulations or real data.

To address these issues, we propose to more directly model known behaviors from fluid dynamics.
Scientists know that the motion of a volume element of a continuous fluid medium in two dimensions consists of (i) expansion or contraction in two orthogonal directions, (ii) rotation about an instantaneous axis, and (iii) translation. A \emph{Helmholtz decomposition} \citep{bhatia2013helmoltzhodge, arfken1999mathematical} from fluid dynamics lets us decompose the vector field of ocean currents into a \emph{divergent} component (or \emph{curl-free}, measuring expansion, contraction, and translation) and a \emph{rotational} component (or \emph{divergence-free}, measuring rotation).\footnote{The divergent component is called curl-free and the rotational component is called divergence-free because their curl and divergence are zero everywhere, respectively. See \cref{prop:div-curl-free}.
}

By contrast to the standard approach, we model the divergent and rotational components (rather than the velocity components) with independent GP priors.
The resulting \emph{Helmholtz GP} prior offers several conceptual advantages.
For one, oceanographers expect the two components to have substantially different magnitudes and length scales; it is straightforward to encode these differences with a Helmholtz GP.
By contrast, we prove that the velocity GP implies an a priori belief that the divergent and rotational components have the same magnitude.
Second, we expect correlation between the longitudinal and latitudinal components of a current, which the Helmholtz GP exhibits -- and the velocity GP lacks by construction.
Finally, the Helmholtz GP is agnostic to the (arbitrary) choice of reference frame, while the velocity GP is not necessarily.

We demonstrate that the Helmholtz GP is amenable to practical inference. Since (i) the Helmholtz decomposition is based on partial (mixed) derivatives, and (ii) the derivative of a GP is a GP, we show that our prior choice implies a tractable GP prior on the current itself. Therefore, we can still perform inference given the original data with no extra approximation. And our method suffers no increase in computational complexity relative to the velocity GP.

Finally, we demonstrate the superior performance of the Helmholtz GP at the current-reconstruction task and divergence-estimation task (as well as vorticity estimation) in a variety of experiments on simulated and real data. Code is available at \url{https://github.com/renatoberlinghieri/Helmholtz-GP}.
 
\textbf{Related work.} 
The Helmholtz decomposition has been used 
for visualization and physical interpretation of an oceanographic field 
\citep{rocha2016mesoscale, zhang2018helm1, zhang2019helm3, lei2020helm2, buhler2014wave, caballero2020integration}.
But these methods assume the velocity vector field is known on a grid whereas our goal is to reconstruct the vector field from sparse observations.
Prior work in atmospheric statistics used a Helmholtz decomposition to perform regression on the residuals of a physical model \citep{daley1985analysis, hollingsworth1986statistical}; this approach would give the same mean prediction\footnote{Here and throughout, we use the word \emph{prediction} in the machine learning sense; it describes the task of making informed guesses about unseen data points using a trained model and need not imply looking forward in time.} as the Helmholtz GP if the same covariance function were used. However, to estimate the covariance function, the authors rely on a series representation and binning procedure; without hand-tuning that depends on the data and physics of the system (e.g., for determining the number of coefficients), this procedure can result in covariance functions that are not positive definite. 


 Researchers have developed GP kernels to capture curl- or divergence-free fields \citep{narcowich1994generalized, lowitzsch2002approximation, fuselier2007refined, macedo2010learning,alvarez2012kernels}. \citet{macedo2010learning} propose using convex combinations of such kernels. However, these works do not propose methods for recovering the weighting of the two components and do not empirically test recovery of the components when the weighting is unknown -- the case of interest in the oceans problem.
\citet{wahlstrom2013modeling, wahlstrom2015modeling, solin2018modeling} use curl- and divergence-free kernels for electromagnetic fields. \citet{wahlstrom2015modeling} proposes independent GP priors on the terms in a Helmholtz decomposition. But the authors assume direct access to noisy observations of each of the divergence-free and curl-free components separately -- whereas we aim to recover the individual components from noisy observations of their sum. Moreover, \citet{wahlstrom2015modeling} constrains the two components to have the same prior magnitudes and length scales, but these quantities can be expected to vary substantially between components in ocean currents.
Finally, \citet{greydanus2022dissipative} extended Hamiltonian Neural Networks \citep{greydanus2019hamiltonian} to model both curl- and divergence-free dynamics simultaneously. Although the prediction problem is similar, the authors test their method only on low-resolution data available on a dense grid. In our experiments on (sparse) buoy data in \cref{sec:experiments} and \cref{app:experiment}, we find that their method often produces predictions that are less accurate and less physically plausible. 

In sum, then, it is not clear from existing work that divergence and vorticity can be usefully or practically recovered when observations come from a general (noisy) vector field that is neither curl- nor divergence-free. Moreover, there is no existing guidance on how to use a Helmholtz GP in practice for identifying divergences or making predictions from such noisy vector-field observations, and no information is available on how a Helmholtz GP compares to a velocity GP on these tasks -- either empirically or theoretically. We discuss related work further in \cref{app:related_work}.

\label{introduction}

\section{Background}\label{sec:background}

In what follows, we first describe the problem setup. Then we establish necessary notation and concepts from the Helmholtz decomposition and Gaussian processes.

\textbf{Problem Statement.} 
We consider a dataset $D$ of $M$ observations, $\{(\bfx_m, \bfy_m)\}_{m=1}^M$. Here $\bfx_m = (x_m^{(1)}, x_m^{(2)})^\top \in \R^2$ represents the location of a buoy, typically a longitude and latitude pair. We treat $\bfx_m$ as a column vector. And $\bfy_m = (y_m^{(1)}, y_m^{(2)})^\top \in \R^2$ gives the corresponding longitudinal and latitudinal velocities of the buoy (the  \emph{drifter trace}). For $m \in \{1, \ldots, M\}$, we consider $\bfy_m$ as a sparse noisy observation of a 2-dimensional vector field, $F : \R^2 \rightarrow \R^2$, mapping spatial locations into longitudinal and latitudinal velocities, $F(\bfx_m) = (F^{(1)}(\bfx_m), F^{(2)}(\bfx_m))^\top$. We assume that the velocity field is stationary in time, and so $F$ is not a function of time. Our primary goals are (1) prediction of the field $F$ at new locations, not observed in the training data, and (2) estimation of the divergence, itself a function of location and which we define next as part of the Helmholtz decomposition. Secondarily, we are interested in recovering vorticity, another functional of $F$ described below.

\textbf{The Helmholtz Decomposition.} The motion of a volume element of a fluid, such as the ocean, can be decomposed into a divergent velocity and a rotational velocity. 

\begin{definition}[Helmholtz decomposition, {{\citealp{bhatia2013helmoltzhodge}}}]
\label{th:helmdec}
A twice continuously differentiable and compactly supported vector field $F : \mathbb{R}^2 \rightarrow \mathbb{R}^2$ can be expressed as the sum of the gradient of a scalar potential $\Phi: \R^2 \to\R$, called the \textit{potential function}, and the vorticity operator of another scalar potential $\Psi: \R^2 \to \R$, called the \textit{stream function}:
\begin{align}
\underbrace{F}_{\text{ocean flow}} = \underbrace{\mathrm{grad} \,\Phi}_{\text{divergent velocity}} + \underbrace{\mathrm{rot}\, \Psi}_{\text{rotational velocity}} \label{eqn:helmholtz-decomp}
\end{align}
where 
\begin{align}
\mathrm{grad} \,\Phi \!\coloneqq\!
\begin{bmatrix}
\partial\Phi/\partial x^{(1)} \\ \partial\Phi/\partial x^{(2)}
\end{bmatrix} \, \text{and} \,
\mathrm{rot}\,\Psi \!\coloneqq\!
\begin{bmatrix}
\partial\Psi/\partial x^{(2)} \\
-\partial\Psi/\partial x^{(1)}
\end{bmatrix}.
\end{align}
\end{definition}

The \emph{divergence} of $F$ (denoted $\delta$) and the \emph{vorticity} of $F$ (denoted $\zeta$) are
\begin{align}
\label{eqn:divergence-standard}
\delta &\!\coloneqq\!
\mathrm{div}(F) \!\coloneqq\!
\frac{\partial F^{(1)}}{\partial x^{(1)}} +        \frac{\partial F^{(2)}}{\partial x^{(2)}}
=\frac{\partial^2\Phi}{\partial^2 x^{(1)}} + \frac{\partial^2\Phi}{\partial^2 x^{(2)}}
 \\
\label{eqn:vorticity-standard}
\zeta &\!\coloneqq\!
\mathrm{curl}(F)\!\coloneqq\!
\frac{\partial F^{(1)}}{\partial x^{(2)}} - \frac{\partial F^{(2)}}{\partial x^{(1)}}
\!=\frac{\partial^2\Psi}{\partial^2 x^{(2)}} + \frac{\partial^2\Psi}{\partial^2 x^{(1)}}. 
\end{align}

In \cref{eqn:divergence-standard}, $\mathrm{div}(F)$ depends only on $\Phi$ because $\mathrm{div}  (\mathrm{rot}\,\Psi ) = 0$. In other words, the rotational velocity is divergence-free. Similarly, in \cref{eqn:vorticity-standard}, $\mathrm{curl}(F)$ depends only on $\Psi$ because $\mathrm{curl} (\mathrm{grad} \,\Phi ))=0$. In other words, the divergent velocity is curl-free. We review the $\mathrm{grad}$, $\mathrm{rot}$, $\mathrm{div}$, and $\mathrm{curl}$ operators -- and explore the equations above in more detail -- in \cref{app:divcurl}. We summarize the various terms in \cref{tab:relationships}. In \cref{app:helmholtz}, we present a graphical illustration of a Helmholtz decomposition of a selected vector field, and we further discuss the importance of divergence and vorticity within ocean currents.




\begin{table}[t]
\caption{Terms and notation around the divergence and vorticity.}
\label{tab:relationships}
\begin{center}
\begin{tabular}{c | c}
$\Phi$ & potential function \\
$\mathrm{grad}\,\Phi$ & divergent velocity \\

$\delta = \mathrm{div} (\mathrm{grad} \,\Phi)$ & divergence \\ \hline
$\Psi$ & stream function \\ 
$\mathrm{rot}\,\Psi$ & rotational velocity \\ 
$\zeta = \mathrm{curl}(\mathrm{rot} \,\Psi)$ & vorticity
\end{tabular}
\end{center}
\end{table}

\textbf{Bayesian Approach and Gaussian Process Prior.} 
In what follows, we will take a Bayesian approach to inferring $F$. In particular, we assume a likelihood, or noise model, relating the observed buoy velocities to the field $F$:
\begin{align}\label{eqn:likelihood}
\!\bfy_m\! =\! F(\bfx_m) + \bfeps_m, \bfeps_m  \!\stackrel{\text{ind}}{\sim}\! \mathcal{N}(0, \obsvar\mathbf{I}_2), 1 \!\leq\! m \!\leq\! M,
\end{align}
for some $\obsvar>0$ and independent ($\stackrel{\text{ind}}{\sim}$) noise across observations. Here and throughout, we use $\mathbf{I}_p \in R^{p \times p}$ to denote the identity matrix in $p$ dimensions. We use $0$ to denote the zero element in any vector space.

Before defining our prior, we review Gaussian processes (GPs). Let $\bfx, \bfx'\in \featurespace$ represent two input vectors. Assume that we want to model a $P$-dimensional function $G: \featurespace \to \R^P$, $G(\bfx) = (G^{(1)}(\bfx), \dotsc, G^{(P)}(\bfx))^\top$.  A $P$-output GP on covariate space $\featurespace$ is determined by a mean function $\mu : \featurespace \to \R^P$, $\mu(\bfx) = (\mu^{(1)}(\bfx), \dotsc, \mu^{(P)}(\bfx))^\top$, and a positive definite kernel function $k : \featurespace \times \featurespace \to \R^{P \times P}$. We use $k(\bfx, \bfx')_{i, j}$ to denote the $(i,j)$th output of $k(\bfx, \bfx')$. We say that $G$ is GP distributed and write $G\sim \mathcal{GP}(\mu, k)$ if for any $N \in \mathbb{N}$, for any $(\bfx_1, \dotsc, \bfx_N) \in \R^{2 \times N}$, and for any vector of indices $(p_1, \dotsc, p_N) \in \{1, \dotsc, P\}^N$, $(G^{(p_n)}(\bfx_n))_{n=1}^N$ is an $N$-dimensional Gaussian random variable with mean vector $(\mu^{(p_n)}(\bfx_n))_{n=1}^N$ and covariance matrix with $(i, j)$th entry $k(\bfx_i, \bfx_j)_{p_i, p_j}$. See \citet{alvarez2012kernels} for a review of multi-output GPs.

\textbf{Velocity Gaussian Process.} In spatial data analysis, commonly $\mu$ is chosen to be identically 0. And a conventional choice for $k$ would be an isotropic kernel\footnote{We say a kernel $k$ is \emph{isotropic} if there exists some $\kappa: \R^+\rightarrow \R$ such that for any $\bfx$ and $\bfx'$ in $\R^2,\ k(\bfx, \bfx')=\kappa(\|\bfx-\bfx'\|).$} separately in each output dimension. That is, for any $\bfx, \bfx' \in \featurespace$,
\begin{align}
 	k_{\text{vel}}(\bfx, \bfx')=  \begin{bmatrix} k^{(1)}(\bfx, \bfx') & 0 \\
  0 & k^{(2)}(\bfx, \bfx') 
  \end{bmatrix}.
  \label{eqn:vel-kernel}
\end{align}
where $k^{(1)}$ and $k^{(2)}$ are isotropic kernels. We call this choice the \emph{velocity GP} to emphasize that the independent priors are directly on the observed velocities. A standard kernel choice for $k^{(i)}$, $i \in \{1,2\}$, is the squared exponential kernel,
\begin{align}\label{eqn:sq-exp-kernel} 
k_{\text{SE}}^{(i)}(\bfx, \bfx') = \sigma^2_{i} \exp\left(-\tfrac{1}{2}\|\bfx -  \bfx'\|_2^2/\ell_i^2\right).
\end{align}

The velocity GP with squared exponential kernels for each component (henceforth, the \emph{\sevelgp}) has four hyperparameters: for $i \in \{1,2\}$, the signal variance $\sigma^2_i > 0$ determines the variation of function values from their mean in the $i$th output dimension, and $\ell_i > 0$ controls the length scale on which the function varies. 



\section{Gaussian Processes at the Helm(holtz)} 


%


Instead of putting separate GP priors with isotropic kernels on the two components of $F$ as in the velocity GP, we propose to put separate GP priors with isotropic kernels on the Helmholtz scalar potentials $\Phi$ and $\Psi$. In this section, we describe our model and how to retrieve the quantities of interest from it. In the next section, we describe its conceptual strengths over the velocity GP, which we see empirically in \cref{sec:experiments}.

\looseness=-1 \textbf{Our Helmholtz GP prior.} To form our new \emph{Helmholtz GP} prior, we put independent GP priors on the Helmholtz stream and potential functions:
\begin{equation}
\label{eqn:helmprior}
 \Phi \sim \GP(0, k_{\Phi}) \quad \text{and} \quad
    \Psi \sim \GP(0, k_{\Psi}), 
\end{equation} 
where we take $k_{\Phi}$ and $k_{\Psi}$ to be isotropic kernels. When these kernels are chosen to be squared exponentials (\cref{eqn:sq-exp-kernel}), we call our model the \emph{\sehelmgp}. The $\sehelmgp$ has four parameters: $\ell_{\Phi}$ and $\sigma^2_{\Phi}$ for $k_{\Phi}$, and $\ell_{\Psi}$ and $\sigma^2_{\Psi}$ for $k_{\Psi}$. We could use any two kernels such that sample paths of the resulting GPs are almost surely continuously differentiable (so that $F$ in \cref{eqn:helmholtz-decomp} is well-defined and continuous). Generally, we will want to be able to consider divergences and vorticities of the implied process, which will require sample paths of the implied process to be at least twice-continuously differentiable. For the latter condition to hold, it is sufficient for $k_{\Phi}(0, \bfx)$ and $k_{\Psi}(0, \bfx)$ to have continuous mixed partial derivatives up to order five; see \citet[Theorem 2.09 \& Section 7.2]{lindgren2012stationary}.

First, we check that our prior yields a GP prior over the vector field $F$.

\begin{restatable}{proposition}{helmprior}\label{prop:helm-prior}
Let $F$ be an ocean current vector field defined by potential and stream functions that are independent and distributed as 
$\Phi \sim \GP(0, k_{\Phi})$ and 
$\Psi \sim \GP(0, k_{\Psi})$, where $k_{\Phi}$ and $k_{\Psi}$ are such that $\Phi$ and $\Psi$ have almost surely continuously differentiable sample paths.
Then
\begin{align}
F = \mathrm{grad} \,\Phi + \mathrm{rot}\, \Psi \sim \GP(0, \khelm ),
\end{align}
where, for $\bfx, \bfx'\in \featurespace$, $i,j \in {1,2}$, $\khelm(\bfx, \bfx')_{i,j}$ is equal to
\begin{align} 
\dfrac{
        \partial^2 k_{\Phi}(\bfx, \bfx')
    }{
        \partial x^{(i)} \partial (x')^{ {(j)}}
    }
    + (-1)^{\mathbf{1}\{i \ne j\}}
    \dfrac{
        \partial^2 k_{\Psi}(\bfx, \bfx')
    }{
        \partial x^{(3-i)} \partial (x')^{ {(3-j)}}
    }.
\end{align}
\end{restatable}
Our proof in \cref{app:prior-helm-gp} relies on two observations: (i) the Helmholtz decomposition is based on partial (mixed) derivatives and (ii) the derivative of a GP is a GP; see, e.g., \citet[Chapter 9.4]{rasmussen2005gp}, and \citet[Theorem 2.2.2]{adler1981geometry}.

\textbf{Making predictions.}
To make predictions using our Helmholtz GP, we need to choose the hyperparameter values and then evaluate the posterior distribution of the ocean current given those hyperparameters.

We choose the GP hyperparameters by maximizing the log marginal likelihood of the training data.
To write that marginal likelihood, we let $\Xtrain \in \R^{2 \times M}$ be the matrix with $m$th column equal to $\bfx_m$.
We define $\Ytrain = (\bfy_1^{(1)}, \ldots, \bfy_M^{(1)},\bfy_1^{(2)}, \ldots, \bfy_M^{(2)})^\top \in \R^{2M}$. We extend the definition of the mean and kernel function to allow for arbitrary finite collections of inputs. In particular, for $\mathbf{X}=(\bfx_1, \dotsc, \bfx_N) \in \R^{2 \times N}$ and $\mathbf{X'}=(\bfx'_1, \dotsc, \bfx'_{N'}) \in \R^{2 \times N'}$,
\begin{align}
    \mu(\mathbf{X}) &= 
    \begin{pmatrix}
        \mu^{(1)}(\mathbf{X}) \\ \mu^{(2)}(\mathbf{X})
    \end{pmatrix}
    \, \text{and} \, \\
    k(\mathbf{X}, \mathbf{X'}) &= 
    \begin{pmatrix}
        k(\mathbf{X}, \mathbf{X'})_{1,1} & k(\mathbf{X}, \mathbf{X'})_{1,2} \\ 
        k(\mathbf{X}, \mathbf{X'})_{2,1} & k(\mathbf{X}, \mathbf{X'})_{2,2}
    \end{pmatrix}
\end{align}
where (a) for $i \in \{1,2\}$, $n \in \{1,\ldots,N\}$, $\mu^{(i)} (\mathbf{X})$ is an $N$-dimensional column vector with $n$th entry $\mu^{(i)} (\bfx_n)$, and (b) for $i,j \in \{1,2\}$, $n \in \{1,\ldots,N\}$, $n' \in \{1,\ldots,N'\}$, $k(\mathbf{X},\mathbf{X'})_{i,j}$ is an $N \times N'$ matrix with $(n, n')$th entry $k(\bfx_n,\bfx'_{n'})_{i,j}$. With this notation, we denote the covariance of the training data with itself, under the full model including noise, as $\Khtrtr = k(\Xtrain,\Xtrain) + \obsvar \mathbf{I}_{2M}$. Then the log marginal likelihood is 
\begin{align}
\label{eqn:log-marginal-likelihood}
\begin{split}
&\log\, p(\Ytrain \mid \Xtrain)= \log \mathcal{N}(\Ytrain; 0, \Khtrtr) \\ 
    &=-\frac{1}{2} \Ytrain^T \Khtrtr^{-1}\Ytrain - \frac{1}{2} \log |\Khtrtr| - \frac{2M}{2} \log 2\pi,
\end{split}
\end{align}
where $|\cdot |$ takes the determinant of its matrix argument.
We provide details of our optimization procedure in \cref{sec:experiments}.

With hyperparameter values in hand, we form probabilistic predictions using the posterior of the GP. In particular, the posterior mean forms our prediction at a new set of points, and the posterior covariance encapsulates our uncertainty.

Consider $N$ new (test) locations at which we would like to predict the current. We gather them in $\Xtest \in \R^{2 \times N}$, with $n$th column equal to $\bfx_n^\star$. We denote the covariance of various training and testing combinations as:
$\Khtetr = k(\Xtest,\Xtrain)$ and
$\Khtete = k(\Xtest,\Xtest)$.
Then a posteriori after observing the training data $D$, the $2N$-long vector $(F^{(1)}(\bfx_1^\star), \dotsc,F^{(1)}(\bfx_N^\star),\dotsc,F^{(2)}(\bfx_1^\star)\dotsc, F^{(2)}(\bfx_N^\star))^\top$ describing the current at the test locations has a normal distribution with mean and covariance
\begin{align}
    \mu_{F \mid D} &= \Khtetr \Khtrtr^{-1} \Ytrain,  \label{eqn:helm_posterior_mean} \\
    K_{F \mid D} &= \Khtete - \Khtetr \Khtrtr^{-1} \Khtetr^\top. \label{eqn:helm_posterior_cov}
\end{align}

For more details, see \citet[Section 2.2]{rasmussen2005gp}. Note that these formulas can be used to evaluate posterior moments of the velocity field for either the Helmholtz GP (setting $k = \khelm$) or the velocity GP (with $k = \kvel$).

\textbf{Recovering divergence and vorticity.} We next show how to recover the posterior distributions on the divergence and vorticity scalar fields given a posterior on the current field $F$. We can estimate divergence and vorticity at any location by using the posterior mean at that point, and we can report uncertainty with the posterior variance. Note that our formulas recover divergence and vorticity for either our Helmholtz GP or the velocity GP. 

\begin{restatable}{proposition}{divergenceofagp}
\label{prop:divergenceofagp}
Let $F \sim \mathcal{GP}(\mu, k)$ be a two-output Gaussian process with almost surely continuously differentiable sample paths. Then, for $\bfx, \bfx'\in \featurespace$, 
\begin{align}
&\delta = \mathrm{div} \, F \sim \mathcal{GP}(\mathrm{div} \, \mu, k^{\delta}) \\ 
&\zeta = \mathrm{curl} \, F \sim \mathcal{GP}(\mathrm{curl} \, \mu, k^{\zeta})
\end{align}
where
%
\begin{align}
 &k^{\delta}(\bfx, \bfx') = \!\!\!\!\!\sum_{(i,j) \in \{1, 2\}^2} \!\!\!\!\!\frac{\partial^2 k(\bfx,\bfx')_{i,j}}{\partial x^{(i)}\partial x^{(j)}}
\\ 
&k^{\zeta}(\bfx, \bfx') \!=\!\!\!\!\! \sum_{(i,j) \in \{1, 2\}^2}\!\!\!\! (-1)^{i+j}\frac{\partial^2 k(\bfx, \bfx')_{i,j}}{\partial x^{(3-i)}\partial x^{(3-j)}}.
\end{align}
\end{restatable}

We provide the proof for \cref{prop:divergenceofagp} in \cref{app:div-vort-gp}.


\textbf{Computational Cost.} Since the latitude and longitude outputs are correlated under the Helmholtz GP, it generally has a higher computational cost than the velocity GP. We establish that the extra cost is no worse than a small factor.

\begin{proposition}
\label{prop:comp-cost}
    Take $M$ training data points. Let $C_{vel}(M)$ and $C_{helm}(M)$ be the computational costs for evaluating the log marginal likelihood (\cref{eqn:log-marginal-likelihood}) via 
    Cholesky or QR factorization algorithms for the velocity GP and Helmholtz GP, respectively. If we assume worst-case scaling for these algorithms, 
        $\lim_{M \rightarrow \infty} C_{helm}(M)/C_{vel}(M) \leq 4.$
\end{proposition}

\looseness=-1 The cost of computing the log marginal likelihood is dominated by the cost of solving the linear system $\Khtrtr^{-1}\Ytrain$ and computing the log determinant $|\Khtrtr|$. Both of these costs in turn arise primarily from the cost of factorizing 
$\Khtrtr$. Let $\mathrm{CF}(s)$ be the cost of factorizing a square matrix with $s$ rows with Cholesky or QR factorization. Due to the two (correlated) outputs, the cost of the Helmholtz GP is dominated by $\mathrm{CF}(2M)$. In the velocity GP, the two outputs are uncorrelated and can be handled separately, so the cost is dominated by $2 \mathrm{CF}(M)$. Therefore, $\lim_{M \rightarrow \infty} C_{helm}(M) / C_{vel}(M)\leq \mathrm{CF}(2M)/ (2 \mathrm{CF}(M))$.
When factorizing the matrix costs $\mathrm{CF}(s) \sim cs^p$ for $p \in (0,3], c > 0$, the result follows by \cref{eqn:asymptotic-ratio}. Standard Cholesky and QR factorization algorithms satisfy the condition with $p = 3$ in the worst case \citep[p. 164, 249]{golub2013matrix}.

In \cref{app:comp-cost} we provide similar computational results for the task of prediction and discuss nuances of how any of these results may change in the presence of special structure.

\section{Advantages of the Helmholtz prior}\label{sec:advantages}

We next describe three key advantages of the Helmholtz GP prior over the velocity GP prior: (1) more physically realistic prior assumptions reflecting the relative magnitude and length scales of the divergence and vorticity, (2) more physically realistic correlation of the longitudinal and latitudinal velocities of current at any point, and (3) equivariance to reference frame.

\textbf{Prior magnitude of the divergence and vorticity.} In real ocean flows, except at small-scale frontal features, the divergence is known a priori to have both a substantively different magnitude and different length scale relative to the vorticity \citep{barkan2019}. In what follows, we argue that the Helmholtz GP is able
to capture the relative contributions of divergence and vorticity directly in the prior -- whereas the velocity GP does not have this direct control.

On the magnitude side, the divergence is known to contribute much less to the current than the vorticity contributes.
If we consider a \sehelmgp, the signal variance hyperparameters $\sigma_\Phi^2$ and $\sigma_\Psi^2$ control the magnitude of $\Phi$ and $\Psi$;
as a direct consequence of the linearity of the divergence $\delta$ and vorticity $\zeta$ in $\Phi$ and $\Psi$ (\Cref{eqn:divergence-standard,eqn:vorticity-standard}), the marginal variances of $\delta$ and $\zeta$ scale linearly with $\sigma_\Phi^2$ and $\sigma_\Psi^2$, respectively. The model can therefore directly and separately control the magnitude of the rotational and divergence components. A similar argument can be applied to more general Helmholtz GPs with parameters controlling the magnitude of $\Phi$ and $\Psi$.

%



By contrast, the velocity GP provides no such control.
In fact, for any isotropic choice of $k^{(1)}$ and $k^{(2)}$ we show that the resulting velocity GP must assume the same variance on the divergence and vorticity in the prior.

\begin{restatable}{proposition}{propEqualVarDivVort}\label{prop:equal-var-div-vort}
Let $k^{(1)}$ and $k^{(2)}$ be isotropic kernels with inputs $\bfx, \bfx' \in \R^2$. Take $F^{(1)} \sim \mathcal{GP}(0, k^{(1)})$ and $F^{(2)}\sim \mathcal{GP}(0, k^{(2)})$ independent. Suppose  $k^{(1)}$ and $k^{(2)}$ are such that $F^{(1)}, F^{(2)}$ have almost surely continuously differentiable sample paths.
Let $\delta$ and $\zeta$ be defined as in \cref{eqn:divergence-standard,eqn:vorticity-standard}.
Then for any $\bfx,\mathrm{Var}[\delta(\bfx)] = \mathrm{Var}[\zeta(\bfx)]$.
\end{restatable}
The proof of \cref{prop:equal-var-div-vort} appears in \cref{app:sec-equal-marginal-var}.

\begin{table*}[!ht]
\vskip 0.15in
\caption{Green identifies the lowest RMSE. Dark green indicates the RMSE is at least two times smaller than the next best model.}
\begin{center}
\begin{small}
\begin{sc}
\begin{tabular}{c|| c | c| c|| c| c| c || c| c | c}
& \multicolumn{3}{c ||}{Velocity $F$} & \multicolumn{3}{c||}{Divergence $\delta$} & \multicolumn{3}{c}{Vorticity $\zeta$} \\
\toprule
 & \makecell{Helm} & \makecell{Vel} & \makecell{D-hnn} & \makecell{Helm} &  \makecell{Vel} & \makecell{D-hnn} & \makecell{Helm}   & \makecell{Vel} 
   & \makecell{D-hnn} \\
\hline \hline
\makecell{vortex} & \darkcell0.24 & 0.72 & 0.54 & \darkcell0.0 & 0.22 & 0.87 & \lightcell0.77 & 1.05 & 1.03 \\ \hline
\makecell{small divergence} & 1.11 & 1.25 & \lightcell0.67 & 2.62 & \lightcell1.45  & 4.14 & \darkcell0.0  & 1.07 & 0.31 \\ \hline
\makecell{medium divergence} & \lightcell0.17 & 0.19 & 0.55 & 0.39 & \lightcell0.33 & 1.32 & \darkcell0.05 & 0.12 &  0.38\\ \hline
\makecell{big divergence} & \darkcell0.04 & 0.10 & 0.19 & \lightcell0.05 & 0.12 & 0.27 & \darkcell0.00 & 0.10 & 0.11  \\ \hline
\makecell{duffing w/ small divergence} & \darkcell0.96 & 2.05 & 2.14 & \lightcell0.94 & 0.95 & 1.89 & \lightcell1.40 & 2.28 & 2.64 \\ \hline
\makecell{duffing w/ medium divergence} & \darkcell0.19 & 0.60 & 1.65 & \darkcell0.14 & 0.50 & 1.15 & \lightcell0.24 & 0.26 & 2.39 \\ \hline
\makecell{duffing w/ big divergence} & 0.41 & \lightcell0.22 & 1.63 & \darkcell0.08 & 0.17 & 1.10 & 0.48 & \darkcell0.16 & 2.41 \\ \hline
\bottomrule
\end{tabular}
\end{sc}
\end{small}
\end{center}
\vskip -0.1in
\label{tab:mse}
\end{table*}

\textbf{Prior length scales of the divergence and vorticity.} The divergence and vorticity are also known to operate
on very different length scales in real ocean flows.  Vorticity operates over long length scales, whereas divergence tends to be more localized. Similarly to the argument above, the Helmholtz GP allows control over the length scale in each of its components, which directly control the length scale of the divergence and vorticity. 
In particular, if $k_\Phi(\bfx, \bfx') = \kappa(\|\bfx-\bfx'\|/\ell)$, for some $\kappa:\R^+\rightarrow \R,$ then $k^{\delta}(\bfx, \bfx') = \ell^{-4}\eta(\|\bfx-\bfx'\|/\ell)$ for another function $\eta:\R^+\rightarrow \R$ that does not depend on $\ell$; see \Cref{sec:appendix_length_scale_relationship}.
By contrast, the velocity GP requires setting the length scales of its priors in tandem, and it is unclear how to control the length scales of the divergence and vorticity.

\looseness=-1 \textbf{Correlations between longitudinal and latitudinal current components.} Ocean flows have correlation between longitudinal and latitudinal velocities at single locations and across different locations.   For instance, within a vortex, the longitudinal velocity at six o'clock (relative to the center of the vortex) coincides with a zero latitudinal velocity at that same location, and also with a non-zero latitudinal velocity at three o'clock. Likewise, the occurrence of divergence at a given point induces a latitudinal velocity at six o'clock (with no longitudinal velocity), as well as a non-zero longitudinal velocity at three o'clock (with no latitudinal velocity).  By modeling the divergence and vorticity directly, the Helmholtz prior induces 
correlation between the longitudinal and latitudinal components,
which is absent in the velocity GP prior. 


\textbf{Equivariance to reference frame.}
We now show the Helmholtz GP is agnostic to the choice of reference frame defined by longitude and latitude, but the velocity GP is not.

\begin{restatable}{proposition}{propEquivariance}\label{prop:Equivariance}
Let $\mu_{F\!\mid\!D}(\Xtest, \Xtrain, \Ytrain)$ denote the Helmholtz GP posterior mean for training data $\Xtrain, \Ytrain$ and test coordinates $\Xtest$,
and let $R$ be an operator rotating coordinates and velocities about $(0,0)$.
Then 
\begin{equation}\label{eqn:helm_equivariance}
\mu_{F\!\mid \!D}(R\Xtest, R\Xtrain, R\Ytrain) \!=\! R\mu_{F\!\mid \! D}(\Xtest, \Xtrain, \Ytrain).
\end{equation}
\end{restatable}

\Cref{prop:Equivariance} formalizes that it is equivalent to either (1) rotate the data and then predict using the Helmholtz GP or (2) predict using the Helmholtz GP and rotate the prediction. The proof of \cref{prop:Equivariance} is given in \cref{app:equivariancehelmgp}.

The equivariance property in \cref{prop:Equivariance} need not hold in general for velocity GP priors. 

\begin{restatable}{proposition}{propEquivarianceVelGP}\label{prop:EquivarianceVelGP}
For isotropic component kernels and zero prior mean, the velocity GP is reference-frame equivariant if and only if the kernels for each component are equal.
\end{restatable}

See \cref{app:equivariancevelgp} for the proof. 
Intuitively, if the kernels are equal, 
both the prior and likelihood (and therefore the entire model) are isotropic, and so there is no special reference frame.
For intuition in the other direction, consider the following counterexample. Let $F^{(1)} \sim \GP(0,k^{(1)})$ for some non-identically zero isotropic $k^{(1)}$. And $F^{(2)} = 0$, a trivial isotropic prior.
Take any data $\Xtrain$, $\Ytrain$, $\Xtest$, and a positive (counterclockwise) $90^{\circ}$ rotation. 
Due to the trivial prior in the second coordinate, the posterior in the second coordinate has mean $\mu^{(2)}_{F\!\mid\!D}(\Xtest, \Xtrain, \Ytrain) = 0$.
If we rotate the data first, the posterior in the second coordinate is still zero, and generally the posterior in the first coordinate will be nontrivial.
But if we first compute the posterior and then rotate the mean, the posterior in the first coordinate will now be zero instead, and the posterior in the second coordinate will be nonzero.
Therefore, the equality in \cref{eqn:helm_equivariance} will not hold for this velocity GP.

\section{Experimental Results}\label{sec:experiments}

\begin{figure*}[t]
    \centering
    \includegraphics[trim={0 0.25cm 0 0.08cm},clip, draft=False, width=\textwidth]{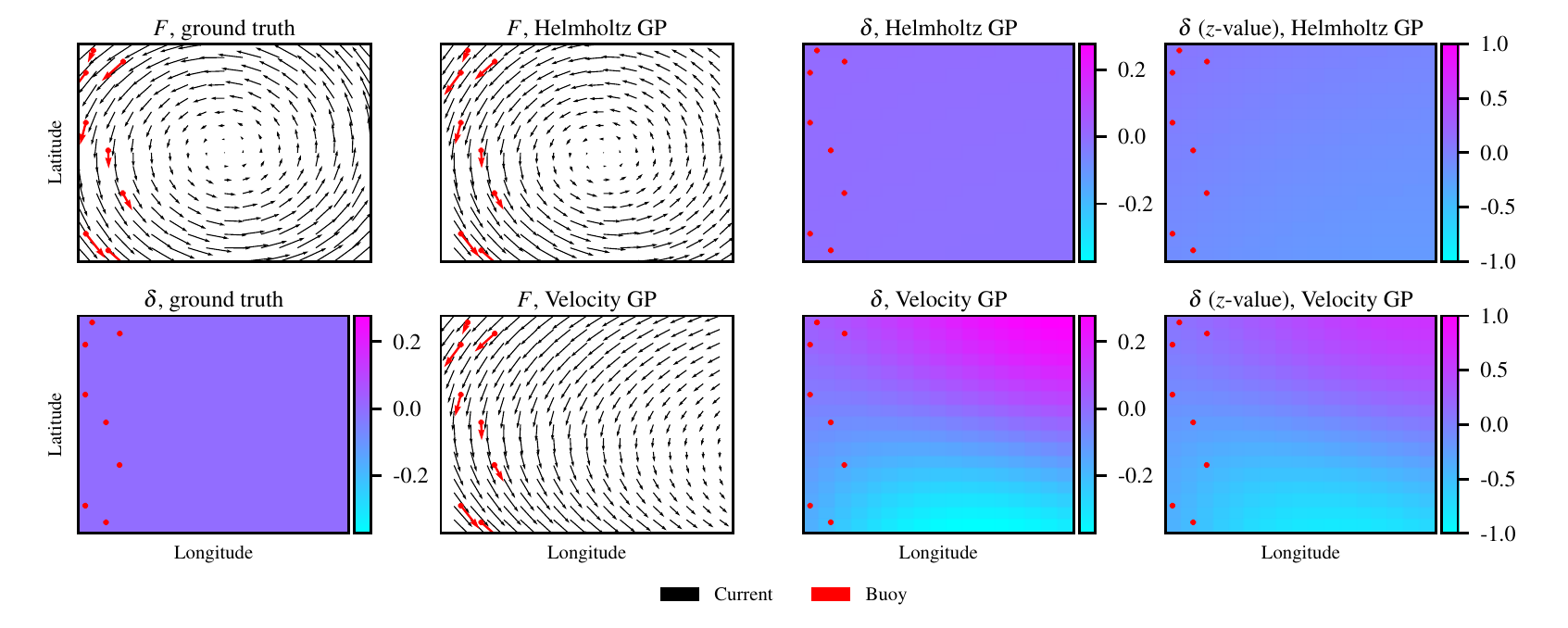}
    \caption{First column: ground truth predictions (upper) and divergence (lower). Second column: current predictions. Third column: divergence estimates. Fourth column: posterior divergence z-values.}
    \label{fig:vortex}
\end{figure*}

 We next empirically compare the $\sehelmgp$ and $\sevelgp$. We find that the $\sehelmgp$ yields better current predictions as well as better divergence (and vorticity) identification, across a variety of simulated and real data sets.
 We also compare to dissipative Hamiltonian neural networks (D-HNNs) \citep{greydanus2022dissipative} in \cref{tab:mse} and \cref{app:experiment} but find that the GP methods generally perform better.

\textbf{Data.} The real datasets we use consist of drifter traces of GPS-tagged buoys in the ocean. While oceanographers have some knowledge that allows a rough assessment of the real data, only in simulations do we have access to ground truth currents, divergences, and vorticities. Therefore, we run a variety of simulations with current vector fields reflecting known common ocean behaviors. We simulate buoy trajectories by initializing buoys at a starting point and allowing the current field to drive their motion. See \cref{app:experiment} for more details of our setup in each specific simulation.

\textbf{Performance.}
In what follows, we emphasize visual comparisons both because the distinctions between methods are generally clear and because it is illuminating to visually pick out behaviors of interest. We also provide root mean squared error (RMSE) comparisons in \cref{tab:mse}. However, we note that the RMSE can be expected to vary as one changes either the ocean area or the densities (or more generally locations) of test points, and both of these choices must always be somewhat arbitrary.



\textbf{Algorithmic details.} In our comparisons, each full model including an $\sehelmgp$ prior or an $\sevelgp$ prior has five hyperparameters: $\sigma^2_{\Phi}$, $\ell_{\Phi}$, $\sigma^2_{\Psi}$, $\ell_{\Psi}$, $\obsvar$ and $\ell_1, \sigmaFu^2, \ell_{2}, \sigmaFv^2, \obsvar$, respectively. In each case, we fit the log of the hyperparameters by maximizing the marginal likelihood using Adam \citep{kingma2015adam}. We optimize in the log-scale and then exponentiate the optimal values to ensure positivity of the hyperparameters. We run each experiment until the log marginal likelihood changes by less than $10^{-4}$, which occurs in fewer than 2000 iterations for all experiments. With the exception of the GLAD data (which presents special difficulties that we describe in \cref{app:experiments-glad}), we found that results were not sensitive to initialization. To train the D-HNN, we ran the code from \citet{greydanus2022dissipative}. More algorithmic details are provided individually for each experiment in \cref{app:experiment}.

\begin{figure*}[t]
    \centering
    \includegraphics[trim={0 0.24cm 0 0.06cm}, clip, width=\textwidth]{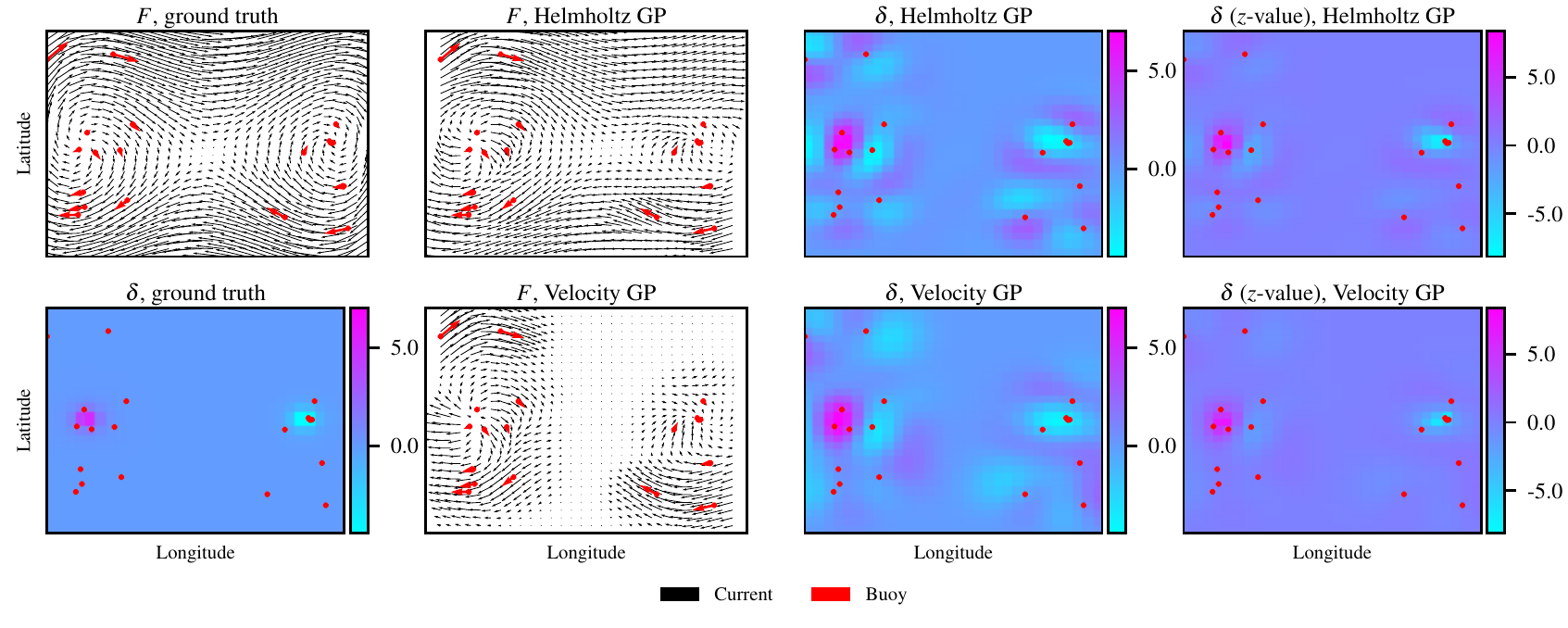}
    \caption{First column: ground truth predictions (upper) and divergence (lower). Second column: current predictions. Third column: divergence estimates. Fourth column: posterior divergence z-values.}
    \label{fig:duffing}
\end{figure*}

\subsection{Simulated experiments}

We focus on simulations of key ocean behaviors of interest to oceanographers: vortices, concentrated divergences, and combinations thereof. As a summary across simulation experiments in \cref{tab:mse}, we see that -- for predictions, divergence estimates, and vorticity estimates -- the $\sehelmgp$ is most accurate by RMSE on a majority of tasks, often by a substantial factor. We next examine individual experiments more closely; all simulated experiments are described in detail in \cref{app:sec-synthetic-experiments}.

\textbf{Vortex with zero divergence.} First, we consider a vortex with no divergence. The $\sehelmgp$ is better at predicting the current and identifying the lack of divergence.

In a vortex, water particles rotate around a central point. The black arrows in the upper left plot of \cref{fig:vortex} show the ground-truth vector field at test points, with longitude on the horizontal axis and latitude on the vertical axis. Red arrows show our simulated buoy trajectories, which we give to all methods as training data. See \cref{app:subsectionvortex} for additional details of the setup and results. 

The second column in \cref{fig:vortex} shows predictions from the $\sehelmgp$ (upper) and $\sevelgp$ (lower) at the test points. The red arrows are still the training data. Despite having access to data only from one side of the vortex, the $\sehelmgp$ is able to reconstruct the full vortex. The $\sevelgp$ is not. 

The ground truth divergence is identically 0 throughout the domain and depicted in the lower left plot. The third column shows divergence estimates from the $\sehelmgp$ (upper) and $\sevelgp$ (lower) on the same color scale. The fourth column helps us understand if either posterior is reporting a nonzero divergence. In particular, for each point we plot a ``z-value'': precisely, the posterior mean at that point divided by the posterior standard deviation. One might, for instance, conclude that a method has detected a nonzero divergence if the magnitude of the z-value is greater than 1.
From the third column, we conclude that the $\sehelmgp$ estimate of the divergence is closer to the ground truth of zero than the $\sevelgp$. From the fourth column, we see that neither method concludes nonzero divergence, but the $\sehelmgp$ posterior is more concentrated near zero.


\textbf{Duffing oscillator with areas of concentrated divergence.} We next simulate a classic example called a Duffing oscillator, and we add two areas of divergence; the ground truth current appears in the upper left plot of \cref{fig:duffing}, and the ground truth divergence appears in the lower left. The simulated buoy trajectories appear in red. See \cref{app:sec-duffing-divergence} for further details on setup and results.

We see in the second column that the $\sehelmgp$ (upper) is largely able to reconstruct the two vortices in the Duffing oscillator (upper left), though it struggles with the upper right current. By contrast, the $\sevelgp$ does not connect the currents continuously across the two sides of the space, in disagreement with conservation of momentum. 

Again, the third column depicts divergence estimates from both methods, and the fourth column depicts z-values. In this case, both methods accurately recover the two areas of divergence. In \cref{app:sec-duffing-divergence,fig:smallduffing,fig:midduffing,fig:bigduffing}, we experiment with smaller and larger areas of divergence with the Duffing oscillator. In \cref{app:sec-divergence}, we isolate areas of divergence without the Duffing oscillator. Across the six experiments involving regions of divergence, the $\sehelmgp$ typically outperforms the $\sevelgp$ in detecting these regions -- often by a substantial margin, as shown in \cref{tab:mse}. In these same experiments, the $\sehelmgp$ similarly outperforms the $\sevelgp$ at predicting the velocity field.


\textbf{A note on vorticity.} Although we have not focused on vorticity estimation in the main text, 
generally we find superior performance on this task from the $\sehelmgp$ relative to the $\sevelgp$ (and D-HNNs), similar to divergence estimation.
See the righthand side of \cref{tab:mse} for an RMSE comparison, and see \cref{app:experiment} for a visual comparison. 
For example, the $\sehelmgp$ can predict zero vorticity when there is no vorticity, whereas the $\sevelgp$ and D-HNN fail in this task (\cref{fig:smalldivapp,fig:middivapp,fig:bigdivapp}).

\begin{figure*}[t]
    \centering
    \includegraphics[trim={0 0.25cm 0 0.06cm},clip,width=\textwidth]{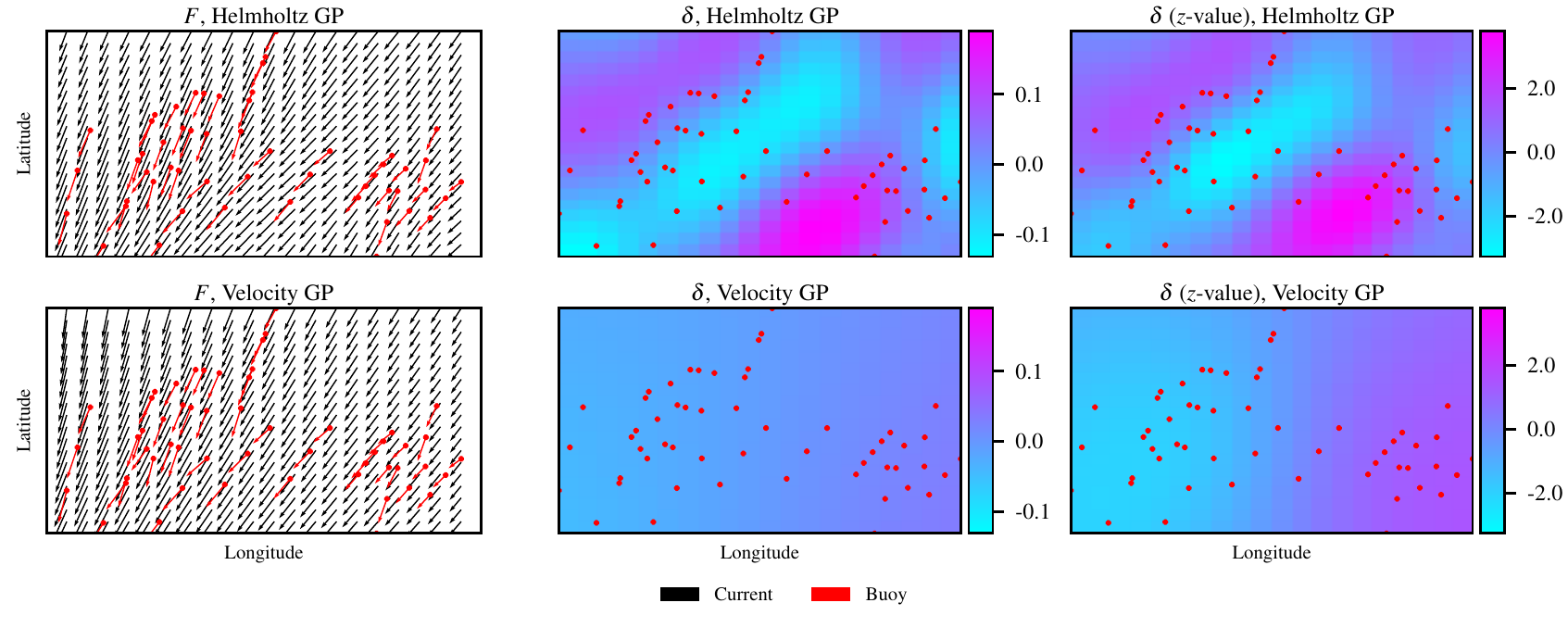}
    \caption{First column: current predictions. Second column: divergence estimates. Third column: posterior divergence z-values.}
    \label{fig:laser}
\end{figure*}

\subsection{Real-data experiments}

Although ground truth currents and divergences are not available for real data, we can still assess performance against oceanographers' expert knowledge.

\textbf{LASER data.} The LAgrangian Submesoscale ExpeRiment \citep{dasaroLASERdata2017} was performed in the Gulf of Mexico in 2016. The full dataset spans two winter months and more than 1000 buoys; see \cref{app:experiments-laser} for details. We focus on a particular spatial area and 2-hour time span in which oceanographers expect to see a particular convergent front. Time information within this range is discarded. We are left with 19 buoys. When recorded spatial coordinates overlap across data points, we observe that both the $\sehelmgp$ and the $\sevelgp$ treat all data as noise. While this issue merits further investigation and model development, for now we downsample across time to form the final data sets here.

The left column of \cref{fig:laser} shows the current predictions using the $\sehelmgp$ (upper) and $\sevelgp$ (lower). Red arrows show the observed buoy data, with 55 total observations across all buoys. The black arrows show the current posterior means at test locations. The two sets of predictions are qualitatively very similar.

The second column shows the divergence predictions for the $\sehelmgp$ (upper) and $\sevelgp$ (lower); the third column shows the z-values for the respective posterior distributions. The $\sehelmgp$ predicts a negative divergence area (light blue diagonal) that agrees with the area where oceanographers expect a convergent front. By contrast, the $\sevelgp$ does not identify any divergence.

We find that the discrepancy between the $\sehelmgp$ and the $\sevelgp$ observed in \cref{fig:laser} depends on the amount of data. When we double the amount of data from that in \cref{fig:laser} (by downsampling less), we find that both methods are able to recover the same convergent front; see \cref{app:experiments-laser,fig:laserfull}. This finding then also corroborates that the convergent front in fact does exist and looks like the posterior detection by the $\sehelmgp$ in \cref{fig:laser}. In that way, the finding further lends support to the superior performance of the $\sehelmgp$ in \cref{fig:laser}.

\textbf{GLAD data.} The Grand Lagrangian Deployment (GLAD) experiment \citep{ozgokmen2012carthe} was conducted  near the Deepwater Horizon site and Louisiana coast in July 2012. The full dataset consists of over 300 buoys. Unlike the winter LASER data, the summer GLAD data faces additional challenges from regular oscillations due to wind (rather than current); see \cref{app:experiments-glad} for more details. Rather than account for these oscillations, for the moment we ameliorate their effect by downsampling considerably both in time and across buoys. We are left with 85 observations.

The current predictions between the $\sehelmgp$ and $\sevelgp$ are generally quite similar (\cref{fig:gladsparse,fig:gladfull}). Although in this case we do not have any ground truth from oceanographers in advance, some features of the $\sehelmgp$ prediction seem more realistic from physical intuition: namely, the downturn in the lower left region of the plot and the vortex in the upper left region of the plot. We also note that the $\sehelmgp$ and $\sevelgp$ predict substantially different divergences: respectively, zero and nonzero.

To check that our $\sehelmgp$ method is computationally reasonable, we run on a larger subset of the GLAD data with 1200 data points. We find that it takes less than 10 minutes to optimize the hyperparameters, form predictions, and estimate the divergences at the test points.


\section{Discussion and Future Work}

We have demonstrated the conceptual and empirical advantages of our Helmholtz GP relative to the velocity GP. A number of challenges remain. While we have focused on the purely spatial case for modularity, in general we expect currents to change over time, and it remains to extend our method to the spatiotemporal case -- which we believe should be straightforward. Moreover, \citet{goncalves2019naive,lodise2020investigating} used more complex kernels with two length scales per dimension in their spatiotemporal extension of the velocity GP. From known ocean dynamics, we expect that in fact two length scales would be appropriate in the Helmholtz GP for modeling the vorticity -- but unnecessary for the divergence.
While our independent Gaussian noise model is standard in spatiotemporal modeling -- and shared by \citet{goncalves2019naive,lodise2020investigating} -- our real-data experimentation suggests that a better noise model might account for short-term wind oscillations and other noise patterns distinctive to oceans. 

\section*{Acknowledgements}

The authors are grateful to the Office of Naval Research for partial support under grant N00014-20-1-2023 (MURI ML-SCOPE). Renato Berlinghieri and Tamara Broderick were also supported in part by an NSF CAREER Award, and Tamay Özgökmen was supported in part by RSMAES, the University of Miami.

\bibliography{references}

\begin{thebibliography}{38}
\providecommand{\natexlab}[1]{#1}
\providecommand{\url}[1]{\texttt{#1}}
\expandafter\ifx\csname urlstyle\endcsname\relax
  \providecommand{\doi}[1]{doi: #1}\else
  \providecommand{\doi}{doi: \begingroup \urlstyle{rm}\Url}\fi

\bibitem[Adler(1981)]{adler1981geometry}
Adler, R.
\newblock \emph{The Geometry of Random Fields}.
\newblock Society for Industrial and Applied Mathematics, 1981.

\bibitem[Alvarez et~al.(2012)Alvarez, Rosasco, and
  Lawrence]{alvarez2012kernels}
Alvarez, M.~A., Rosasco, L., and Lawrence, N.~D.
\newblock Kernels for vector-valued functions: A review.
\newblock \emph{Foundations and Trends in Machine Learning}, 4\penalty0
  (3):\penalty0 195--266, 2012.

\bibitem[Arfken \& Weber(1999)Arfken and Weber]{arfken1999mathematical}
Arfken, G.~B. and Weber, H.~J.
\newblock Mathematical methods for physicists, 1999.

\bibitem[Berta et~al.(2015)Berta, Griffa, Magaldi, {\"O}zg{\"o}kmen, Poje,
  Haza, and Olascoaga]{berta2015improved}
Berta, M., Griffa, A., Magaldi, M.~G., {\"O}zg{\"o}kmen, T.~M., Poje, A.~C.,
  Haza, A.~C., and Olascoaga, M.~J.
\newblock Improved surface velocity and trajectory estimates in the {G}ulf of
  {M}exico from blended satellite altimetry and drifter data.
\newblock \emph{Journal of Atmospheric and Oceanic Technology}, 32\penalty0
  (10):\penalty0 1880--1901, 2015.

\bibitem[Bhatia et~al.(2013)Bhatia, Norgard, Pascucci, and
  Bremer]{bhatia2013helmoltzhodge}
Bhatia, H., Norgard, G., Pascucci, V., and Bremer, P.-T.
\newblock The {H}elmholtz-{H}odge decomposition—a survey.
\newblock \emph{IEEE Transactions on Visualization and Computer Graphics},
  19\penalty0 (8):\penalty0 1386--1404, 2013.

\bibitem[B{\"u}hler et~al.(2014)B{\"u}hler, Callies, and
  Ferrari]{buhler2014wave}
B{\"u}hler, O., Callies, J., and Ferrari, R.
\newblock Wave--vortex decomposition of one-dimensional ship-track data.
\newblock \emph{Journal of Fluid Mechanics}, 756:\penalty0 1007--1026, 2014.

\bibitem[Caballero et~al.(2020)Caballero, Mulet, Ayoub, Manso-Narvarte, Davila,
  Boone, Toublanc, and Rubio]{caballero2020integration}
Caballero, A., Mulet, S., Ayoub, N., Manso-Narvarte, I., Davila, X., Boone, C.,
  Toublanc, F., and Rubio, A.
\newblock Integration of {HF} radar observations for an enhanced coastal mean
  dynamic topography.
\newblock \emph{Frontiers in Marine Science}, pp.\  1005, 2020.

\bibitem[Chavanne \& Klein(2010)Chavanne and Klein]{chavanne2010can}
Chavanne, C.~P. and Klein, P.
\newblock Can oceanic submesoscale processes be observed with satellite
  altimetry?
\newblock \emph{Geophysical Research Letters}, 37\penalty0 (22), 2010.

\bibitem[Contreras et~al.(2022)Contreras, Renault, and
  Marchesiello]{contreras2022understanding}
Contreras, M., Renault, L., and Marchesiello, P.
\newblock Understanding energy pathways in the {Gulf Stream}.
\newblock \emph{Journal of Physical Oceanography}, 2022.

\bibitem[D'Asaro et~al.(2017)D'Asaro, Guigand, Haza, Huntley, Novelli,
  Özgökmen, and Ryan]{dasaroLASERdata2017}
D'Asaro, E., Guigand, C., Haza, A., Huntley, H., Novelli, G., Özgökmen, T.,
  and Ryan, E.
\newblock Lagrangian submesoscale experiment (laser) surface drifters,
  interpolated to 15-minute intervals, 2017.
\newblock URL
  \url{https://data.gulfresearchinitiative.org/data/R4.x265.237:0001}.

\bibitem[D’Asaro et~al.(2018)D’Asaro, Shcherbina, Klymak, Molemaker,
  Novelli, Guigand, Haza, Haus, Ryan, Jacobs, Huntley, Laxague, Chen, Judt,
  McWilliams, Barkan, Kirwan, Poje, and Özgökmen]{dasaro2018ocean}
D’Asaro, E.~A., Shcherbina, A.~Y., Klymak, J.~M., Molemaker, J., Novelli, G.,
  Guigand, C.~M., Haza, A.~C., Haus, B.~K., Ryan, E.~H., Jacobs, G.~A.,
  Huntley, H.~S., Laxague, N. J.~M., Chen, S., Judt, F., McWilliams, J.~C.,
  Barkan, R., Kirwan, A.~D., Poje, A.~C., and Özgökmen, T.~M.
\newblock Ocean convergence and the dispersion of flotsam.
\newblock \emph{Proceedings of the National Academy of Sciences}, 115\penalty0
  (6):\penalty0 1162--1167, 2018.

\bibitem[Fu \& Ferrari(2008)Fu and Ferrari]{fu2008observing}
Fu, L.-L. and Ferrari, R.
\newblock Observing oceanic submesoscale processes from space.
\newblock \emph{Eos, Transactions American Geophysical Union}, 89\penalty0
  (48):\penalty0 488--488, 2008.

\bibitem[Fu et~al.(2010)Fu, Chelton, Le~Traon, and Morrow]{fu2010eddy}
Fu, L.-L., Chelton, D.~B., Le~Traon, P.-Y., and Morrow, R.
\newblock Eddy dynamics from satellite altimetry.
\newblock \emph{Oceanography}, 23\penalty0 (4):\penalty0 14--25, 2010.

\bibitem[Fuselier(2007)]{fuselier2007refined}
Fuselier, Jr., E.~J.
\newblock \emph{Refined error estimates for matrix-valued radial basis
  functions}.
\newblock PhD thesis, Texas A\&M University, 2007.

\bibitem[Golub \& Van~Loan(2013)Golub and Van~Loan]{golub2013matrix}
Golub, G.~H. and Van~Loan, C.~F.
\newblock \emph{Matrix computations}.
\newblock JHU press, 2013.

\bibitem[Gonçalves et~al.(2019)Gonçalves, Iskandarani, Özgökmen, and
  Thacker]{goncalves2019naive}
Gonçalves, R.~C., Iskandarani, M., Özgökmen, T., and Thacker, W.~C.
\newblock Reconstruction of submesoscale velocity field from surface drifters.
\newblock \emph{Journal of Physical Oceanography}, 49\penalty0 (4), 2019.

\bibitem[Greydanus \& Sosanya(2022)Greydanus and
  Sosanya]{greydanus2022dissipative}
Greydanus, S. and Sosanya, A.
\newblock Dissipative {H}amiltonian neural networks: Learning dissipative and
  conservative dynamics separately.
\newblock \emph{arXiv preprint arXiv:2201.10085}, 2022.

\bibitem[Greydanus et~al.(2019)Greydanus, Dzamba, and
  Yosinski]{greydanus2019hamiltonian}
Greydanus, S., Dzamba, M., and Yosinski, J.
\newblock Hamiltonian neural networks.
\newblock \emph{Advances in Neural Information Processing Systems}, 32, 2019.

\bibitem[Han \& Huang(2020)Han and Huang]{lei2020helm2}
Han, L. and Huang, R.~X.
\newblock Using the {H}elmholtz decomposition to define the {Indian Ocean}
  meridional overturning streamfunction.
\newblock \emph{Journal of Physical Oceanography}, 50\penalty0 (3), 2020.

\bibitem[Haza et~al.(2018)Haza, D’Asaro, Chang, Chen, Curcic, Guigand,
  Huntley, Jacobs, Novelli, Özgökmen, Poje, Ryan, and
  Shcherbina]{haza2018laser}
Haza, A.~C., D’Asaro, E., Chang, H., Chen, S., Curcic, M., Guigand, C.,
  Huntley, H.~S., Jacobs, G., Novelli, G., Özgökmen, T.~M., Poje, A.~C.,
  Ryan, E., and Shcherbina, A.
\newblock Drogue-loss detection for surface drifters during the lagrangian
  submesoscale experiment ({LASER}).
\newblock \emph{Journal of Atmospheric and Oceanic Technology}, 35\penalty0
  (4), 2018.

\bibitem[Kingma \& Ba(2015)Kingma and Ba]{kingma2015adam}
Kingma, D.~P. and Ba, J.
\newblock Adam: A method for stochastic optimization.
\newblock In \emph{International Conference on Learning Representations
  (ICLR)}, 2015.

\bibitem[Lindgren(2012)]{lindgren2012stationary}
Lindgren, G.
\newblock \emph{Stationary stochastic processes: theory and applications}.
\newblock CRC Press, 2012.

\bibitem[Lodise et~al.(2020)Lodise, Özgökmen, Gonçalves, Iskandarani, Lund,
  Horstmann, Poulain, Klymak, Ryan, and Guigand]{lodise2020investigating}
Lodise, J., Özgökmen, T., Gonçalves, R.~C., Iskandarani, M., Lund, B.,
  Horstmann, J., Poulain, P.-M., Klymak, J., Ryan, E.~H., and Guigand, C.
\newblock Investigating the formation of submesoscale structures along
  mesoscale fronts and estimating kinematic quantities using {L}agrangian
  drifters.
\newblock \emph{Fluids}, 5\penalty0 (3), 2020.

\bibitem[Lowitzsch(2002)]{lowitzsch2002approximation}
Lowitzsch, S.
\newblock \emph{Approximation and interpolation employing divergence-free
  radial basis functions with applications}.
\newblock Texas A\&M University, 2002.

\bibitem[Mac{\^e}do \& Castro(2010)Mac{\^e}do and Castro]{macedo2010learning}
Mac{\^e}do, I. and Castro, R.
\newblock Learning divergence-free and curl-free vector fields with
  matrix-valued kernels.
\newblock Technical report, Instituto Nacional de Matemática Pura e Aplicada,
  2010.

\bibitem[Mariano et~al.(2016)Mariano, Ryan, Huntley, Laurindo, Coelho, Griffa,
  Özgökmen, Berta, Bogucki, Chen, Curcic, Drouin, Gough, Haus, Haza, Hogan,
  Iskandarani, Jacobs, Kirwan~Jr., Laxague, Lipphardt~Jr., Magaldi, Novelli,
  Reniers, Restrepo, Smith, Valle-Levinson, and Wei]{mariano2016statistical}
Mariano, A.~J., Ryan, E.~H., Huntley, H.~S., Laurindo, L., Coelho, E., Griffa,
  A., Özgökmen, T.~M., Berta, M., Bogucki, D., Chen, S.~S., Curcic, M.,
  Drouin, K., Gough, M., Haus, B.~K., Haza, A.~C., Hogan, P., Iskandarani, M.,
  Jacobs, G., Kirwan~Jr., A.~D., Laxague, N., Lipphardt~Jr., B., Magaldi,
  M.~G., Novelli, G., Reniers, A., Restrepo, J.~M., Smith, C., Valle-Levinson,
  A., and Wei, M.
\newblock Statistical properties of the surface velocity field in the northern
  {G}ulf of {M}exico sampled by {GLAD} drifters.
\newblock \emph{Journal of Geophysical Research: Oceans}, 121\penalty0
  (7):\penalty0 5193--5216, 2016.

\bibitem[Novelli et~al.(2017)Novelli, Guigand, Cousin, Ryan, Laxague, Dai,
  Haus, and {\"O}zg{\"o}kmen]{novelli2017biodegradable}
Novelli, G., Guigand, C.~M., Cousin, C., Ryan, E.~H., Laxague, N.~J., Dai, H.,
  Haus, B.~K., and {\"O}zg{\"o}kmen, T.~M.
\newblock A biodegradable surface drifter for ocean sampling on a massive
  scale.
\newblock \emph{Journal of Atmospheric and Oceanic Technology}, 34\penalty0
  (11):\penalty0 2509--2532, 2017.

\bibitem[{\"O}zg{\"o}kmen(2012)]{ozgokmen2012carthe}
{\"O}zg{\"o}kmen, T.
\newblock {CARTHE: GLAD} experiment {CODE}-style drifter trajectories (lowpass
  filtered, 15 minute interval records), northern {Gulf of Mexico} near {DeSoto
  Canyon}, {J}uly-{O}ctober 2012.
\newblock \emph{Gulf of Mexico Research Initiative}, 10:\penalty0 N7VD6WC8,
  2012.

\bibitem[Poje et~al.(2014)Poje, Özgökmen, Lipphardt, Haus, Ryan, Haza,
  Jacobs, Reniers, Olascoaga, Novelli, Griffa, Beron-Vera, Chen, Coelho, Hogan,
  Kirwan, Huntley, and Mariano]{poje2014submesoscale}
Poje, A.~C., Özgökmen, T.~M., Lipphardt, B.~L., Haus, B.~K., Ryan, E.~H.,
  Haza, A.~C., Jacobs, G.~A., Reniers, A. J. H.~M., Olascoaga, M.~J., Novelli,
  G., Griffa, A., Beron-Vera, F.~J., Chen, S.~S., Coelho, E., Hogan, P.~J.,
  Kirwan, A.~D., Huntley, H.~S., and Mariano, A.~J.
\newblock Submesoscale dispersion in the vicinity of the {Deepwater Horizon}
  spill.
\newblock \emph{Proceedings of the National Academy of Sciences}, 111\penalty0
  (35):\penalty0 12693--12698, 2014.

\bibitem[Rasmussen \& Williams(2005)Rasmussen and Williams]{rasmussen2005gp}
Rasmussen, C.~E. and Williams, C. K.~I.
\newblock \emph{Gaussian processes for machine learning}.
\newblock MIT Press, 2005.

\bibitem[Rocha et~al.(2016)Rocha, Chereskin, Gille, and
  Menemenlis]{rocha2016mesoscale}
Rocha, C.~B., Chereskin, T.~K., Gille, S.~T., and Menemenlis, D.
\newblock Mesoscale to submesoscale wavenumber spectra in {Drake Passage}.
\newblock \emph{Journal of Physical Oceanography}, 46\penalty0 (2):\penalty0
  601--620, 2016.

\bibitem[Solin et~al.(2018)Solin, Kok, Wahlstr{\"o}m, Sch{\"o}n, and
  S{\"a}rkk{\"a}]{solin2018modeling}
Solin, A., Kok, M., Wahlstr{\"o}m, N., Sch{\"o}n, T.~B., and S{\"a}rkk{\"a}, S.
\newblock Modeling and interpolation of the ambient magnetic field by
  {G}aussian processes.
\newblock \emph{IEEE Transactions on robotics}, 34\penalty0 (4):\penalty0
  1112--1127, 2018.

\bibitem[Srinivasan et~al.(2023)Srinivasan, Barkan, and
  McWilliams]{srinivasan2023forward}
Srinivasan, K., Barkan, R., and McWilliams, J.~C.
\newblock A forward energy flux at submesoscales driven by frontogenesis.
\newblock \emph{Journal of Physical Oceanography}, 53\penalty0 (1):\penalty0
  287--305, 2023.

\bibitem[Wahlstr{\"o}m(2015)]{wahlstrom2015modeling}
Wahlstr{\"o}m, N.
\newblock \emph{Modeling of magnetic fields and extended objects for
  localization applications}.
\newblock PhD thesis, Link{\"o}ping University Electronic Press, 2015.

\bibitem[Wahlstr{\"o}m et~al.(2013)Wahlstr{\"o}m, Kok, Sch{\"o}n, and
  Gustafsson]{wahlstrom2013modeling}
Wahlstr{\"o}m, N., Kok, M., Sch{\"o}n, T.~B., and Gustafsson, F.
\newblock Modeling magnetic fields using {Gaussian processes}.
\newblock In \emph{2013 IEEE International Conference on Acoustics, Speech and
  Signal Processing}, pp.\  3522--3526. IEEE, 2013.

\bibitem[Yaremchuk \& Coelho(2014)Yaremchuk and Coelho]{yaremchuk2014filtering}
Yaremchuk, M. and Coelho, E.~F.
\newblock Filtering drifter trajectories sampled at submesoscale resolution.
\newblock \emph{IEEE Journal of Oceanic Engineering}, 40\penalty0 (3):\penalty0
  497--505, 2014.

\bibitem[Zhang et~al.(2018)Zhang, Wei, Liu, and Fu]{zhang2018helm1}
Zhang, C., Wei, H., Liu, Z., and Fu, X.
\newblock Characteristic ocean flow visualization using {H}elmholtz
  decomposition.
\newblock In \emph{2018 Oceans-MTS/IEEE Kobe Techno-Oceans (OTO)}, pp.\  1--4.
  IEEE, 2018.

\bibitem[Zhang et~al.(2019)Zhang, Wei, Bi, and Liu]{zhang2019helm3}
Zhang, C., Wei, H., Bi, C., and Liu, Z.
\newblock Helmholtz–{H}odge decomposition-based {2D} and {3D} ocean surface
  current visualization for mesoscale eddy detection.
\newblock \emph{Journal of Visualization}, 22, 01 2019.

\end{thebibliography}
\bibliographystyle{icml2023}

\appendix

\newpage

\renewcommand{\partname}{}
\renewcommand{\thepart}{}

\appendix
\onecolumn
\addcontentsline{toc}{section}{Appendix} 
\part{Appendix} 
\parttoc 

\section{Related Work}
\label{app:related_work}

In what follows, we present related work in more detail. 
We first describe in more detail the differences between our work and that of \citet{goncalves2019naive,lodise2020investigating}. Then we discuss why we chose to put our priors on the Helmholtz decomposition, rather than an alternative decomposition.

\paragraph{Velocity GP vs.\ the GP of \citet{goncalves2019naive,lodise2020investigating}.} \citet{goncalves2019naive,lodise2020investigating} used many of the components of the velocity GP that we describe in the main text, but their prior was substantially more complex than the velocity GP. Like the velocity GP, they focused on a GP prior with a squared exponential covariance function. Unlike the velocity GP as described in the main text here, their squared exponential prior included not only terms in each of the longitude and latitude directions, but also a term in the time direction. Each term has its own length scale. As a second principle difference, their covariance was in fact a sum of two such squared exponential kernels -- introducing a total of 6 length scales (one for longitude, latitude, and time in each of the two kernels), 2 signal variances, and a single noise variance. They mention also trying 3 kernels (instead of 2), but it appears all their results were reported for 2 kernels. We have here tried to take a modular approach to examine the squared exponential prior on its own, so our velocity GP should not be seen as a direct reflection of the performance of the \citet{goncalves2019naive,lodise2020investigating} covariance function.

\paragraph{Why we focused on the Helmholtz decomposition.} The Helmholtz decomposition is a widely recognized dynamically significant method for dissecting the oceanic velocity field. An alternative -- albeit related -- decomposition that sees frequent use is the Geostrophic-Ageostrophic (G-Ag hereinafter) decomposition \citep{vallis2017}, which relies on the dominance of geostrophic balance at large spatial and time scales in the ocean. Though not directly related to our current work here, we discuss it briefly to provide
a holistic oceanographic context to our choice of the Helmholtz decomposition. The Helmholtz decomposition is defined through exact linear operators into velocity components that are orthogonal complements and can be separated easily, allowing priors to be placed on the underlying potentials; the G-Ag decomposition, however, can  be derived only by first eliminating the Ageostrophic flow (which represents faster, smaller scales) through an ad hoc time smoothing of drifter velocities using a multi-day filter. The G-Ag components are notably not orthogonal complements and consequently have to be separately estimated through a velocity GP, leading to a more complex modeling pipeline with additional physical hyperparameters (like the smoothing time) that are not easily determined. While there is a measure of correspondence between the geostrophic and rotational components, and the ageostrophic and divergent components,  the lack of precision in defining the G-Ag components makes the Helmholtz a natural modeling pathway. Recent oceanographic studies \citep{barkan2019, srinivasan2023forward} showing that the Helmholtz decomposition is directly relevant to the dynamics of oceanic components at smaller spatial scales of around O(1 km) offer further justification for our present choice.

\section{Divergence, Gradient, and Curl Operators in 2D}
\label{app:divcurl}

In this section we provide some background for the Helmholtz decomposition in 2D. In the first part, we provide definitions for $\mathrm{grad}, \mathrm{div}$, $\mathrm{curl}$, and $\mathrm{rot}$ operators. In \cref{prop:div-curl-free} we then characterize a property of vector fields obtained combining these operators.

Consider a scalar-valued differentiable function $f : \mathbb{R}^2 \rightarrow \mathbb{R}$. The \textit{gradient} of $f$ is the vector-valued function $\nabla f$ whose value at point $\textbf{x}$ is the vector whose components are the partial derivatives of $f$ at $\textbf{x}$. Formally, 
\begin{equation*}
    \mathrm{grad}\,f(\bfx) := 
\begin{bmatrix}
\frac{\partial f(\textbf{x})}{\partial x^{(1)}} \\ \frac{\partial f(\textbf{x})}{\partial x^{(2)}} 
\end{bmatrix}
= \bi \frac{\partial f(\textbf{x})}{\partial x^{(1)}} + 
\bj \frac{\partial f(\textbf{x})}{\partial x^{(2)}} 
\end{equation*}
where $\bi$ and $\bj$  are the standard unit vectors in the direction of the $x^{(1)}$ and $x^{(2)}$ coordinates. From this rewriting, one can note that taking the gradient of a function is equivalent to taking a vector operator $\nabla$, called \textit{del}:
$$
\nabla = 
\bi \frac{\partial}{\partial x^{(1)}} + \bj \frac{\partial}{\partial x^{(2)}} \equiv 
\bigg(\frac{\partial}{\partial x^{(1)}}, \frac{\partial}{\partial x^{(2)}}\bigg)
$$

Using this operator, two operations on vector fields can be defined. 

\begin{definition}
\normalfont Let $A \subset \mathbb{R}^2$ be an open subset and let $F : A \rightarrow \mathbb{R}^2$ be a vector field. The \textit{divergence} of $F$ is the scalar function $\normalfont \text{div} F : A \rightarrow \mathbb{R}$, defined by 
$$
\normalfont \text{div}\,F(\textbf{x}) := (\nabla \cdot F)(\textbf{x}) = \frac{\partial F^{(1)}}{\partial x^{(1)}} + \frac{\partial F^{(2)}}{\partial x^{(2)}}
$$
\end{definition}

\begin{definition}
\normalfont Let $A \subset \mathbb{R}^2$ be an open subset and let $F : A \rightarrow \mathbb{R}^2$ be a vector field. The \textit{curl} of $F$ is the scalar function $\normalfont \text{curl} F : A \rightarrow \mathbb{R}$, defined by 
$$
\normalfont \text{curl} F(\textbf{x}) := 
\frac{\partial F^{(1)}}{\partial x^{(2)}} - \frac{\partial F^{(2)}}{\partial x^{(1)}}
$$
\end{definition}

Note that this curl definition follows directly from the definition of curl in three dimensions, where this quantity describes infinitesimal circulation. 

In the 3D world, curl and divergence are enough to characterize the Helmholtz decomposition. For the 2D version, however, we need to characterize an additional operator - which we call \textit{rot} operator - that plays the role of the standard curl operator in the 3D version. In 2D, the rot formally requires the introduction of a third unit vector, $\bk$ that is orthogonal to the plane containing, $\bi$ and $\bj$.

\begin{definition}
\normalfont
Let $f : \mathbb{R}^2 \rightarrow \mathbb{R}$ be a scalar field. The \textit{rot} of $f$ is the vector field $\bk\times\nabla f,$ defined by
\begin{equation*}
    \mathrm{rot}\,f(\textbf{x})\equiv \bk\times\nabla f = 
\begin{bmatrix}
\frac{\partial f}{\partial x^{(2)}} \\ \frac{- \partial f}{\partial x^{(1)}} 
\end{bmatrix}
= \bi \frac{\partial f}{\partial x^{(2)}} - 
\bj \frac{\partial f}{\partial x^{(1)}} 
\end{equation*}
where $\bi$ and $\bj$ represents, respectively, the standard unit vectors in the direction of the $x^{(1)}$ and $x^{(2)}$ coordinates; $\bk$ is the unit vector orthogonal to the plane containing $\bi$ and $\bj$ satisfying the identities, $\bk\times\bj = -\bi$ and $\bk\times\bi = \bj$.
\end{definition}

Thus the rot operator can be thought of as a $\pi/2$ rotation of the grad operator. The precise reason why we need the introduction of a separate rot operator in 2D is because of a hidden peculiarity that the stream function, $\Psi$ is actually
the only non-zero component of a 3D vector potential field, $\boldsymbol{A}(\textbf{x})$, but that non-zero component is along the $\bk$ direction, $\boldsymbol{A}\equiv(0, 0, \Psi(\textbf{x}))$; equivalently $\boldsymbol{A}=\Psi\bk$. Given this observation, it can be shown that $\nabla_{3D}\times\boldsymbol{A} = \bk\times\nabla\Psi$, where $\nabla_{3D}$ is the direct 3D extension of the 2D $\nabla$ operator defined above.
The ideas of gradient, divergence, rot, and curl lead to the following characterization of vector fields. 

\begin{definition}
\normalfont A vector field $F: A \rightarrow \mathbb{R}^2$ is called \textit{rotation-free} (or curl-free) if the curl is zero, $\normalfont{\text{curl}} F = 0$, and it is called \textit{incompressible} (or divergence-free) if the divergence is zero, $\normalfont{\text{div}} F = 0$.
\end{definition}

\begin{proposition}\label{prop:div-curl-free}
Let $f$ be a scalar field and $\mathcal{C}^2$ the class of functions whose second derivatives exist and are continuous. 
\begin{enumerate}
    \item If $f$ is $\mathcal{C}^2$, then $\normalfont{\textrm{curl}}(\mathrm{grad}\,f) = 0$. Every gradient of a scalar field is rotation free.
    \item If $f$ is $\mathcal{C}^2$, then $\normalfont{\textrm{div}}(\mathrm{rot}\,f) = 0$. Every $\mathrm{rot}$ transformation of a scalar field is incompressible.
\end{enumerate}
\end{proposition}

\begin{proof}
For (1), we have the following:
$$
\mathrm{curl}(\mathrm{grad}\,f) = \text{curl}
\begin{bmatrix}
\frac{\partial f(\textbf{x})}{\partial x^{(1)}} \\ \frac{\partial f(\textbf{x})}{\partial x^{(2)}} 
\end{bmatrix}
= 
\frac{\partial f (\textbf{x})/ \partial x^{(1)}}{\partial x^{(2)}} - \frac{\partial f (\textbf{x})/ \partial x^{(2)}}{\partial x^{(1)}} = 0.
$$

For (2):
$$
\text{div}(\mathrm{rot}\, f) = \text{div}
\begin{bmatrix}
\frac{\partial f(\textbf{x})}{\partial x^{(2)}} \\ \frac{- \partial f(\textbf{x})}{\partial x^{(1)}} 
\end{bmatrix}
= 
\frac{\partial f(\textbf{x}) / \partial x^{(2)}}{\partial x^{(1)}} + \frac{- \partial f(\textbf{x}) / \partial x^{(1)}}{\partial x^{(2)}} = 0 .
$$
\end{proof}

For more material on vector calculus, we refer the reader to \citet{arfken1999mathematical}.

\section{Helmholtz Decomposition in the Ocean}
\label{app:helmholtz}

In what follows we relate the Helmholtz decomposition to ocean currents. In the first part, we provide intuition of how divergence and vorticity are significant in the context of oceanography. Next, in \cref{fig:helmholtz}, we present a visual representation of the Helmholtz decomposition and highlight the relevant aspects.

The divergence and vorticity of the ocean flow are relevant for oceanographic studies. Divergence characterizes fronts -- small structures with spatial scales on the order of $0.1$-$10$ km and temporal scales on the order of $1$-$100$h. These are associated with strong vertical motions comprised of a narrow and intense downwelling (flow into the ocean from the surface) and broad, diffuse upwelling (flow from depths to the surface). The strong downwelling regions play a crucial role in air-sea fluxes (including uptake of gases into the ocean) and for biological productivity, since floating particles
in the ocean (that include plankton and algae) are concentrated at these fronts. On the other hand, vorticity characterizes eddies, larger structures that usually evolve over a long timescale. These account for kinetic energy in the ocean, which makes them a crucial part of global balances of energy, momentum, heat, salt, and chemical constituents (such as carbon dioxide).

In \cref{fig:helmholtz} we provide visual intuition on how the Helmholtz theorem decomposes a vector field (ocean flow) into a divergent velocity field and a rotation velocity field. In this plot, one can see that from the divergence we can read areas of downwelling/sink (arrows pointing inwards to a single point) and upwelling/source (arrows pointing outwards from a single point). The vorticity, instead, characterizes rotational elements of the vector field, e.g., vortices/eddies in our ocean setting.

\begin{figure}[!h]
    \centering
    \includegraphics[scale=0.24]{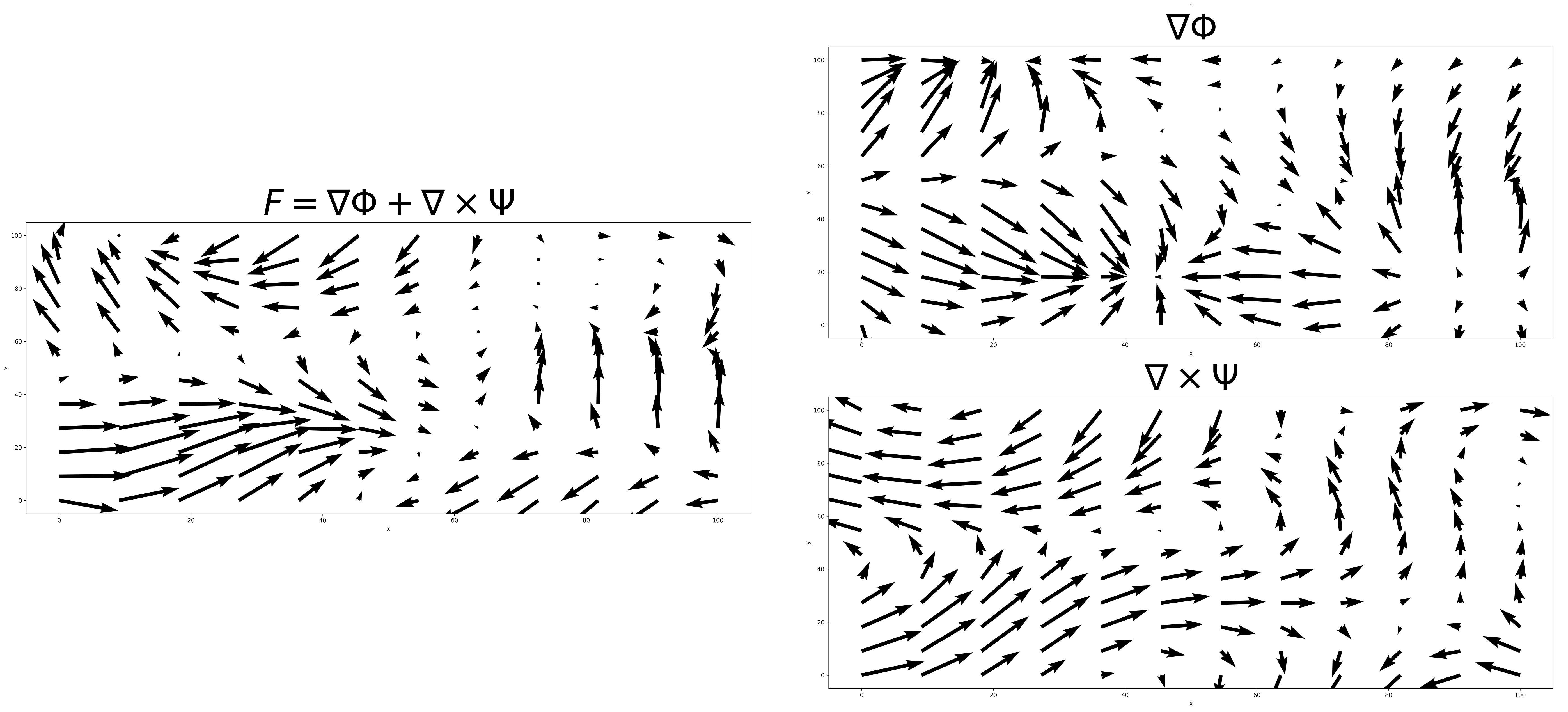}
    \caption{Helmholtz decomposition of vector field $F$. \textit{Left}: original vector field $F$. \textit{Top-right}: divergence component of Helmholtz decomposition of $F$. \textit{Bottom-right}: vorticity component of Helmholtz decomposition of $F$.}
    \label{fig:helmholtz}
\end{figure}

\section{Optimal Interpolation vs. Gaussian Processes}

Optimal interpolation is a powerful and widely used technique for the analysis of atmospheric data. The technique is only optimal under strong assumptions, and is therefore instead often referred to as \textit{statistical interpolation} \citep{daley1993atmospheric}. We formulate the statistical interpolation algorithm following \citet[Section 4.2]{daley1993atmospheric}. We are interested in modeling a variable $y = f(\boldsymbol{x})$, where $\bf x$ represent the 2D spatial coordinates of a point. For simplicity, assume $f : \mathbb{R}^2 \rightarrow \mathbb{R}$ is a univariate function. The multivariate statistical interpolation algorithm is very similar. See Chapter 5 in \citet{daley1993atmospheric}.

We have some locations $\boldsymbol{x}_k, k\in{1,\ldots,K}$, where we observe the variable of interest directly, $f_O(\boldsymbol{x}_k)$, and some other locations $\boldsymbol{x}_i, i\in{1,\ldots,I}$, where we want to infer (\textit{interpolate}) the value of the function of interest, $f_A(\boldsymbol{x}_i)$. Moreover, there might be some background knowledge about this function (e.g. a climate model providing a first guess), that we denote by 
$f_B(\boldsymbol{x}_k)$ and $f_B(\boldsymbol{x}_i)$. We assume a linear relationship between the observations and the function we want to model (note that this does not imply the modeled function is linear in the spatial locations). Then 
\begin{equation}
    f_A(\boldsymbol{x}_i) = f_B(\boldsymbol{x}_i) + \sum_{i=1}^K W_{ik}(f_O(\boldsymbol{x}_k) - f_B(\boldsymbol{x}_k))
\end{equation}
where $K$ is the number of locations where we observe the variable of interest, and $W_{ik}$ is the weight given to the residual $f_O(\boldsymbol{x}_k) - f_B(\boldsymbol{x}_k)$ when interpolating at $\boldsymbol{x}_i$. Finally, let the true value of the function at points $\boldsymbol{x}_i$ and $\boldsymbol{x}_k$ be $f_T(\boldsymbol{x}_i)$ and $f_T(\boldsymbol{x}_k)$, respectively. 

Define $\bar{\boldsymbol{W}}_i$ to be the column vector of length $K$ of weights, $\bar{\boldsymbol{B}}_i$ the column vector of length $K$ whose elements are $\E[(f_B(\boldsymbol{x}_k)-f_T(\boldsymbol{x}_k))(f_B(\boldsymbol{x}_i)-f_T(\boldsymbol{x}_i))]$, and $\boldsymbol{B}$ and $\boldsymbol{O}$ the background and observation error covariance matrices, respectively. That is, for $k, l \in \{1,\ldots,K\}$,
\begin{align*}
    \boldsymbol{B}_{k,l} = \E [(f_B(\boldsymbol{x}_k)-f_T(\boldsymbol{x}_k))(f_B(\boldsymbol{x}_l)-f_T(\boldsymbol{x}_l))], \\
    \boldsymbol{O}_{k,l} = \E [(f_O(\boldsymbol{x}_k)-f_T(\boldsymbol{x}_k))(f_O(\boldsymbol{x}_l)-f_T(\boldsymbol{x}_l))].
\end{align*}

Then, assuming that (i) the background and the observations are unbiased and (ii) that the correlation between background and observation error is zero, with some algebra it can be shown that the optimal weights satisfy 
$$
[\boldsymbol{B} + \boldsymbol{O}] \bar{\boldsymbol{W}}_i = \bar{\boldsymbol{B}}_i
$$
We refer the reader to Section 4.2 in \citet{daley1993atmospheric} for details on this derivation. This is the same as applying a least square approach on the residuals obtained by subtracting the background model estimates from the observations. 

Now assume that we do not have any background knowledge ($f_B(\boldsymbol{x}_i) = f_B(\boldsymbol{x}_k) = 0$) and we do not know the ground truth. Instead, we model it with a Gaussian process with zero mean and covariance function $\Tilde{K}$, $f_T \sim \GP (0, \Tilde{K})$. Assuming that the model is well-specified and the observation noise is $\sigma^2_{\text{obs}}$, we have
\begin{align*}
    & \bar{\boldsymbol{B}}_{i,k} = \E (f_T(\boldsymbol{x}_i), f_T(\boldsymbol{x}_k)) = \Tilde{K}(\boldsymbol{x}_i, \boldsymbol{x}_k) \\
    & \boldsymbol{B}_{k,l} = \E (f_T(\boldsymbol{x}_k), f_T(\boldsymbol{x}_l)) = \Tilde{K}(\boldsymbol{x}, \boldsymbol{x})_{k,l} \\
    & \boldsymbol{O}_{k,l} = \E (f_O(\boldsymbol{x}_k) - f_T(\boldsymbol{x}_k), f_O(\boldsymbol{x}_l) - f_T(\boldsymbol{x}_l)) = \sigma^2_{\text{obs}} \\
\end{align*}
where by $\Tilde{K}(\boldsymbol{x}, \boldsymbol{x})$ we mean the $I \times I$ matrix formed by evaluating the $\tilde{K}$ on all pairs of observed locations. So the optimal weights are
$$
\bar{\boldsymbol{W}}_{\ell,k} = \sum_{i=1}^I[\Tilde{K}(\boldsymbol{x}, \boldsymbol{x}) + \sigma^2_{\text{obs}}\mathbf{I}]^{-1}_{\ell, i} \Tilde{K}(\boldsymbol{x}_i, \boldsymbol{x}_k),
$$
These weights lead to exactly the same predictions as when doing standard Gaussian process regression. Therefore, we can see Gaussian process regression as a specific case of statistical interpolation, where we do not include any background knowledge (corresponding to a zero mean prior), and we do not have access to the ground truth function. We then choose a covariance function that encodes physical intuition of how the underlying system behaves, and estimate parameters from data (for example, by maximum likelihood). \citet[Section 4.3]{daley1993atmospheric} provides a detailed discussion of the choice of covariance functions in the oceanographic and atmospheric literature. 

\section{Helmholtz Gaussian Process Prior}
\label{app:prior-helm-gp}

In this section, we state and prove \cref{prop:helm-prior} from the main text.

\helmprior*

\begin{proof}
We obtain the result in two steps.
First, we argue that under the assumptions of the proposition, 
$F$ is distributed as a Gaussian process and so may be characterized through its mean and covariance function.
Second, we show $F$ has mean zero,
and the proposed covariance kernel.

To see that $F$ is a Gaussian process, observe that it is the sum of linear transformations of two independent Gaussian processes.
This follows from the fact that $\mathrm{grad}$ and $\mathrm{rot}$ are linear operators
on any vector space of differentiable functions, and because $k_\Phi$ and $k_\Psi$ are chosen to have almost surely continuously differentiable sample paths. Therefore, $\mathrm{grad} \Phi$ and $\mathrm{rot} \Psi$ are two independent GPs, and so $F$ is a Gaussian process as well. 

We next turn to the mean and covariance functions.
By linearity of expectation,
\begin{align*}
    \E[F]
    &= \E \left[\mathrm{grad} \,\Phi\right] + \E\left[\mathrm{curl} \,\Psi\right]\\
    &= \mathrm{grad} \,\E \Phi + \mathrm{curl} \,\E\Psi\\
    &= 0,
\end{align*}
where the last line follows from the assumption that $\Phi$ and $\Psi$ both have mean 0 everywhere.
It remains to calculate the covariance function.
Since $\Phi$ and $\Psi$ are assumed independent 
we compute the covariance as the sum of covariances for 
$\mathrm{grad} \,\Phi$ and $\mathrm{curl} \,\Psi.$
Consider two points $\bfx$ and $\bfx'.$
\begin{align*}
\Cov \left[(\mathrm{grad} \,\Phi)(\bfx),
(\mathrm{grad} \,\Phi)(\bfx')\right] 
 &=
    \Cov                 
    \left[\begin{pmatrix}
    \frac{\partial\Phi(\bfx)}{\partial x^{(1)}} \\
    \frac{\partial\Phi(\bfx)}{\partial x^{(2)}}
    \end{pmatrix}
    , 
    \begin{pmatrix}
    \frac{\partial\Phi(\bfx')}{\partial (x')^{(1)}} \\
    \frac{\partial\Phi(\bfx')}{\partial (x')^{(2)}} 
    \end{pmatrix} 
    \right] \\
&= \begin{bmatrix}
    \Cov\left(\frac{\partial\Phi(\bfx)}{\partial x^{(1)}}, \frac{\partial\Phi(\bfx')}{\partial (x')^{(1)}}\right) 
    & 
    \Cov\left(\frac{\partial\Phi(\bfx)}{\partial x^{(1)}}, \frac{\partial\Phi(\bfx')}{\partial (x')^{(2)}}\right) 
    \\
    \Cov\left(\frac{\partial\Phi(\bfx)}{\partial x^{(2)}}, \frac{\partial\Phi(\bfx')}{\partial (x')^{(1)}}\right) 
    & 
    \Cov\left(\frac{\partial\Phi(\bfx)}{\partial x^{(2)}}, \frac{\partial\Phi(\bfx')}{\partial (x')^{(2)}}\right) 
    \end{bmatrix} \\
&= \begin{bmatrix}
    \dfrac{\partial^2 k_{\Phi}(\bfx, \bfx')}{\partial x^{(1)} \partial (x')^{{(1)}}} & \dfrac{\partial^2 k_{\Phi}(\bfx, \bfx')}{\partial x^{(1)} \partial (x')^{{(2)}}}\\
    \dfrac{\partial^2 k_{\Phi}(\bfx, \bfx')}{\partial x^{(2)} \partial (x')^{{(1)}}} & \dfrac{\partial^2 k_{\Phi}(\bfx, \bfx')}{\partial x^{(2)} \partial (x')^{{(2)}}},
    \end{bmatrix}
\end{align*}
where exchange of integration and differentiation to obtain the final matrix is permissible by the almost surely continuously differentiable sample paths assumption. 

Similarly,
\begin{align*}
\Cov \left[(\mathrm{rot} \,\Psi)(\bfx),(\mathrm{rot} \,\Psi))(\bfx')\right]
 &=
    \Cov                 
    \left[\begin{pmatrix}
    \frac{\partial\Psi(\bfx)}{\partial x^{(2)}} \\
    -\frac{\partial\Psi(\bfx)}{\partial x^{(1)}}
    \end{pmatrix}
    , 
    \begin{pmatrix}
    \frac{\partial\Psi(\bfx')}{\partial (x')^{(2)}} \\
    -\frac{\partial\Psi(\bfx')}{\partial (x')^{(1)}} 
    \end{pmatrix} 
    \right] \\
&= \begin{bmatrix}
    \Cov\left(\frac{\partial\Psi(\bfx)}{\partial x^{(2)}}, \frac{\partial\Psi(\bfx')}{\partial (x')^{(2)}}\right) 
    & 
    \Cov\left(\frac{\partial\Psi(\bfx)}{\partial x^{(2)}}, -\frac{\partial\Psi(\bfx')}{\partial (x')^{(1)}}\right) 
    \\
    \Cov\left(-\frac{\partial\Psi(\bfx)}{\partial x^{(1)}}, \frac{\partial\Psi(\bfx')}{\partial (x')^{(2)}}\right) 
    & 
    \Cov\left(\frac{\partial\Psi(\bfx)}{\partial x^{(1)}}, \frac{\partial\Psi(\bfx')}{\partial (x')^{(1)}}\right) 
    \end{bmatrix} \\
&= \begin{bmatrix}
    \dfrac{\partial^2 k_{\Psi}(\bfx, \bfx')}{\partial x^{(2)} \partial (x')^{{(2)}}} & -\dfrac{\partial^2 k_{\Psi}(\bfx, \bfx')}{\partial x^{(2)} \partial (x')^{{(1)}}}\\
    -\dfrac{\partial^2 k_{\Psi}(\bfx, \bfx')}{\partial x^{(1)} \partial (x')^{{(2)}}} & \dfrac{\partial^2 k_{\Psi}(\bfx, \bfx')}{\partial x^{(1)} \partial (x')^{{(1)}}}
    \end{bmatrix}.
\end{align*}
The desired expression for $k_{\mathrm{Helm}}$ is obtained by taking the sum of these two matrices.
\end{proof}


\section{Divergence and Vorticity of A Gaussian Process}
\label{app:div-vort-gp}

In this section, we state and prove \cref{prop:divergenceofagp} from the main text.

\divergenceofagp*

\begin{proof}
By the assumption that the sample paths are almost surely continuously differentiable, $\mathrm{div}\,F$ and $\mathrm{curl}\,F$ are well-defined. Since the image of a Gaussian process under a linear transformation is a Gaussian processes both $\mathrm{div}\,F$ and $\mathrm{curl}\,F$ are Gaussian processes. It remains to compute the moments. The expectation can be calculated via linearity,
\begin{align}
\mathbb{E} (\mathrm{div}\,F) &= \mathrm{div} (\mathbb{E} F) = \mathrm{div} \,\mu, \label{eqn:expected-div}\\
\mathbb{E} (\mathrm{curl}\,F) &= \mathrm{curl} (\mathbb{E} F) = \mathrm{curl} \,\mu. \label{eqn:expected-curl}
\end{align}

We next turn to the covariance. Define the centered process $G =F - \mu$. By \cref{eqn:expected-div} and \cref{eqn:expected-curl}, $\mathrm{div}\, G$ and $\mathrm{curl}\, G$ are centered Gaussian processes with the covariance functions $k^{\delta}$ and $k^{\zeta}$ respectively. 

Consider two points $\bfx, \bfx' \in \featurespace$. Unpacking the definition of $\mathrm{div}$,
\begin{align*}
k^{\delta}(\bfx, \bfx') &= 
\mathbb{E}\left[\left(\frac{\partial G^{(1)}(\bfx)}{\partial x^{(1)}} + \frac{\partial G^{(2)}(\bfx)}{\partial x^{(2)}}\right)\left(\frac{\partial G^{(1)}(\bfx')}{\partial (x')^{(1)}} + \frac{\partial G^{(2)}(\bfx')}{\partial (x')^{(2)}}\right)\right] \\
&= \sum_{(i,j) \in \{1, 2\}^2} \mathbb{E}\left[\frac{\partial G^{(i)}(\bfx)}{\partial x^{(i)}} \frac{\partial G^{(j)}(\bfx')}{\partial (x')^{(j)}}\right] \\
&= \!\!\!\!\!\sum_{(i,j) \in \{1, 2\}^2} \!\!\!\!\!\frac{\partial^2 k(\bfx,\bfx')_{i,j}}{\partial x^{(i)}\partial (x')^{(j)}},
\end{align*}

where exchange of integration and differentiation in the final line is permissible given that the sample paths are almost surely continuously differentiable. Similarly,

\begin{align*}
k^{\zeta}(\bfx, \bfx') &= 
\mathbb{E}\left[\left(\frac{\partial G^{(1)}(\bfx)}{\partial x^{(2)}} - \frac{\partial G^{(2)}(\bfx)}{\partial x^{(1)}}\right)\left(\frac{\partial G^{(1)}(\bfx')}{\partial (x')^{(2)}} - \frac{\partial G^{(2)}(\bfx')}{\partial (x')^{(1)}}\right)\right] \\
&= \sum_{(i,j) \in \{1, 2\}^2} (-1)^{i+j}\mathbb{E}\left[\frac{\partial G^{(i)}(\bfx)}{\partial x^{(3-i)}} \frac{\partial G^{(j)}(\bfx')}{\partial (x')^{(3-j)}}\right] \\
&= \!=\!\!\!\!\! \sum_{(i,j) \in \{1, 2\}^2}\!\!\!\! (-1)^{i+j}\frac{\partial^2 k(\bfx, \bfx')_{i,j}}{\partial x^{(3-i)}\partial (x')^{(3-j)}}.
\end{align*}

\end{proof}

\section{Computational Costs for evaluating Helmholtz GP Posterior}
\label{app:comp-cost}

In this section, we provide a bound for the cost of computing velocity predictions using the Helmholtz GP. We also discuss in more detail the assumption of using a Cholesky factorization or QR decomposition.

\begin{proposition} \label{prop:cost_prediction}
    Suppose we have observed $M$ training data points and would like to predict the current at $N$ new (test) locations. Assume we use a Cholesky or QR factorization, together with solving the system of equations with back-substitution. Let $C_{vel}(M, N)$  be the total worst-case\footnote{As in the main text, we assume that the computation incurs the worst-case cost of a Cholesky factorization or QR decomposition. If the matrices involved have special structure, the cost might be much less than the worst-case.} computational cost for evaluating both the posterior mean (\cref{eqn:helm_posterior_mean}) and covariance (\cref{eqn:helm_posterior_cov}) for the velocity GP. Let and $C_{helm}(M, N)$ be the analogous total cost for the Helmholtz GP. Then 
    \begin{equation}\label{eqn:comp-cost-pred}
        \lim_{M, N \rightarrow \infty} C_{helm}(M, N)/C_{vel}(M, N) \leq 4
    \end{equation}
    where $M$ and $N$ can tend to infinity at arbitrary, independent rates.
\end{proposition}

\begin{proof}
Recall that the posterior mean and covariance can be obtained by solving $\Khtetr \Khtrtr^{-1} \Ytrain$ and computing $\Khtete - \Khtetr \Khtrtr^{-1} \Khtetr^\top$, respectively \cref{eqn:helm_posterior_mean,eqn:helm_posterior_cov}. To compute the mean we (A) compute a Cholesky (or QR) factorization of a $2M \times 2M$ matrix, (B) perform a back-solve of a $2M$ dimensional system of equations, and (C) compute $N$ $2M$-dimensional inner products. To compute the covariance we (D) compute a Cholesky (or QR) factorization of a $2M \times 2M$ matrix, (E) perform a back-solve of $N$ distinct $2M$ dimensional systems of equations, (F) compute an $(2N \times 2M) \times (2M \times 2N)$ matrix-multiplication, and (G) subtract $2N \times 2N$ matrices.

We first argue that it suffices to consider steps A and D separately from the remaining steps. 
In particular, for non-negative numbers $\{a_i\}_{i=1}^I$ and $\{b_i\}_{i=1}^I$, we observe that 
\begin{align} \label{eq:bound_by_max_frac}
\frac{\sum_{i=1}^I a_i}{\sum_{i=1}^I b_i} \leq \max \frac{a_i}{b_i}.
\end{align}

For our purposes, let $a_1$ be the cost of steps A and D in the Helmholtz GP computation, and let $a_2$ be the cost of steps B, C, E, F, and G in the Helmholtz GP computation. Analogously, define $b_1$ and $b_2$ for the velocity GP computation.
The proof of \cref{prop:comp-cost} already showed that $a_1/b_1$ is asymptotically bounded above by $4$. By \cref{eq:bound_by_max_frac}, then, it suffices to separately check $a_2/b_2$.

Next, we will use the following observation to focus on the asymptotically dominant terms of $a_2$ and $b_2$, respectively. Consider some index $n \rightarrow \infty$. For $r_n \sim s_n$ and $t_n \sim u_n$ with all terms bounded away from zero,
\begin{align}
\frac{r_n}{s_n} \sim \frac{t_n}{u_n}. \label{eqn:asymptotic-ratio}
\end{align}
%
Among steps B, C, E, F, and G, the dominant costs are due to performing the triangular back-substitution (step E) and the matrix-matrix multiplication (step F) used to compute the covariance. 
We can again consider these two steps separately due to \cref{eq:bound_by_max_frac}.

First consider back-substitution. Given a $t$ (square) lower triangular system of equations of dimension $s$, the cost of solving this system with back-substitution is $\sim ts^2$ floating-point operations \citep[Section 3.1.2]{golub2013matrix}. When evaluating posteriors with the Helmholtz GP, we have $\Khtetr$ of shape $2N \times 2M$ and $\Khtrtr=LL^\top$ (or $QR$) of shape $2M \times 2M$. So back-substitution in the case of the Helmholtz GP incurs cost (in floating point operations) $\sim 8NM^2$. For the velocity GP, we can exploit the fact that the two outputs are uncorrelated and can be handled separately, so the cost of back-substitution is $\sim 4NM^2$. Therefore, and again using \cref{eqn:asymptotic-ratio}, the ratio of the costs of the back-substitution algorithm is asymptotically $2$.

Once back-substitution has been performed, matrix-multiplication must be performed. Using textbook matrix-multiplication  leads to the same considerations as back-substitution. 

In all the steps used to compute the posterior moments, we see that the cost of the Helmholtz GP is not (asymptotically) more than $4$ times the cost of the velocity GP, and the result follows.
\end{proof}

\textbf{Simplifying structure.}
In \cref{prop:comp-cost,prop:cost_prediction},
we have made the assumption that users are solving a linear system or computing a log determinant with a general and standard choice, such as a Cholesky factorization or QR decomposition. We expect essentially the same result to hold for any other general method for computing these quantities. However, if there is special structure that can be used to solve the linear system or compute the log determinant more efficiently, that might change the bounds we have found here. Conversely, we are immediately aware of special structure that we can expect to always apply in the application to modeling current. And any such structure would likely also require special algorithmic development and coding.


\section{Benefits of the Helmholtz GP: additional information and supplemental proofs}

In what follows, we provide additional information on the benefits of using the Helmholtz GP. In \cref{app:sec-equal-marginal-var} we state and prove \cref{prop:equal-var-div-vort}, showing how with independent velocity priors we obtain equal marginal variances for vorticity and divergence. In \cref{sec:appendix_length_scale_relationship} we provide more intuition on the result that length scales are conserved across \texorpdfstring{$k_\Phi$}{kphi} vs.\ \texorpdfstring{$k^\delta$}{kdelta}
and across \texorpdfstring{$k_\Psi$}{kPsi} vs. \texorpdfstring{$k^\zeta$}{kzeta}. Finally, in \cref{app:equivariancehelmgp} we state and prove \cref{prop:Equivariance} about the equivariance of Helmholtz GP predictions.

\subsection{Equality of marginal variances of vorticity and divergence with independent velocity priors} \label{app:sec-equal-marginal-var}
\propEqualVarDivVort*

\begin{proof}
Because $\dt{k}{1}$ and $\dt{k}{2}$ are assumed to be isotropic we may write for any $\bfx, \bfx' \in\R^2$
$$
\dt{k}{1}(\bfx, \bfx')=  \kappa_1(\|\bfx-\bfx'\|^2)
\ \ \mathrm{and}\ \ 
\dt{k}{2}(\bfx, \bfx')=  \kappa_2(\|\bfx-\bfx'\|^2)
$$
for some $\kappa_1, \kappa_2: \R^+ \rightarrow \R.$
Because isotropy implies stationarity, it suffices to consider the variance at any a single point,
and so we consider $\bfx=\bfx'=(0,0)$. By assumption, we have
$$
F \sim \GP\left(
\begin{pmatrix}
0 \\ 0
\end{pmatrix}, 
\begin{pmatrix}
k^{(1)} & 0 \\ 0 & k^{(2)}
\end{pmatrix}
\right)$$
By \cref{prop:divergenceofagp}, the induced divergence and vorticity are Gaussian processes with mean $0$ and covariances

\begin{align*}
k^{\delta}(\bfx, \bfx') &=  \frac{\partial^2 k^{(1)}(\bfx,\bfx')}{\partial x^{(1)}\partial (x')^{(1)}} + \frac{\partial^2 k^{(2)}(\bfx,\bfx')}{\partial x^{(2)}\partial (x')^{(2)}} \\
k^{\zeta}(\bfx, \bfx') &=  \frac{\partial^2 k^{(1)}(\bfx,\bfx')}{\partial x^{(2)}\partial (x')^{(2)}} + \frac{\partial^2 k^{(2)}(\bfx,\bfx')}{\partial x^{(1)}\partial (x')^{(1)}} 
\end{align*}
respectively. 

Then, we can compute the variance at $\bfx=\bfx'=(0,0)$ by

\begin{align*}
\Var[\delta(0,0)]
    &=  \frac{\partial^2 k^{(1)}(\bfx,\bfx')}{\partial x^{(1)}\partial (x')^{(1)}}\big|_{\bfx=0,\bfx'=0} + \frac{\partial^2 k^{(2)}(\bfx,\bfx')}{\partial x^{(2)}\partial (x')^{(2)}}\big|_{\bfx=0,\bfx'=0} \\
    &= \frac{\partial^2 \kappa_1(\|\bfx-\bfx'\|^2)}{\partial x^{(1)}\partial (x')^{(1)}}\big|_{\bfx=0,\bfx'=0} + \frac{\partial^2 \kappa_2(\|\bfx-\bfx'\|^2)}{\partial x^{(2)}\partial (x')^{(2)}}\big|_{\bfx=0,\bfx'=0} \\
    &= \frac{\partial}{\partial x^{(1)}}(-2\kappa'_1(\|\bfx\|^2)x^{(1)}) \big|_{\bfx=0} + \frac{\partial}{\partial x^{(2)}}(-2\kappa'_2(\|\bfx\|^2)x^{(2)}) \big|_{\bfx=0} \\
    &= -2(\kappa_1^{\prime\prime}(0) +\kappa_2^{\prime\prime}(0))
\end{align*}

Consequently, we have that for any $\bfx\in \R^2$,
$\Var[\delta(\bfx)]=-2(\kappa_1^\prime(0) +\kappa_2^\prime(0))$.

The computation is similar for the vorticity.
We have that 

\begin{align*}
\Var[\zeta(0,0)]
    &=  \frac{\partial^2 k(^{(1)}(\bfx,\bfx')}{\partial x^{(2)}\partial (x')^{(2)}}\big|_{\bfx=0,\bfx'=0} + \frac{\partial^2 k(^{(2)}(\bfx,\bfx')}{\partial x^{(1)}\partial (x')^{(1)}}\big|_{\bfx=0,\bfx'=0} \\
    &= \frac{\partial^2 \kappa_1(\|\bfx-\bfx'\|^2)}{\partial x^{(2)}\partial (x')^{(2)}}\big|_{\bfx=0,\bfx'=0} + \frac{\partial^2 \kappa_2(\|\bfx-\bfx'\|^2)}{\partial x^{(1)}\partial (x')^{(1)}}\big|_{\bfx=0,\bfx'=0} \\
    &= \frac{\partial}{\partial x^{(2)}}(-2\kappa'_1(\|\bfx\|^2)x^{(1)}) \big|_{\bfx=0} + \frac{\partial}{\partial x^{(1)}}(-2\kappa'_2(\|\bfx\|^2)x^{(2)}) \big|_{\bfx=0} \\
    &= -2(\kappa_1^{\prime\prime}(0) +\kappa_2^{\prime\prime}(0))
\end{align*}

Therefore for any $\bfx\in \R^2$,
$\Var[\zeta(\bfx)]=-2(\kappa_1^\prime(0) +\kappa_2^\prime(0))$,
and we see $\Var[\zeta(\bfx)]= \Var[\delta(\bfx)].$
This completes the proof.
\end{proof}

\subsection{Conservation of length scales across \texorpdfstring{$k_\Phi$}{kphi} vs.\ \texorpdfstring{$k^\delta$}{kdelta}
(and \texorpdfstring{$k_\Psi$}{kPsi} vs. \texorpdfstring{$k^\zeta$}{kzeta})}\label{sec:appendix_length_scale_relationship}
This subsection provides a derivation of the claim that 
if $k_\Phi(\bfx,\bfx'; \ell)=\kappa(\|\bfx-\bfx'\|/\ell),$
for some $\kappa:\R_+\rightarrow \R,$
then $\khelm^{\delta}(\bfx,\bfx'; \ell)=\ell^{-4}\eta(\|\bfx-\bfx'\|/\ell)$
for some $\eta:\R_+ \rightarrow \R$ that does not depend on $\ell.$
The relationship (argument to see it) is identical $K_\Psi$ and $\khelm^\zeta.$

We may see the claim to be true
by expanding out the dependencies of $\khelm^{\delta}$ and $\khelm^{\zeta}$ on 
$k_{\Phi}$ and $k_\Psi,$
seeing that they involve fourth order partial derivatives,
and applying a change of variables four times;
each change of variables contributes one factor of $\ell^{-1}.$

In particular,
\begin{align}\label{eqn:divergence_kernel}
\begin{split}
\khelm^{\delta} &= \left(
\frac{\partial^{4}}{\partial (x^{(2)})^{2}\partial ((x')^{(2)})^{2}} + 
\frac{\partial^{4}}{\partial (x^{(2)})^{2}\partial ((x')^{(1)})^{2}} + {} \right. 
 \left.
\frac{\partial^{4}}{\partial (x^{(1)})^{2}\partial ((x')^{(2)})^{2}} +
\frac{\partial^{4}}{\partial (x^{(1)})^{2}\partial ((x')^{(1)})^{2}}
\right)k_{\Phi} \quad \mathrm{and} \\
\khelm^{\zeta} &= \left(
\frac{\partial^{4}}{\partial (x^{(2)})^{2}\partial ((x')^{(2)})^{2}} +
\frac{\partial^{4}}{\partial (x^{(2)})^{2}\partial ((x')^{(1)})^{2}} + {} \right. 
 \left.
\frac{\partial^{4}}{\partial (x^{(1)})^{2}\partial ((x')^{(2)})^{2}} +
\frac{\partial^{4}}{\partial (x^{(1)})^{2}\partial ((x')^{(1)})^{2}}
\right)k_{\Psi}.
\end{split}
\end{align}

Consider first the potential function and divergence.
If $k_\Phi(\bfx, \bfx')=\kappa(\|\bfx-\bfx'\|),$
then any second order mixed partial derivative may be written through a change of variables ($x$ to $x/\ell$) as 
$\frac{\partial^2}{\partial \bfx\partial \bfx'} k_\Phi(\bfx, \bfx')=\ell^{-2}\frac{\partial^2}{\partial \ell \bfx \partial \ell \bfx'} \kappa(\|\ell \bfx-\ell \bfx'\|/\ell)=\ell^{-2}\frac{\partial^2}{\partial \bfx\partial \bfx'} \kappa(\|\bfx-\bfx'\|/\ell).$
Analogously, when we differentiate four times rather than twice to obtain $\khelm^{\delta}$
we have that if 
$k_\Phi(\bfx,\bfx')=\kappa(\|\bfx-\bfx'\|/\ell),$
for some $\kappa,$
then $\khelm^{\delta}=\ell^{-4}\eta(\|\bfx-\bfx'\|/\ell)$
for some $\eta$ that does not depend on $\ell.$

\subsection{Equivariance of Helmoltz GP predictions}
\label{app:equivariancehelmgp}
\propEquivariance*
To prove the proposition, it is helpful to distinguish between random variables and the values they take on.
We use boldface to denote the random variables, for example $\Ytrain$.
When a random variable $\Ytrain$ takes a value $Y$ we write $\Ytrain=Y.$
The rotation operator $R$ is characterized by a $2\times 2$ rotation matrix;
if $\Xtrain=[(\dt{x_1}{1}, \dt{x_1}{2})^\top, \dots, 
(\dt{x_N}{1}, \dt{x_N}{2})^\top],$
then $R\Xtrain=[R(\dt{x_1}{1}, \dt{x_1}{2})^\top, \dots, 
R(\dt{x_N}{1}, \dt{x_N}{2})^\top] = [(\dt{(R\bfx_1)}{1}, \dt{(R\bfx_1)}{2})^\top, \dots, (\dt{(R\bfx_N)}{1}, \dt{(R\bfx_N)}{2})^\top]$, where we denote by $\dt{(R\bfx)}{1}$ the rotated first coordinate, and $\dt{(R\bfx)}{2}$ the rotated second coordinate. When the input is flattened, as in the case of $\Ytrain$ or $\mu_{F \mid D}$, the R operator is applied as follows: (1) unflatten the vector to get it in the same form as $\Xtrain$, then (2) apply the operator R as specified above, and finally (3) flatten the output vector to go back to the original $\Ytrain$ shape.
Our proof relies on $k_\Phi$ and $k_\Psi$ being isotropic kernels.

\begin{lemma}[Invariance of the likelihood]\label{lemma:invariance_of_likelihood}
Suppose $F$ is distributed as a Helmholtz GP,
and there are $M$ observations $\Ytrain\mid F, \Xtrain=X \sim \mathcal{N}\left([F^{(1)}(X), F^{(2)}(X)]^\top, \Khtrtr(X, X)\right)$, where $I_{2M}$ denotes the identity matrix of size $2M$.
Then the marginal likelihood of the observations is \emph{invariant} to rotation.
That is, for any $2\times 2$ rotation matrix $R,$
$$
p(\Ytrain=Y\mid \Xtrain=X) = p(\Ytrain=RY \mid \Xtrain = RX).
$$
\end{lemma}
\begin{proof}
By assumption, $k_\Phi$ is stationary and so, for any two locations $\bfx$ and $\bfx'$ in $\featurespace$  we may write
$k_\Phi(\bfx, \bfx')=\kappa(\|\bfx-\bfx'\|)$ for some function $\kappa:\R_+\rightarrow \R.$
Following Appendix C, we may write the induced covariance for $\mathrm{grad}\Phi$ as

\begin{align*}
 \Cov \left[(\mathrm{grad} \,\Phi)(\bfx),
(\mathrm{grad} \,\Phi)(\bfx')\right] 
&= \begin{bmatrix}
    \dfrac{\partial^2 k_{\Phi}(\bfx, \bfx')}{\partial x^{(1)} \partial (x')^{{(1)}}} & \dfrac{\partial^2 k_{\Phi}(\bfx, \bfx')}{\partial x^{(1)} \partial (x')^{{(2)}}}\\
    \dfrac{\partial^2 k_{\Phi}(\bfx, \bfx')}{\partial x^{(2)} \partial (x')^{{(1)}}} & \dfrac{\partial^2 k_{\Phi}(\bfx, \bfx')}{\partial x^{(2)} \partial (x')^{{(2)}}} 
\end{bmatrix} \\
&= \begin{bmatrix}
    \dfrac{\partial^2 \kappa(\|\bfx-\bfx'\|)}{\partial x^{(1)} \partial (x')^{{(1)}}} & \dfrac{\partial^2 \kappa(\|\bfx-\bfx'\|)}{\partial x^{(1)} \partial (x')^{{(2)}}}\\
    \dfrac{\partial^2 \kappa(\|\bfx-\bfx'\|)}{\partial x^{(2)} \partial (x')^{{(1)}}} & \dfrac{\partial^2 \kappa(\|\bfx-\bfx'\|)}{\partial x^{(2)} \partial (x')^{{(2)}}},
\end{bmatrix}
\end{align*}

Similarly, we may compute $\Cov \left[(\mathrm{grad} \,\Phi)(R\bfx),
(\mathrm{grad} \,\Phi)(R\bfx')\right]$ through a change of variables ($\bfx$ to $R\bfx$) as 

\begin{align*}
\Cov \left[(\mathrm{grad} \,\Phi)(R\bfx),
(\mathrm{grad} \,\Phi)(R\bfx')\right] &= \begin{bmatrix}
    \dfrac{\partial^2 k_{\Phi}(R\bfx, R\bfx')}{\partial x^{(1)} \partial (x')^{{(1)}}} & \dfrac{\partial^2 k_{\Phi}(R\bfx, R\bfx')}{\partial x^{(1)} \partial (x')^{{(2)}}}\\
    \dfrac{\partial^2 k_{\Phi}(R\bfx, R\bfx')}{\partial x^{(2)} \partial (x')^{{(1)}}} & \dfrac{\partial^2 k_{\Phi}(R\bfx, R\bfx')}{\partial x^{(2)} \partial (x')^{{(2)}}} 
\end{bmatrix} \\
&= \begin{bmatrix}
    \dfrac{\partial^2 \kappa(\|R\bfx-R\bfx'\|)}{\partial x^{(1)} \partial (x')^{{(1)}}} & \dfrac{\partial^2 \kappa(\|R\bfx-R\bfx'\|)}{\partial x^{(1)} \partial (x')^{{(2)}}}\\
    \dfrac{\partial^2 \kappa(\|R\bfx-R\bfx'\|)}{\partial x^{(2)} \partial (x')^{{(1)}}} & \dfrac{\partial^2 \kappa(\|R\bfx-R\bfx'\|)}{\partial x^{(2)} \partial (x')^{{(2)}}},
\end{bmatrix} \\
&= R^\top \Cov \left[(\mathrm{grad} \,\Phi)(\bfx),
(\mathrm{grad} \,\Phi)(\bfx')\right] R
\end{align*}

and see that $ \Cov \left[(\mathrm{grad} \,\Phi)(\bfx),
(\mathrm{grad} \,\Phi)(\bfx')\right] =R\Cov \left[(\mathrm{grad} \,\Phi)(R\bfx),
(\mathrm{grad} \,\Phi)(R\bfx')\right])R^\top.$

Similarly, for a collections of $M$ locations $X$ 
we have that $\Cov \left[(\mathrm{grad} \,\Phi)(X),
(\mathrm{grad} \,\Phi)(X)\right] =(R \otimes I_M)\Cov \left[(\mathrm{grad} \,\Phi)(RX),
(\mathrm{grad} \,\Phi)(RX)\right] (R^\top \otimes I_M),$ where $\otimes$ denotes the Kronecker product.

An identical argument (up to a change in the sign of off-diagonal terms) can be used to derive the induced covariance for $\mathrm{rot}\Psi$, $\Cov \left[(\mathrm{rot} \,\Psi)(\bfx),
(\mathrm{rot} \,\Psi)(\bfx')\right]$. We obtain 
\begin{align*}
\Cov \left[(\mathrm{rot} \,\Psi)(\bfx),
(\mathrm{rot} \,\Psi)(\bfx')\right]&=R\Cov \left[(\mathrm{rot} \,\Psi)(R\bfx),
(\mathrm{rot} \,\Psi)(R\bfx')\right]R^\top \quad \text{and}\\
\Cov \left[(\mathrm{rot} \,\Psi)(X),
(\mathrm{rot} \,\Psi)(X)\right]&=(R \otimes I_M)\Cov \left[(\mathrm{rot} \,\Psi)(RX),
(\mathrm{rot} \,\Psi)(RX)\right] (R^\top \otimes I_M).
\end{align*}

 Together, this implies that if we write the covariance of $M$ vector velocity \emph{tr}aining observations $\Ytrain$ at $X$ as 
 \begin{align*}
 \Khtrtr(X,X) :&=\Var [\Ytrain|\Xtrain=X] \\
 &= k_{Helm}(X, X) + \sigma_{obs}^2I_{2M} \\
 &= \Cov \left[(\mathrm{grad} \,\Phi)(X),
(\mathrm{grad} \,\Phi)(X)\right] +\Cov \left[(\mathrm{rot} \,\Psi)(X),
(\mathrm{rot} \,\Psi)(X)\right]+ \sigma_{obs}^2I_{2M}
 \end{align*}
 then 
 $$
 \Khtrtr(X,X)  = (R \otimes I_M) \Khtrtr(RX, RX) (R^\top \otimes I_M)
 $$
 
As a result, for any $R, Y$ and $X$ we may compute the log likelihood according to the likelihood model as  
 \begin{align*}
\log p(\Ytrain=RY &\mid \Xtrain=RX) \\
&= \log \mathcal{N}\left(RY;0, \Khtrtr(RX, RX) \right) \\
&=-M\log(2\pi) -\frac{1}{2}\log\left|  \Khtrtr(RX, RX) \right| \\
&- \frac{1}{2} \log ((R\otimes I_M) Y)^\top \left[ \Khtrtr(RX, RX) \right]^{-1}((R\otimes I_M) Y)\\
&=-M\log(2\pi) -\frac{1}{2}\log\left|  \Khtrtr(X, X) \right| - \frac{1}{2} \log Y^\top \Khtrtr(X, X)^{-1} Y\\
&=\log p(\Ytrain=Y \mid \Xtrain=X),
\end{align*}
as desired.

\end{proof}
\begin{lemma}[Invariance of the conditionals]\label{lemma:invariance_of_conditionals}
The conditionals distributions of the Helmoltz GP are invariant to rotation.
That is, for any $2\times 2$ rotation matrix $R,$
\begin{align*}
&p(\Ytest=\ytest \mid \Xtest=\xtest, \Xtrain=\xtrain, \Ytrain=\ytrain)  \\
=&p(\Ytest=R\ytest \mid \Xtest=R\xtest, \Xtrain=R\xtrain, \Ytrain=R\ytrain) 
\end{align*}
\end{lemma}
\begin{proof}
The lemma is obtained by applying Bayes' rule and \Cref{lemma:invariance_of_likelihood} as
\begin{align*}
&p(\Ytest=R\ytest \mid \Xtest=R\xtest, \Xtrain=R\xtrain, \Ytrain=R\ytrain) \\
&=\frac{
p(\Ytest=R\ytest, \Xtest=R\xtest, \Xtrain=R\xtrain, \Ytrain=R\ytrain)
}{
\int p(\Ytest=R\ytest^\prime, \Xtest=R\xtest, \Xtrain=R\xtrain, \Ytrain=R\ytrain) dR\ytest^\prime 
}\\
&=\frac{
p(\Ytest=\ytest, \Xtest=\xtest, \Xtrain=\xtrain, \Ytrain=\ytrain)
}{
\int p(\Ytest=\ytest^\prime, \Xtest=\xtest, \Xtrain=\xtrain, \Ytrain=\ytrain) d\ytest^\prime 
}\\
&=p(\Ytest=\ytest \mid \Xtest=\xtest, \Xtrain=\xtrain, \Ytrain=\ytrain).
\end{align*}
\end{proof}

\textbf{Proof of the equivariance proposition:}

We now prove the proposition.
Recall that 
$$\mu_{F\mid D}(\Xtest, \Xtrain, \Ytrain)=\E[\Ytest\mid \Xtest=\xtest, \Ytrain=\ytrain, \Xtrain=\xtrain].$$
Therefore, for any $R$, we may compute $\mu_{F\mid D}(R\Xtest, R\Xtrain, R\Ytrain)$ as 
\begin{align*}
\mu_{F\mid D}(R\Xtest, R\Xtrain, R\Ytrain)
&=\E[\Ytest\mid \Xtest=R\xtest, \Ytrain=R\ytrain, \Xtrain=R\xtrain] \\
&=\int \ytest \  p(\Ytest=\ytest \mid \Xtest=R\xtest, \Ytrain=R\ytrain, \Xtrain=R\xtrain)  d\ytest \\
&=\int R\ytest\  p(\Ytest=R\ytest \mid \Xtest=R\xtest, \Ytrain=R\ytrain, \Xtrain=R\xtrain)  d\ytest \\
&=\int R\ytest\  p(\Ytest=\ytest \mid \Xtest=\xtest, \Ytrain=\ytrain, \Xtrain=\xtrain)  d\ytest \\
&=\E[R\Ytest \mid \Xtest=\xtest, \Ytrain=\ytrain, \Xtrain=\xtrain] \\
&=R\mu_{F\mid D}(\Xtest, \Xtrain, \Ytrain)
\end{align*}
Where in the third line we perform a change of variables, noting that $|R|=1.$
The fourth line follows from \Cref{lemma:invariance_of_conditionals}.
The final line is a result of linearity of expectation and the definition of $\mu_{F\mid D},$
and provides the desired equality.

\subsection{Non-Equivariance of Velocity GP predictions}
\label{app:equivariancevelgp}
In this appendix, we show that the velocity GP requires special constraints to exhibit reference-frame equivariance, and these constraints force an undesirable coupling of divergence and vorticity length scales.

\propEquivarianceVelGP*

\proof 
We first show that if the velocity GP is reference-frame equivariant and has isotropic kernels for each component, then the kernels for the two velocity components are equal. In order to show this, it suffices to show that the prior is not equivariant, as this is the posterior given the empty dataset. Let $F \sim GP(0, k)$ be a function from $\mathbb{R}^2 \rightarrow \mathbb{R}^2$ by stacking $F(\cdot)=\left[F^{(1)}(\cdot), F^{(2}(\cdot)\right]^T$. Rotational invariance of the prior is equivalent to the condition that for an arbitrary 2-dimensional rotation matrix $R, R^{-1} F(R \cdot)=F(\cdot)$, where equality is in distribution for an entire sample path. Consider a 90 degree rotation of the coordinate axis. In the case of a 90 degree rotation, by isotropy of the kernels, $R^{-1} F(R \cdot)=R^{-1} F(\cdot)=\left[F_2(\cdot), F_1(\cdot)\right]^T$. Considering each coordinate of $R^{-1} F(R \cdot)$, equality in distribution to $F$ implies that $F^{(2)}(\cdot)=F^{(1)}(\cdot)$ in distribution, and so the two components must have the same kernel.

We next show that if the two kernels are equal, then the velocity GP is rotationally equivariant. By \cref{lemma:invariance_of_conditionals}, it suffices to show there exist kernels $k_\Phi, k_\Psi$ such that the prior is equal in distribution to a Helmholtz GP with these kernels. Let $F$ denote a zero mean velocity GP with kernels $k^{(1)} = k^{(2)}$. Let $F'$ denote a Helmholtz GP with kernels $k_\Phi = k_\Psi$, and $k_\Phi(\bfx, \bfx') = \int_{s'=0}^{x^{(1)}}\int_{s=0}^{(x')^{(1)}} k^{(1)}(s, s')dsds'+ \int_{s'=0}^{x^{(2)}}\int_{s=0}^{(x')^{(2)}} k^{(2)}(s, s')dsds'$. Applying \cref{prop:helm-prior} and the fundamental theorem of calculus, we see the covariance functions of $F$ and $F'$ are equal, and since $F$ and $F'$ are zero mean Gaussian processes, they are therefore equal in distribution.

\newpage

\section{Experimental Results}
\label{app:experiment}

In this section, we provide more details on our experimental results. The section is organized in three parts. \Cref{app:sec-synthetic-experiments} focuses on experiments with simulated data. \Cref{app:experiments-laser} focuses on experiments with real data from the LASER experiment \citep{dasaroLASERdata2017}. \Cref{app:experiments-glad} focuses on real data from the GLAD experiment \citep{ozgokmen2012carthe}. In each section, we have one subsection for each experiment. These subsections provide simulation details (e.g., what is the underlying vector field, and how we generated the buoys trajectories), model fitting details (e.g., hyperparameter optimization), and results. At the end of each subsection, we include a figure with these results. All the figures have the same structure. The first column represents ground truths. Second, third and fourth columns contain, respectively, $\sehelmgp$, $\sevelgp$, and D-HNN results. The first two rows represent results for the velocity prediction task: row 1 shows reconstructed velocity fields, row 2 differences from ground truth. Rows 3, 4, and 5 are about divergence: first divergence predictions, then standard deviation and z-values for the two GP models. Finally, rows 6, 7, and 8 concern vorticity: vorticity predictions, standard deviation and z-values for the two GP models. See \cref{fig:vortexapp} for an example. For real data experiments, where we do not have ground truths, we omit the first column and the second row, i.e., all plots involving comparisons with ground truth quantities. See \cref{fig:lasersparse} for an example.

For the simulated experiments, all root mean square errors are evaluated on the grids used to simulate the experiment. Specific grids are discussed in the ``simulation details`` paragraph of each individual experiments subsection. More explicitly, the root mean square error is calculated as,
\begin{align}
    \text{RMSE}=\sqrt{\frac{1}{|L|}\sum_{\bfx \in L} \|F(\bfx) - \hat{F}(\bfx)\|_2^2}
\end{align}
where $F(\bfx)$ denotes the simulated vector field, $\hat{F}(\bfx)$ denotes the predictions of a given model and $L$ is the grid used to simulate the vector field.

\paragraph{Initialization}

In all experiments except the GLAD data, we initialize our parameters such that $\log \ell_{\Phi} = 0, \log \sigma_{\Phi}=0, \log \ell_{\Psi} = 1, \log \sigma_{\Psi} = -1, \log \obsvar = -2$
and likewise
$\log \ell_{1} = 0, \log \sigmaFu=0, \log \ell_{2} = 1, \log \sigmaFv = -1, \log \obsvar = -2$.
We describe the special difficulties of the GLAD data in \cref{app:experiments-glad}. In all other cases, we found that results were not sensitive to initialization.

\subsection{Simulated Experiments}\label{app:sec-synthetic-experiments}

We focus on simulations of key ocean behaviors of interest to oceanographers: vortices, straight currents, concentrated divergences, and combinations thereof.

\subsubsection{Simulated experiment 1: single vortex}
\label{app:subsectionvortex}

A single vortex in the ocean is a fluid flow pattern in which water particles rotate around a central point, with the flow pattern resembling a spiral. These vortices can occur due to a variety of factors such as the wind, currents, and tides. Single ocean vortices, also known as ocean eddies, can have a significant impact on ocean circulation and can transport heat, salt, and nutrients across vast distances. They can also affect the distribution of marine life. The vortex constructed has zero divergence and constant vorticity.

\paragraph{Simulation details.} To simulate a vortex vector field in a two dimensional space, we first define a grid of points $L$ of size 17 x 17, equally spaced over the interval $[-1,1] \times [-1,1]$. For each point $\bm{x} = (x^{(1)}, x^{(2)}) \in L$, we compute the vortex longitudinal and latitudinal velocities by: 
\begin{align*}
    F^{(1)}(\bm{x}) &= -x^{(2)} \\
    F^{(2)}(\bm{x}) &= x^{(1)}
\end{align*}
From these equations we obtain that the divergence of the vortex is 0 for any $\bm{x} = (x^{(1)}, x^{(2)}) \in L$:
\begin{align*}
\delta(\bm{x}) = \mathrm{div} \cdot F = \frac{\partial F^{(1)}}{\partial x^{(1)}} + \frac{\partial F^{(2)}}{\partial x^{(2)}} = 0 + 0 = 0
\end{align*}
and the vorticity is -2 for any $\bm{x} = (x^{(1)}, x^{(2)}) \in L$: 
\begin{align*}
    \zeta(\bm{x}) = \mathrm{curl} \cdot F = \frac{\partial F^{(1)}}{\partial x^{(2)}} - \frac{\partial F^{(2)}}{\partial x^{(1)}} = -1 -1 = -2
\end{align*}

In our simulated experiment, we then use this vector field to simulate buoys trajectories, i.e. the evolution of buoys positions and velocities across time. In doing so, we make an implicit stationarity assumption about the vector field. That is, we assume that across the total time where we want to simulate buoys trajectories, the vector field remains the same. Then we fix starting positions for the desired amount of buoys, in this case 4. We set these to be just on one side of the vortex, to evaluate the ability of the models to reconstruct the full vortex by having access to observations covering only a portion of it. We pick the total time (here 1) for which we observe the trajectories, and the amount of time steps at which we want to observe the buoys trajectories (here 2), to split the total time. To find the trajectories, we solve the velocity-time ordinary differential equation, $d \bm{x}/dt = F$, where $d/dt$ represents the time-derivative operator. Once we obtain the evolution of buoys' locations, we obtain the corresponding velocities by doing a linear interpolation of the underlying vortex field. By doing this interpolation, we end up with our simulated dataset, consisting in this case of 8 observations. 

\paragraph{Model fitting.} We are interested in evaluating the models' capabilities of reconstructing the full vortex, and capturing the underlying divergence and vorticity structure. To do so, we consider test locations corresponding to the grid $L$, so that we can compare our results with the ground truth, for velocities, divergence, and vorticity. To fit the $\sehelmgp$, we initialize the hyperparameters as explained in the Initialization paragraph at the start of this appendix:  $\ell_{\Phi} = 1, \sigma_{\Phi}=1, \ell_{\Psi} = 2.7, \sigma_{\Psi} = 0.368, \obsvar = 0.135$. The objective function of our optimization routine is the log marginal likelihood from \cref{eqn:likelihood}. We optimize the parameter using the gradient-based algorithm Adam \citep{kingma2015adam}. Note that we optimize the hyperparameters in the $\log$-scale. That is, we consider as parameters in the optimization step $\log\ell_{\Phi}, \log\sigma_{\Phi}, \log\ell_{\Psi}, \log\sigma_{\Psi}$, and $\log\obsvar$, and we exponentiate these when evaluating the log marginal likelihood. In doing so, we ensure that the optimal parameters are positive, as needed in this model. We run the optimization routine until the algorithm reaches convergence. In this case, the convergence criterion is the difference of log marginal likelihood in two consecutive optimization steps being less than $10^{-4}$. This convergence is achieved in less than 1000 iterations. The optimal hyperparameters are: $\ell_{\Phi} = 1.1131, \sigma_{\Phi}=0.0342, \ell_{\Psi} = 1.5142, \sigma_{\Psi} = 0.8884, \obsvar = 0.1597$. The same optimization routine is performed for the $\sevelgp$. In this case, the initial hyperparameters are $\ell_{1} = 1, \sigmaFu=1, \ell_{2} = 2.7, \sigmaFv = 0.368, \obsvar = 0.135$. The optimal hyperparameters are: $\ell_{1} = 1.6191, \sigmaFu=0.9710, \ell_{2} = 2.7183, \sigmaFv = 0.5811, \obsvar = 0.1759$. For both optimization routines, we tried different initial parametrizations, and the results agree substantially both in terms of RMSEs and visual reconstruction. Finally, to train the D-HNN model, we run the training routine provided in \citet{greydanus2022dissipative} code. 

\paragraph{Results.} We show the results in \cref{fig:vortexapp}. For each of the plots, the horizontal and vertical axes represent, respectively, latitude and longitude. The first row represents the ground truth simulated vector field (left), and the reconstruction using the $\sehelmgp$ (center-left), the $\sevelgp$ (center-right) and the D-HNN (right). Red arrows are the observed buoy data, black arrows show the predicted current at test locations. We can see how our method predicts a full vortex covering the spatial domain, whereas the $\sevelgp$ predicts a smooth curve with much longer length scale, that does not resemble the ground truth. The D-HNN prediction looks more similar to a vortex, but still not as good as the $\sehelmgp$. To support this claim we also show differences from the ground truth in the second row. Finally, note that the RMSE for the $\sehelmgp$ is 0.24, whereas for the $\sevelgp$ it is 0.72 and for the D-HNN is 0.54. 

In the third row, we analyze the divergence. The left box shows the constantly zero ground truth. Our model prediction (center-left) correctly captures this behavior, whereas the $\sevelgp$ (center-right) predicts an irregular pattern not resembling the truth. The same happens for D-HNN (right box). In the fourth row we show the standard deviation of divergence predictions for the two GP models, and we can see how the $\sehelmgp$ is very certain that there is no divergence, whereas the uncertainty for the $\sevelgp$ predictions is higher. Finally, in the fifth row, we show the z-values for the divergence prediction, defined as the ratio between the mean and the standard deviation. This is a measure of how far from zero the prediction is, measured in terms of standard deviation. Some standard cut-off values for this quantity are $-1$ and $+1$, and one usually concludes that the prediction is significantly different (in the sense of one standard deviation) from 0 if the corresponding z-value is beyond these thresholds. By using this indicator, we conclude that none of the two predictions are significantly far from zero, so both models are accurate in predicting zero divergence, but our prediction is more precise, in the sense that the mean is closer to the real value and the uncertainty is lower. This is confirmed by looking at RMSEs: $0.0$ for the $\sehelmgp$, $0.22$ for $\sevelgp$, and $0.87$ for the D-HNN. 

Finally, in the last three rows we analyze results for the vorticity. The left box shows the constant (-2) ground truth. The $\sehelmgp$ (center-left) predicts that the vorticity is around that value, especially in the center of the vortex, whereas in the corners the behavior is not as good. The $\sevelgp$ (center-right) performs much worse also here, by predicting vorticity very close to zero, or positive, on almost all the spatial domain. The D-HNN (right box) predicts negative vorticity in most of the domain, but the pattern is very irregular. In the second-to-last row we show the standard deviation of divergence predictions for the two GP models, and we can see how the range of uncertainties on this task is more similar than before, meaning that there are areas where both models are not very confident. Still, if we look at the z-values in the last row, combined with the prediction plots, we see our model is better at predicting the magnitude and size of the vorticity area. In terms of RMSEs, we have $0.77$ for the $\sehelmgp$, $1.05$ for the $\sevelgp$, and $1.03$ for the D-HNN. 

In general, in this experiment we have shown that when working with this very simple underlying vector field, our model behaves better than the alternatives. In particular, we have seen how the prediction of the vortex is very accurate for the $\sehelmgp$, whereas the two other models are more off (and this is reflected in the respective RMSEs). In terms of divergence, our model predicts with certainty that there is no divergence, whereas the $\sevelgp$ approach is less precise (by predicting non-zero divergence with high uncertainty). Finally, we saw how in terms of vorticity our model is the only one able to understand that there is a non-zero vorticity: even if the prediction is not perfect, it is still significantly better than all the other models. 

\afterpage{
\begin{figure*}[!h]
    \centering
    \includegraphics[width=0.96\textwidth]{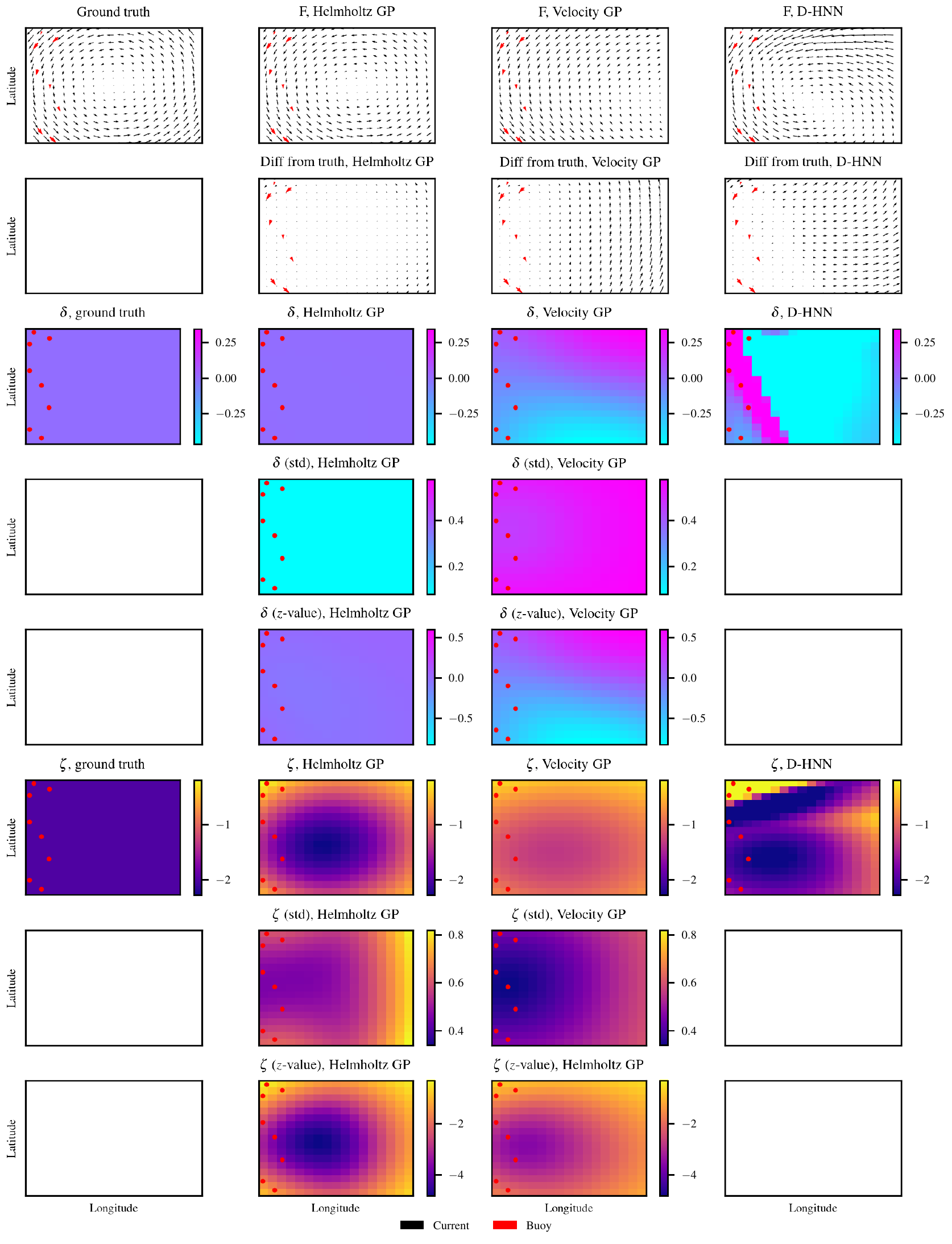}
    \caption{Single vortex. First column: ground truths. Second column: $\sehelmgp$ results. Third column: $\sevelgp$ results. Fourth column: D-HNN results.}
    \label{fig:vortexapp}
\end{figure*}
\clearpage
}

\newpage

\subsubsection{Simulated experiment 2: vortex adjacent to straight current}\label{app:sec-vortex-current}

This task elaborates on the previous one by splitting the spatial domain in two regions: the upper region has a single vortex, and the lower region has a constant and straight flow from left to right. 

The velocity field is discontinuous at the boundary, and therefore the divergence and vorticity are not defined as fields along this line.
In general, though, if we were to consider a smaller sub-region containing the boundary on the lefthand side (where the vortex has downward velocity), we expect such a sub-region to have a negative flux. Analogously, if we were to consider a sub-region containing the boundary on the righthand side (where the vortex has upward velocity), we would expect that sub-region to have a positive flux. 

Since the divergence and vorticity fields are not defined everywhere, it is not obvious what the desired behavior is in terms of recovering these fields. Due to this ambiguity, we do not report the results of this experiment in the main text. However, these discontinuities are reminiscent of fronts, which are of substantial interest to oceanographers, so we include this experiment in this appendix for completeness and in case it spurs future advancements.

\paragraph{Simulation details.} To simulate such a vortex vector field, we first define a grid of points $L$ of size 25 x 50, equally spaced over the interval $[-1,1] \times [-1,2]$. We can see this grid as composed of two subgrids $L_1$ and $L_2$, each of dimension 25 x 25, with $L_1$ representing the top grid and $L_2$ the lower one.  Next, for each point $\bm{x} = (x^{(1)}, x^{(2)}) \in L_1$, we compute the vortex as done in \cref{app:subsectionvortex}:
\begin{align*}
    F^{(1)}(\bm{x}) &= -x^{(2)} \\
    F^{(2)}(\bm{x}) &= x^{(1)}
\end{align*}
and we still have $\delta(\bm{x}) = 0$ and $\zeta(\bm{x}) = -2$ for any $\bm{x} = (x^{(1)}, x^{(2)}) \in L_1$. 

For each point $\bm{x} \in L_2$, we simulate a constant field with the following equations:
\begin{align*}
    F^{(1)}(\bm{x}) &= 0.7 \\
    F^{(2)}(\bm{x}) &= 0.
\end{align*}
The divergence and vorticity for each $\bm{x} \in L_2$
are $\delta(\bm{x}) = 0$ and $\zeta(\bm{x}) = 0$. 

As done for the previous experiment, we then use this vector field to simulate buoys trajectories making the stationarity assumption. Here we consider 7 buoys, covering the full region, observed for a total time of 0.5 and 2 time steps. We reconstruct the buoys trajectories by solving the ODE and interpolating as specified before. By doing this interpolation, the simulated dataset consists of 14 observations.

\paragraph{Model fitting.} We fit the three models with the routine specified in \cref{app:subsectionvortex}. To fit the $\sehelmgp$, we initialize the hyperparameters as follows:  $\ell_{\Phi} = 1, \sigma_{\Phi}=1, \ell_{\Psi} = 2.7, \sigma_{\Psi} = 0.368, \obsvar = 0.135$. The optimal hyperparameters are: $\ell_{\Phi} = 3.8698, \sigma_{\Phi}=0.0885, \ell_{\Psi} = 1.2997, \sigma_{\Psi} = 0.9773, \obsvar = 0.0609$. The same optimization routine is performed for the $\sevelgp$. In this case, the initial hyperparameters are $\ell_{1} = 1, \sigmaFu=1, \ell_{2} = 2.7, \sigmaFv = 0.368, \obsvar = 0.135$. The optimal hyperparameters are: $\ell_{1} = 0.9397, \sigmaFu=1.0755, \ell_{2} = 2.7183, \sigmaFv = 0.5528, \obsvar = 0.0087$. For both optimization routines, we tried different initial parametrizations, and the results agree substantially both in terms of RMSEs and visual reconstruction. 

\paragraph{Results.} We show the results in \cref{fig:contvortexapp}. As before, for each of the plots, the horizontal and vertical axes represent, respectively, latitude and longitude. The first row represents the ground truth simulated vector field (left), and the reconstruction using the $\sehelmgp$ (center-left), the $\sevelgp$ (center-right) and the D-HNN (right). Red arrows are the observed buoy data, black arrows show the predicted current at test locations. We can see how our method predicts accurately the vortex structure, whereas it has some problems in the lower right corner, and smooths out the discontinuity at the boundary. The $\sevelgp$ is accurate as well for the vortex part, but has a significant issue in the lower subgrid: the current flows from left to right, then gets interrupted, and then restarts in a different direction. This behavior goes against the idea that currents are continuous (by conservation of momentum). The D-HNN predictions look very similar to the $\sehelmgp$. In the second row we include the differences from the ground truth, and these show as well that $\sehelmgp$ and D-HNN are accurate, whereas the $\sevelgp$ has issues in the lower part of the grid. The RMSE for the $\sehelmgp$ is $0.30$, whereas for the $\sevelgp$ it is $0.49$ and for the D-HNN is $0.28$. 

In the third row, we analyze the divergence. The left box shows the constantly zero ground truth except for at the boundary line, where the divergence is undefined. Our model captures the zero divergence outside the boundary line, but does not estimate any divergence at or around the boundary line. The $\sevelgp$ (center-right) estimates an irregular pattern. The D-HNN divergence estimate has the perhaps desirable property that it is negative on the left side and positive on the right side, which might lead to reasonable predictions of flux into and out of regions containing the boundary between the straight current and vortex.
As it is unclear what a reasonable approximation to the divergence is when the divergence is concentrated on a line, we do not report RMSE on this example. 

Finally, in the last three rows we analyze results for the vorticity. The left box shows the ground truth, again ignoring the boundary line where the vorticity is undefined. Here both the GP models' predictions look very similar. 
Both predict that there is a negative vorticity area in the top grid, and a close-to-zero vorticity area in the lower grid.
The D-HNN vorticity estimates in the upper and lower regions are similar to those of the GPs. Also the D-HNN estimate suggests a crisper boundary than the GP approaches do.
In the standard deviation and z-value plots, we see that both GP models seem quite certain about the existence of a negative vorticity at least in part of the upper region.
We again do not report RMSE for vorticity in this example as the vorticity is not well-defined over the entire region we consider. 

In summary, in this experiment we showed a situation in which the $\sehelmgp$ is at least as good as the other two models in predicting the velocity field. It is not entirely clear what desirable reconstruction of the divergence and vorticity field is, as they are not defined on the boundary of the two regions, and the models show substantially different properties. Finally, we note that this field violates the modeling assumptions made by the the two GP models (particularly continuity) and further modeling innovation is likely needed to improve fidelity in examples like this one.


\afterpage{
\begin{figure*}[!h]
    \centering
    \includegraphics[width=0.96\textwidth]{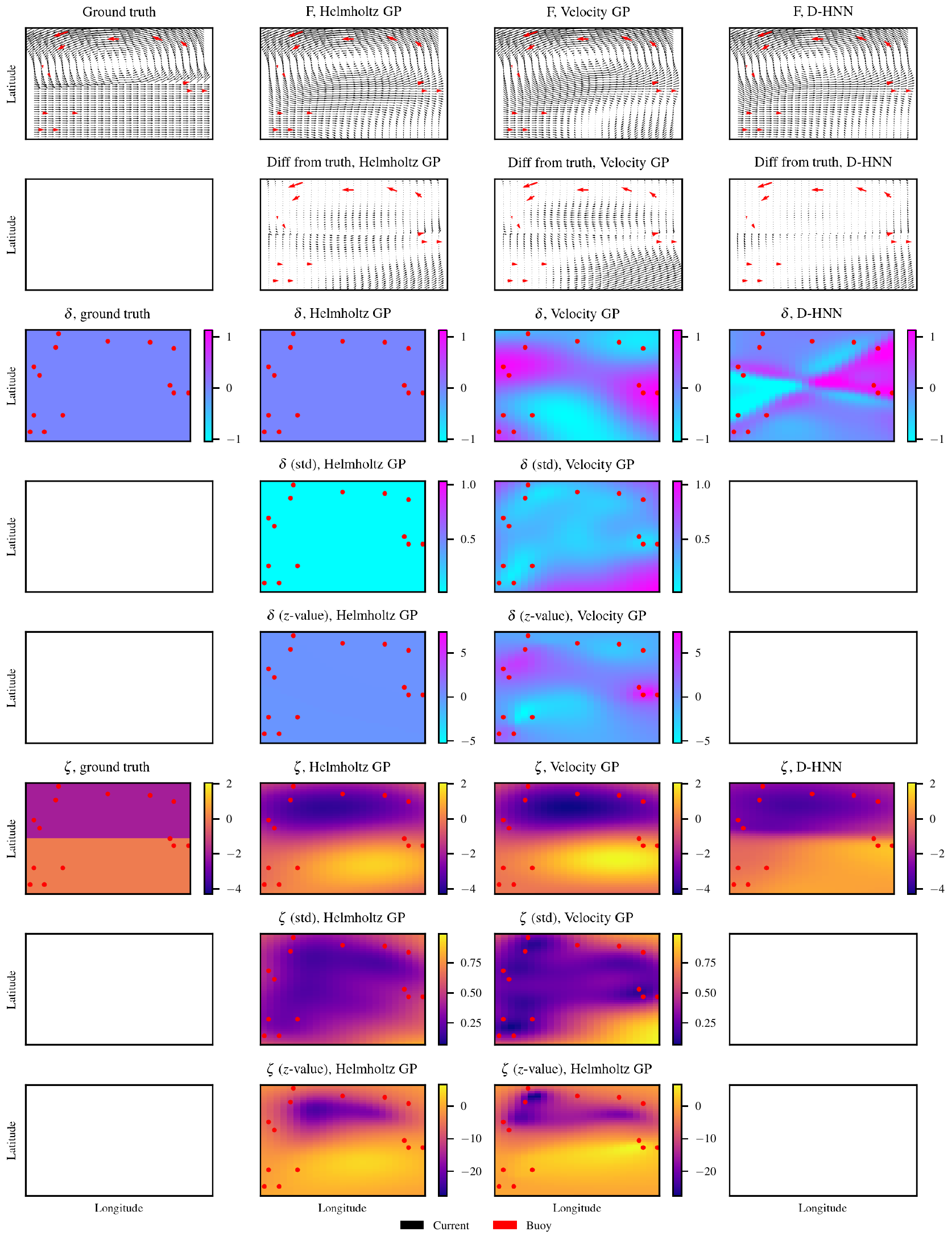}
    \caption{Vortex adjacent to straight current. First column: ground truths. Second column: $\sehelmgp$ results. Third column: $\sevelgp$ results. Fourth column: D-HNN results.}
    \label{fig:contvortexapp}
\end{figure*}
\clearpage
}

\newpage

\subsubsection{Simulated experiment 3: areas of concentrated divergence}\label{app:sec-divergence}

 For this experiment, we consider three different scenarios, differing only in the size of the divergence area. In each of these, we observe 5 buoys, each for 2 time steps. For all the scenarios, the two models perform very well. We show the results in \cref{fig:smalldivapp,fig:middivapp,fig:bigdivapp,}. Both models reconstruct the velocity field well, see \cref{tab:mse} for more details. Moreover, this field has no vorticity and divergence that peaks at the center of the region and slowly decreases in a circular way. This behavior is captured by both models in an accurate way. We conclude that in this important simulated experiment our model is at least as good as the $\sevelgp$ approach.

 A vector field with a single diffuse area of divergence simulates the behavior of an ocean fluid flow in which the water particles are spreading out from a particular region. This behavior can be caused by a variety of factors, such as the movement of warm and cold water masses, or the mixing of fresh and salt water, and can lead to increase in nutrient concentration and high primary production and biodiversity. The fluid particles in this area are not rotating in a circular motion as in a vortex, but instead moving away from each other, resulting in a decrease in density and velocity.

\paragraph{Simulation details.} To simulate a vector field with a divergence area in a two dimensional space, we first define a grid of points $L$ of size 20 x 20, equally spaced over the interval $[-2,2] \times [-2,2]$. The point $(0,0)$ represents the center of the divergence area. To obtain a vector field with divergence area around this point, for each point $\bm{x} = (x^{(1)}, x^{(2)}) \in L$, we can compute the longitudinal and latitudinal velocities by: 
\begin{align*}
    F^{(1)}(\bm{x}) &= \frac{x^{(1)}}{b_d+R_d^2(\bm{x})}  \\
    F^{(2)}(\bm{x}) &= \frac{x^{(2)}}{b_d+R_d^2(\bm{x})}
\end{align*}
with $R_d(\bm{x}) = ((x^{(1)})^2 + (x^{(2)})^2 $ being the distance from the center of divergence, and $b_d$ a parameter governing the size of the area of divergence. Larger $b_d$ implies larger area, but also smaller value at the center. Intuitively, this parameter measures how diffuse the divergence around a center point is. This intuition can be confirmed by computing the actual divergence value: 
\begin{align*}
\delta(\bm{x}) = \mathrm{div} \cdot F = \frac{\partial F^{(1)}}{\partial x^{(1)}} + \frac{\partial F^{(2)}}{\partial x^{(2)}} = \frac{b_d + (x^{(1)})^2 - (x^{(2)})^2}{(b_d+R_d^2(\bm{x}))^2} + \frac{b_d - (x^{(1)})^2 + (x^{(2)})^2}{(b_d+R_d^2(\bm{x}))^2} = \frac{2b_d}{(b_d+R_d^2(\bm{x}))^2}.
\end{align*}
For the vorticity instead we have
\begin{align*}
    \zeta(\bm{x}) = \mathrm{curl} \cdot F = \frac{\partial F^{(1)}}{\partial x^{(2)}} - \frac{\partial F^{(2)}}{\partial x^{(1)}} = \frac{-2(x^{(1)})(x^{(2)})}{(b_d+R_d^2(\bm{x}))^2} + \frac{2(x^{(1)})(x^{(2)})}{(b_d+R_d^2(\bm{x}))^2} = 0.
\end{align*}

The goals for each model then are to (1) reconstruct the velocity field in an accurate way, (2) predict that there is a divergent area and its size, and (3) predict zero vorticity. Finally note that in this experiment, we propose three different scenarios, where the only difference is how diffuse the divergence areas are. Specifically, we run three different experiments with $b_{\text{small}} = 0.4$, $b_{\text{medium}} = 2$, and $b_{\text{big}} = 15$. 

As before, our observations are simulated buoy trajectories. For each scenario the simulation part is the same. We simulate 5 buoys, starting in the non-divergent areas, observed for a total time of 3, and we consider 2 time steps.  Overall we have 10 observations. As usual, to get these trajectories we solve the velocity-time ODE and interpolate.

\paragraph{Model fitting.} For each of the three scenarios, we fit the three models with the routine specified in \cref{app:subsectionvortex}. The hyperparameter initialization for both GPs is always the same across the three different scenarios: $\ell_{\Phi} = 1, \sigma_{\Phi}=1, \ell_{\Psi} = 2.7, \sigma_{\Psi} = 0.368, \obsvar = 0.135$ for the $\sehelmgp$, $\ell_{1} = 1, \sigmaFu=1, \ell_{2} = 2.7, \sigmaFv = 0.368, \obsvar = 0.135$ for the $\sevelgp$. We provide the optimal hyperparameters for each scenario in the corresponding subsections.

\paragraph{Result: small divergence area, $b_{\text{small}} = 0.5$.}

The optimal hyperparameters in this scenario are the following: 
\begin{itemize}
    \item $\ell_{\Phi} = 1.1314, \sigma_{\Phi}=1.9422, \ell_{\Psi} = 5.3132, \sigma_{\Psi} = 0.1864, \obsvar = 0.1821$ for the $\sehelmgp$
    \item $\ell_{1} = 0.5078, \sigmaFu=1.6570, \ell_{2} = 2.7183, \sigmaFv = 1.8658, \obsvar = 0.1396$ for the $\sevelgp$.
\end{itemize}

In \cref{fig:smalldivapp} we show the results of this scenario. As before, for each of the plots, the horizontal and vertical axes represent, respectively, latitude and longitude. The first row represents the ground truth simulated vector field (left), and the reconstruction using the $\sehelmgp$ (center-left), the $\sevelgp$ (center-right) and the D-HNN (right). Red arrows are the observed buoy data, black arrows show the predicted current at test locations. All three models have some problems in reconstructing the underlying field. The two GPs are particularly problematic, because they predict constant strong current that abruptly stops in regions where there are no buoys. The predictions are particularly bad for the $\sevelgp$, which fails to understand the direction and size of the current in most of the region. 
The D-HNN prediction is the one that looks better here, but it is still problematic in the sense that far away from the buoys the current starts to rotate. The plots in the second row showing the difference from the ground truth show that all these models provide poor performances on this task. In terms of RMSE, we have $1.11$ for the $\sehelmgp$, $1.25$ for the $\sevelgp$, and $0.67$ for the D-HNN, confirming that our model performs much better.

In the third row, we analyze the divergence. The left box shows the divergence structure of this field. As described in the preamble, since $b_d$ is small, we have a small area of divergence with big magnitude. The two GP models identify this area. The $\sevelgp$ is more accurate in predicting the size of the divergence area. The $\sehelmgp$ predicts that there is a divergence area in the middle and gets the correct magnitude, but predicts it to be larger than it actually is. If we consider the z-value plots, we can see that this intuition is confirmed: the $\sevelgp$ predicts only a small area to have significant non-zero divergence, whereas our model overestimates the size of this area. The prediction of the D-HNN is less accurate. In terms of RMSEs, we have $2.62$ for the $\sehelmgp$, $1.45$ for the $\sevelgp$, and $4.14$ for the D-HNN. 

In the last three rows of the plot we have, as usual, the vorticity analysis. The left box shows the ground truth. Here the $\sehelmgp$ perfectly predicts zero vorticity, and the D-HNN is almost correct too. The $\sevelgp$, on the contrary, predicts very irregular vorticity, with very high uncertainty. If we consider the z-value plots, we see there is one region (in the center) where the vorticity is predicted to be non-zero in a significant manner. This is a problematic behavior that the $\sevelgp$ has and our model has not. We have $0.0$ RMSE for the $\sehelmgp$, $1.07$ for the $\sevelgp$, and $0.31$ for the D-HNN. 

\paragraph{Result: medium divergence area, $b_{\text{small}} = 5$.}

The optimal hyperparameters in this scenario are the following: 
\begin{itemize}
    \item $\ell_{\Phi} = 1.9387, \sigma_{\Phi}=1.2387, \ell_{\Psi} = 2.3894, \sigma_{\Psi} = 0.2192, \obsvar = 0.0675$ for the $\sehelmgp$ 
    \item $\ell_{1} = 1.6067, \sigmaFu=0.8181, \ell_{2} = 2.7183, \sigmaFv = 0.9859, \obsvar = 0.0742$ for the $\sevelgp$. 
\end{itemize}

\cref{fig:middivapp} shows the results of this scenario. In the top part we have as always the velocity predictions. Since the divergence area is more diffuse, both the velocity of the current and the lengthscale of variation are smaller, in the sense that there are less sharp deviations. Compared to the previous scenario, this property of the field makes the prediction task easier for all three models. In particular, the $\sehelmgp$ and $\sevelgp$ predict a field that almost resembles identically the ground truth. The D-HNN still has some issues, specifically it predict some rotations far away from the observations. This behavior can be seen by looking at the difference from ground truth in the second row. We have the following RMSEs: $0.17$ for the $\sehelmgp$, $0.19$ for the $\sevelgp$, and $0.55$ for the D-HNN.  

For the divergence, by looking at the ground truth plot on the left, we see the area of divergence is now more diffuse, and the magnitude is lower. Both the $\sehelmgp$ and the $\sevelgp$ predict this area accurately, both in terms of size and magnitude (they both predict this area to be a bit larger than it actually is). The D-HNN picks up divergence in a very irregular way. In terms of uncertainty, both GP models are more certain about their predictions around the buoys, and the z-values reflect this behavior: the area where the divergence is significantly different from zero (z-value above 1) is almost identical to the actual ground truth. The RMSEs are: $0.39$ for the $\sehelmgp$, $0.33$ for the $\sevelgp$, $1.32$ for the D-HNN. 

For the vorticity, we observe that the performances of all models are now worse. The $\sehelmgp$ still predicts vorticity very close to zero almost everywhere, but not exactly zero as before. The predictions for the $\sevelgp$ still look less accurate and irregular. The D-HNN performance is very poor. In terms of uncertainty, the $\sehelmgp$ has low uncertainty about its prediction, and this leads to an area where there is significantly non-zero vorticity (in terms of z-value). This behavior is somehow problematic, but note that the predicted mean is in absolute value very close to zero in that area too. The z-value for the $\sevelgp$ is as in the previous scenario, predicting significantly non-zero divergence in an area where the mean is quite distant from zero. Again, this is a very undesirable behavior. The RMSEs are: $0.05$ for the $\sehelmgp$, $0.12$ for the $\sevelgp$, $0.38$ for the D-HNN. 

\paragraph{Result: big divergence area, $b_{\text{small}} = 15$.}

We finally study the last scenario, with the big area of divergence. The optimal hyperparameters in this scenario are the following: 
\begin{itemize}
    \item $\ell_{\Phi} = 3.3732, \sigma_{\Phi}=0.8362, \ell_{\Psi} = 14.7644, \sigma_{\Psi} = 0.0659, \obsvar = 0.0074$ for the $\sehelmgp$ 
    \item $\ell_{1} = 2.3456, \sigmaFu=0.3376, \ell_{2} = 2.7183, \sigmaFv = 0.3355, \obsvar = 0.0055$ for the $\sevelgp$. 
\end{itemize}

In \cref{fig:bigdivapp} we show the results of this scenario. Here the divergence areas are even more diffuse, and this  seems to help a lot the $\sehelmgp$ predictions but not so much the other methods. 

For the velocity prediction task, the three models produce predictions that are close to the truth. It is clear, however, that the predictions of the $\sehelmgp$ are more precise, whereas both the $\sevelgp$ and D-HNN predict some rotational shapes that should not be there. This result is confirmed by the RMSE: $0.04$ for the $\sehelmgp$,  $0.10$ for the $\sevelgp$, and $0.19$ for the D-HNN. 

In terms of divergence, predictions for the two GP models are similar, but our model is slightly better at predicting the full size of the region, with low uncertainty. The D-HNN prediction is again poor. The z-values show how in the central area, both models significantly predict non-zero divergence, but further away in the corners z-values get closer and closer to zero. This behavior is due to the distribution of the buoys' observations. The RMSEs are: $0.05$ for the $\sehelmgp$, $0.12$ for the $\sevelgp$, and $0.27$ for the D-HNN. 

Finally, if we consider the vorticity, we can see how here the $\sehelmgp$ is superior to the other two methods, as in the two previous scenarios. It is able to detect that there is no vorticity, with very low uncertainty. The $\sevelgp$, on the contrary, predicts non-zero positive vorticity in the left side of the plot, and non-zero negative vorticity in the right side. These predictions are with low uncertainty and hence significant, as can be seen by looking at the z-values plot (most of the domain has z-values beyond the thresholds +1 and -1).  The prediction with D-HNN is in similar to the $\sevelgp$ one. The RMSEs are: $0.0$ for the $\sehelmgp$, $0.10$ for the $\sevelgp$, and $0.11$ for the D-HNN. 

In general, we saw how in these experiment the $\sehelmgp$ is at least as good as the other two methods in almost all the velocity prediction tasks, as good as the $\sevelgp$ for the divergence tasks, and remarkably better in predicting that there is no vorticity.

\afterpage{
\begin{figure*}[!h]
    \centering
    \includegraphics[width=0.96\textwidth]{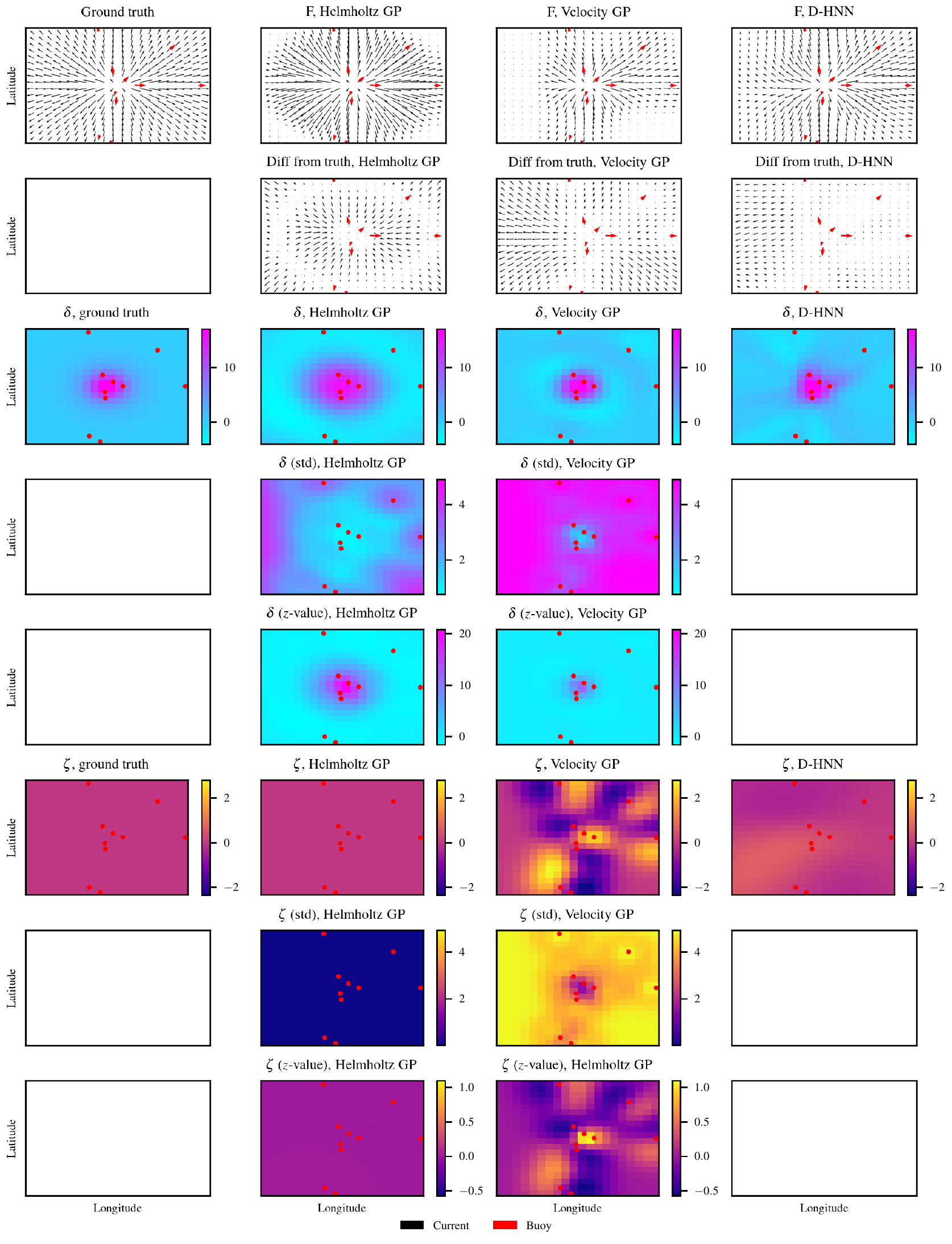}
    \caption{Small area of divergence. First column: ground truths. Second column: $\sehelmgp$ results. Third column: $\sevelgp$ results. Fourth column: D-HNN results.}
    \label{fig:smalldivapp}
\end{figure*}
\clearpage}

\afterpage{
\begin{figure*}[!h]
    \centering
    \includegraphics[width=0.96\textwidth]{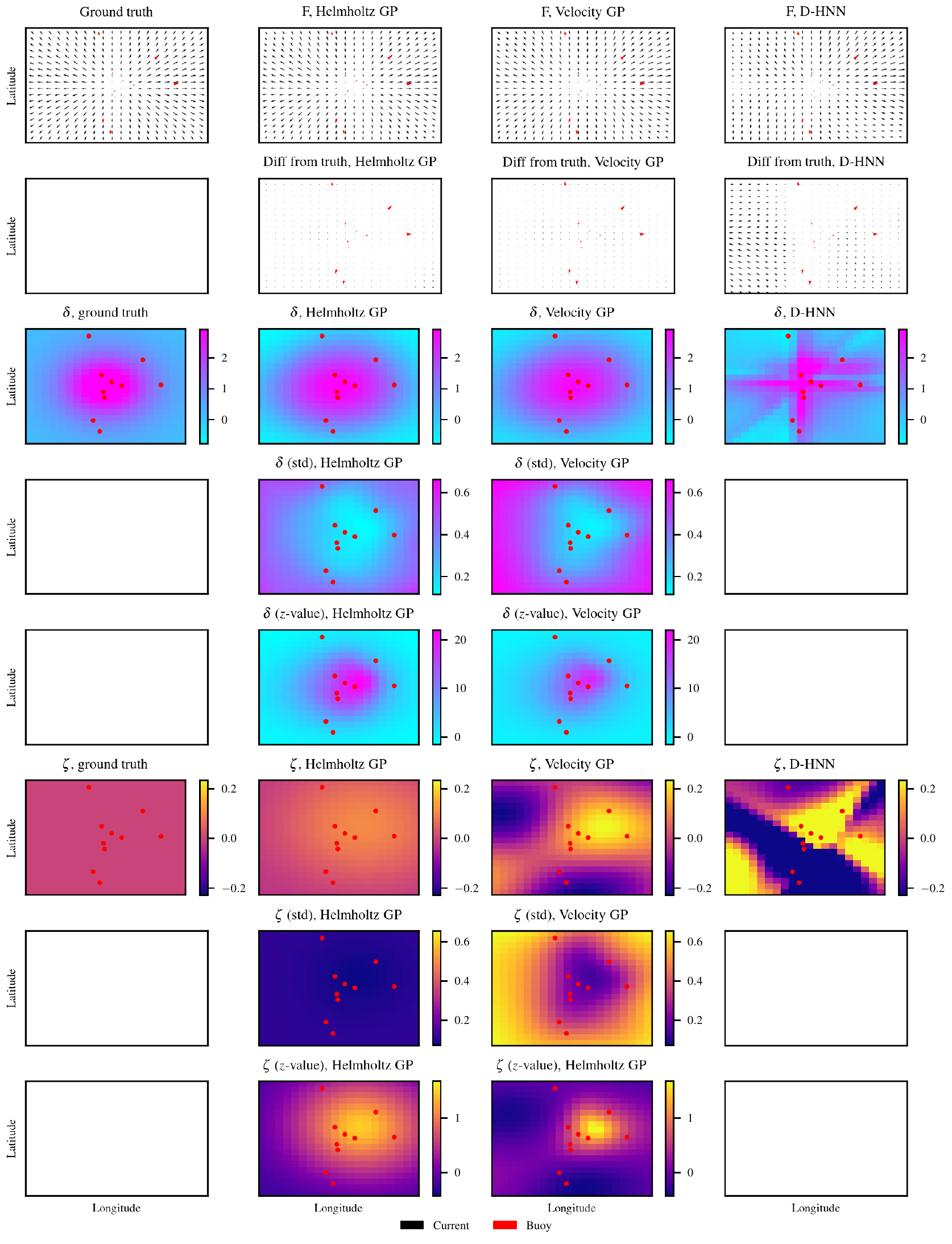}
    \caption{Medium area of divergence. First column: ground truths. Second column: $\sehelmgp$ results. Third column: $\sevelgp$ results. Fourth column: D-HNN results.}
    \label{fig:middivapp}
\end{figure*}
\clearpage}

\afterpage{
\begin{figure*}[!h]
    \centering
    \includegraphics[width=0.96\textwidth]{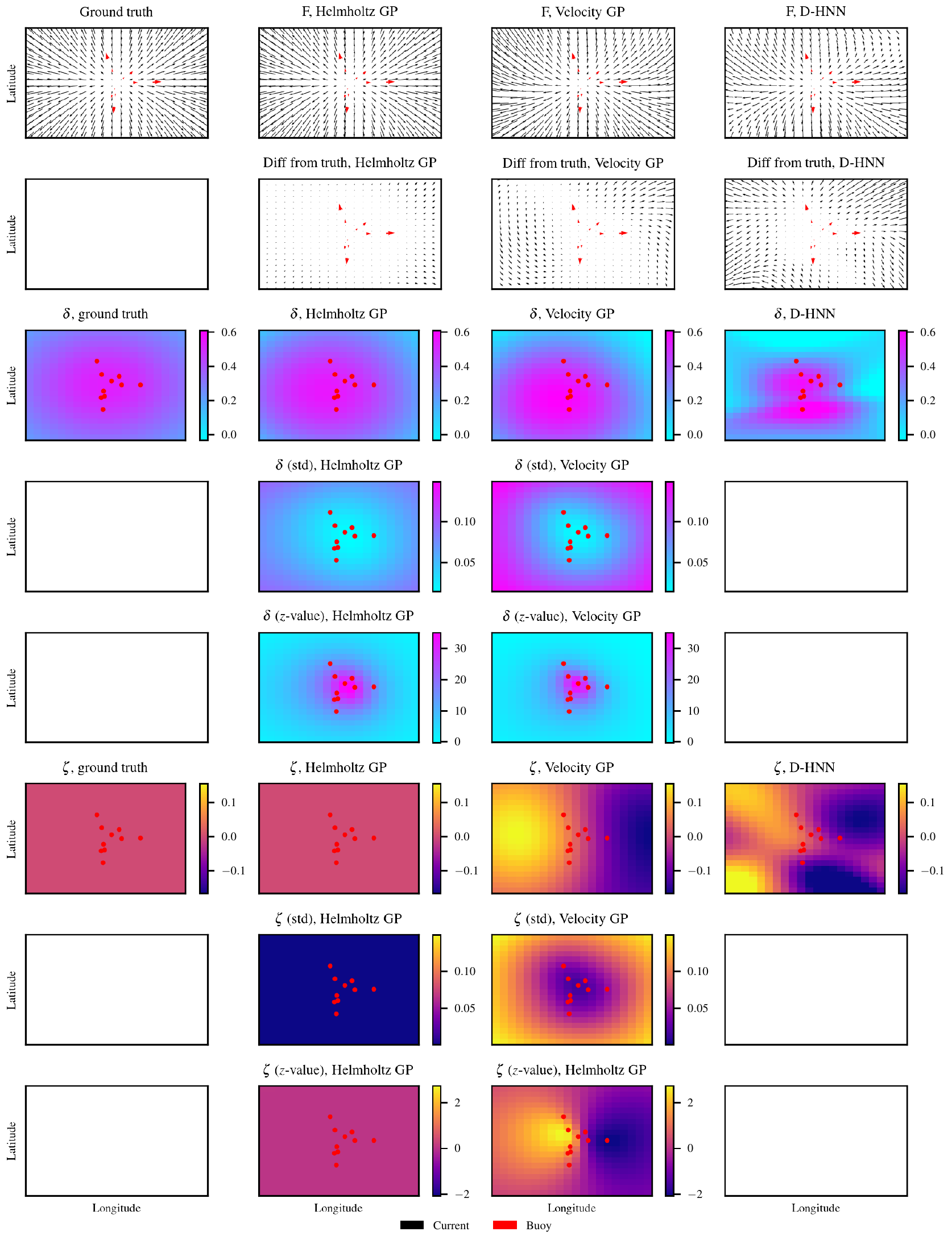}
    \caption{Big area of divergence. First column: ground truths. Second column: $\sehelmgp$ results. Third column: $\sevelgp$ results. Fourth column: D-HNN results.}
    \label{fig:bigdivapp}
\end{figure*}
\clearpage}

\newpage

\subsubsection{Simulated experiment 4: Duffing oscillator with areas of concentrated divergence}\label{app:sec-duffing-divergence}

The Duffing oscillator is a nonlinear dynamic system that can be used to study the dynamics of oceanic phenomena such as tides and currents. In this experiment, we add to this system a divergence area on the left region and a convergence area on the right one. See top-left plot in \cref{fig:smallduffing}. In this way we obtain a field that has both divergence (positive on the right, negative on the left), and vorticity (for the underlying Duffing system). 

\paragraph{Simulation details.} To simulate a Duffing oscillator in a two dimensional space, we first define a grid of points $L$ of size 30 x 30, equally spaced over the interval $[-4,4] \times [-4,4]$. Next, for each point $\bm{x} = (x^{(1)}, x^{(2)}) \in L$, we compute the Duffing longitudinal and latitudinal velocities by: 
\begin{align*}
    \Tilde{F}^{(1)}(\bm{x}) &= x^{(2)} \\
    \Tilde{F}^{(2)}(\bm{x}) &= (x^{(1)} - 0.1*(x^{(1)})^3)*(1+0.1 *\cos (50*\pi / 4)).
\end{align*}
On top of this field we add a divergent field at location $(-3,0)$, using equations:
\begin{align*}
    D^{(1)}(\bm{x}) &= \frac{(x^{(1)} - (-3))}{b_d+R_d^2(\bm{x})}  \\
    D^{(2)}(\bm{x}) &= \frac{x^{(2)}}{b_d+R_c^2(\bm{x})}
\end{align*}
with $R_d(\bm{x}) = (x^{(1)} - (-3))^2 + (x^{(2)} - 0)^2 $ being the distance from the center of divergence, and $b_d$ a parameter governing the size of the area of divergence. Larger $b_d$ implies larger area, but also smaller value at the center. It can be seen as a parameter measuring how diffuse the divergence around a center point is. We also have a convergent field around $(3,0)$, determined by the equations: 
\begin{align*}
    C^{(1)}(\bm{x}) &= - \frac{(x^{(1)} - 3)}{b_c+R_d^2(\bm{x})}  \\
    C^{(2)}(\bm{x}) &= - \frac{(x^{(2)})}{b_c+R_c^2(\bm{x})}
\end{align*}
with $R_c(\bm{x}) = (x^{(1)} - 3)^2 + (x^{(2)} - 0)^2 $, the distance from the center of convergence. To get the full velocity field, we sum up these three quantities: 
\begin{align*}
    F^{(1)}(\bm{x}) = \Tilde{F}^{(1)}(\bm{x}) +  D^{(1)}(\bm{x}) + C^{(1)}(\bm{x})\\
    F^{(2)}(\bm{x}) = \Tilde{F}^{(2)}(\bm{x}) +  D^{(2)}(\bm{x}) + C^{(2)}(\bm{x}).
\end{align*}
In this system, the divergence and vorticity do not have a simple form, but can be calculated. For the sake of our divergence analysis, it is sufficient to say that there are two areas of interest, around the center of divergence and convergence. In this experiment, we propose three different scenarios, where the only difference is how diffuse the divergence areas are. For simplicity, we assume $b = b_c = b_d$, and we run three different experiments with $b_{\text{small}} = 0.5$, $b_{\text{medium}} = 5$, and $b_{\text{big}} = 15$. 

As done before, to predict currents, divergence, and vorticity we simulate buoys. For each scenario the simulation part is the same. We first simulate 3 buoys, starting in the non-divergent areas, observed for a total time of 5, and 2 time steps. We then simulate 4 additional buoys, starting around the divergent areas, for a total time of 5, and 4 time steps. That is, we make observations coarser for buoys in these regions. Overall we have 22 observations. As usual, to get these observations we solve the velocity-time ODE and interpolate.

\paragraph{Model fitting.} For each of the three scenarios, we fit the three models with the routine specified in \cref{app:subsectionvortex}. The hyperparameter initialization for both GPs is the same across the three different scenarios: $\ell_{\Phi} = 1, \sigma_{\Phi}=1, \ell_{\Psi} = 2.7, \sigma_{\Psi} = 0.368, \obsvar = 0.135$ for the $\sehelmgp$, $\ell_{1} = 1, \sigmaFu=1, \ell_{2} = 2.7, \sigmaFv = 0.368, \obsvar = 0.135$ for the $\sevelgp$. We provide the optimal hyperparameters for each scenario in the corresponding subsections.

\paragraph{Result: small divergence area, $b_{\text{small}} = 0.5$.}

The optimal hyperparameters in this scenario are the following: 
\begin{itemize}
    \item $\ell_{\Phi} = 0.6335, \sigma_{\Phi}=0.3734, \ell_{\Psi} = 3.9115, \sigma_{\Psi} = 6.9294, \obsvar = 0.0083$ for the $\sehelmgp$
    \item $\ell_{1} = 0.7212, \sigmaFu=1.8767, \ell_{2} = 2.7183, \sigmaFv = 1.1361, \obsvar = 0.0084$ for the $\sevelgp$.
\end{itemize}

In \cref{fig:smallduffing} we show the results of this scenario. As before, for each of the plots, the horizontal and vertical axes represent, respectively, latitude and longitude. The first row represents the ground truth simulated vector field (left), and the reconstruction using the $\sehelmgp$ (center-left), the $\sevelgp$ (center-right) and the D-HNN (right). Red arrows are the observed buoy data, black arrows show the predicted current at test locations. First of all, we can see how our method predicts accurately the duffing structure in the left part of the plot, whereas has some issues in the right one, where we have the convergence area. The $\sevelgp$ prediction is more problematic: the correct current is predicted around the buoys, but farther away the prediction goes to zero, reverting to the prior mean. This is a problematic behavior, e.g., because it predicts very non-continuous currents. The D-HNN prediction is problematic as well: the current looks more continuous, but the general shape is very different from the ground truth. This behavior can be seen well from the second row, the comparison to the ground truth. In terms of RMSE, we have $0.96$ for the $\sehelmgp$, $2.05$ for the $\sevelgp$, and $2.14$ for the D-HNN, confirming that our model performs much better.

In the third row, we analyze the divergence. The left box shows the divergence structure of this field. There is a small area with very positive divergence on the left, and a small area with very negative divergence on the right. The two GP models are good at identifying these areas. At the same time, they both predict some other areas of divergence around the observed buoys. Nonetheless, if we consider the z-value plots (on the fifth row) we can see how the z-values for both models are very high in the two areas of divergence, meaning that there is a strongly significant non-zero mean in those areas, as desired. The D-HNN predicts a quite different divergence structure. The RMSEs are: $0.94$ for the $\sehelmgp$, $0.95$ for the $\sevelgp$, and $1.89$ for the D-HNN.  

Finally, in the last three rows we analyze results for the vorticity. The left box shows the ground truth. Here the $\sehelmgp$ prediction looks more accurate than the other two. Nonetheless, even our model is not fully able to capture the full vorticity structure. The predictions for the $\sevelgp$ look particularly problematic because it is highly affected by the location of the buoys, and that is reflected in the uncertainty and z-values plots. The D-HNN predicts a very different field on this task as well. The RMSEs are: $1.40$ for the $\sehelmgp$, $2.28$ for the $\sevelgp$, and $2.64$ for the D-HNN.

\paragraph{Result: medium divergence area, $b_{\text{small}} = 5$.}

The optimal hyperparameters in this scenario are the following: 
\begin{itemize}
    \item $\ell_{\Phi} = 1.2029, \sigma_{\Phi}=0.1666, \ell_{\Psi} = 3.3679, \sigma_{\Psi} = 9.5514, \obsvar = 0.0112$ for the $\sehelmgp$ 
    \item $\ell_{1} = 6.7677, \sigmaFu=4.5316, \ell_{2} = 2.7183, \sigmaFv = 23.3219, \obsvar = 0.0305$ for the $\sevelgp$. 
\end{itemize}

\cref{fig:midduffing} shows the results of this scenario. In the top part we have as always the velocity predictions. In this case, the ground truth field is very similar to before, but the divergence areas are more diffuse, and hence the current has generally longer lengthscale of variation (that is, the deviations are less sharp). This feature helps the predictions for all three methods. We can see indeed how now the three models produce predictions that are closer to the truth than before. Still, by looking at the difference from ground truth plots, we can see that the prediction of our model is slightly better than the $\sevelgp$, and significantly better than the D-HNN. We have the following RMSEs: $0.19$ for the $\sehelmgp$, $0.60$ for the $\sevelgp$, and $1.65$ for the D-HNN. These confirm what can see visually in the plots. 

In terms of divergence, by looking at the the ground truth plot on the left, one can immediately notice how the areas of divergence are now more diffuse, and the magnitudes are lower. The $\sehelmgp$ predicts accurately the two areas, with some noise in the central region. The $\sevelgp$ is less accurate, but overall understand that there are these two areas. The D-HNN fails in identifying the two regions. It is interesting to observe the z-value plots in this experiment: for the $\sehelmgp$, the z-values are very high in the two desired areas, meaning that our model is very certain about divergence being different from zero in those areas. For the $\sevelgp$, the z-values still look good, just less accurate than for our model.  The RMSEs are: $0.14$ for the $\sehelmgp$, $0.50$ for the $\sevelgp$, and $1.15$ for the D-HNN.  

Finally, we consider the vorticity. Here the two GP models agree significantly on the shape of their predictions, and they are both very similar to the ground truth. This result is reflected in the RMSEs: $0.24$ for the $\sehelmgp$, $0.26$ for the $\sevelgp$. The prediction for the D-HNN is far from the truth (RMSE $2.39$). The uncertainty is lower close to the data for both GP models. In general, both GP models seem to work well in recovering divergence and vorticity in this scenario. The $\sehelmgp$ is superior for the divergence, the $\sevelgp$ for the vorticity.  

\paragraph{Result: big divergence area, $b_{\text{small}} = 15$.}

The optimal hyperparameters in this scenario are the following: 
\begin{itemize}
    \item $\ell_{\Phi} = 2.9194, \sigma_{\Phi}=0.4599, \ell_{\Psi} = 3.2411, \sigma_{\Psi} = 10.1815, \obsvar = 0.0137$ for the $\sehelmgp$ 
    \item $\ell_{1} = 7.3457, \sigmaFu=4.0581, \ell_{2} = 2.7183, \sigmaFv = 24.7519, \obsvar = 0.0202$ for the $\sevelgp$. 
\end{itemize}

In \cref{fig:bigduffing} we show the results of this scenario. Here the divergence areas are even more diffuse, and the overall field ends up having longer lengthscale of variation. The results on velocity predictions, divergence, and vorticity are aligned with the medium size scenario. 

For the velocity prediction task, the three models produce predictions that are close to the truth. Now the two GP models are similar, as can be seen in the difference from the truth plots, and they are both significantly better than the D-HNN. This result is confirmed by the RMSEs: $0.41$ for the $\sehelmgp$,  $0.22$ for the $\sevelgp$, and $1.63$ for the D-HNN.

In terms of divergence, the $\sehelmgp$ accurately predicts the two areas of divergence, still with some noise in the central region. The $\sevelgp$ is less accurate, especially in the top right region, but overall understand that there are these two areas. The D-HNN prediction is poor. As in the past experiment, it is interesting to observe the z-value plots: both GP models have very high z-values in the areas of divergence, proving their ability to capture the locations of these. The RMSEs are: $0.08$ for the $\sehelmgp$, $0.17$ for the $\sevelgp$, $1.10$ for the D-HNN. 

Finally, also if we consider the vorticity, the results are similar to the previous scenario. Predictions are good for the two GPs, with meaningful z-values. Now the $\sevelgp$ predictions align almost perfectly with the ground truth, and this is reflected in the lower RMSE ($0.16$ vs. $0.48$ for the $\sehelmgp$). The D-HNN still fails to predict structure precisely ($2.41$ RMSE) 

In summary, with this experiment we showed that the $\sehelmgp$ is generally better than the other models in predicting the underlying velocity field (significantly better in the first scenario). In terms of divergence and vorticity, we do not see a large difference compared to the $\sevelgp$: both models are very good; $\sehelmgp$ is slightly better for the divergence and $\sevelgp$ is slightly better for the vorticity. This behavior is very interesting, showing how both models are able to predict a complex divergence pattern (more complex than the previous experiment).

\afterpage{
\begin{figure*}[!h]
    \centering
    \includegraphics[width=0.96\textwidth]{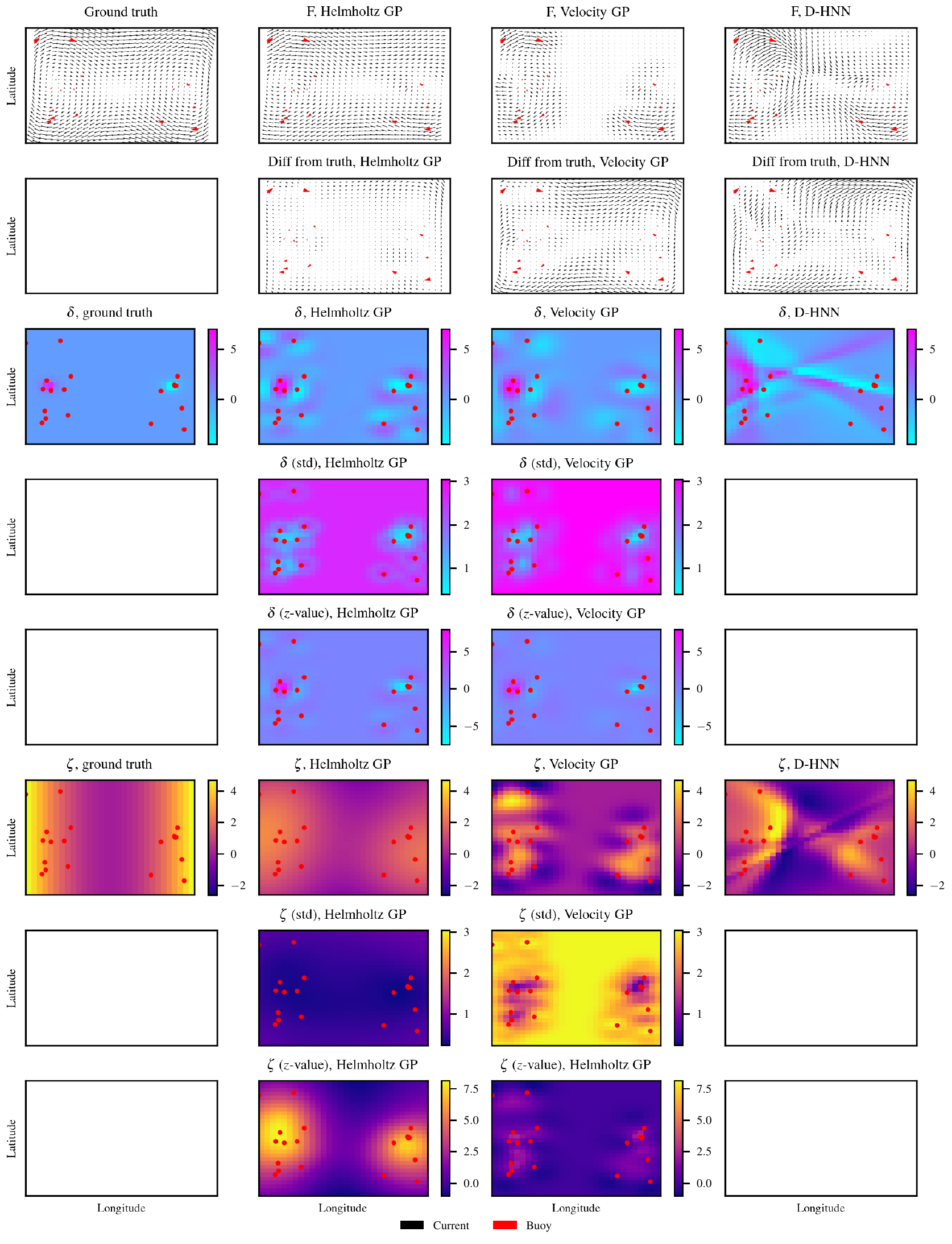}
    \caption{Duffing with small area of divergence. First column: ground truths. Second column: $\sehelmgp$ results. Third column: $\sevelgp$ results. Fourth column: D-HNN results.}
    \label{fig:smallduffing}
\end{figure*}
\clearpage}

\afterpage{
\begin{figure*}[!h]
    \centering
    \includegraphics[width=0.96\textwidth]{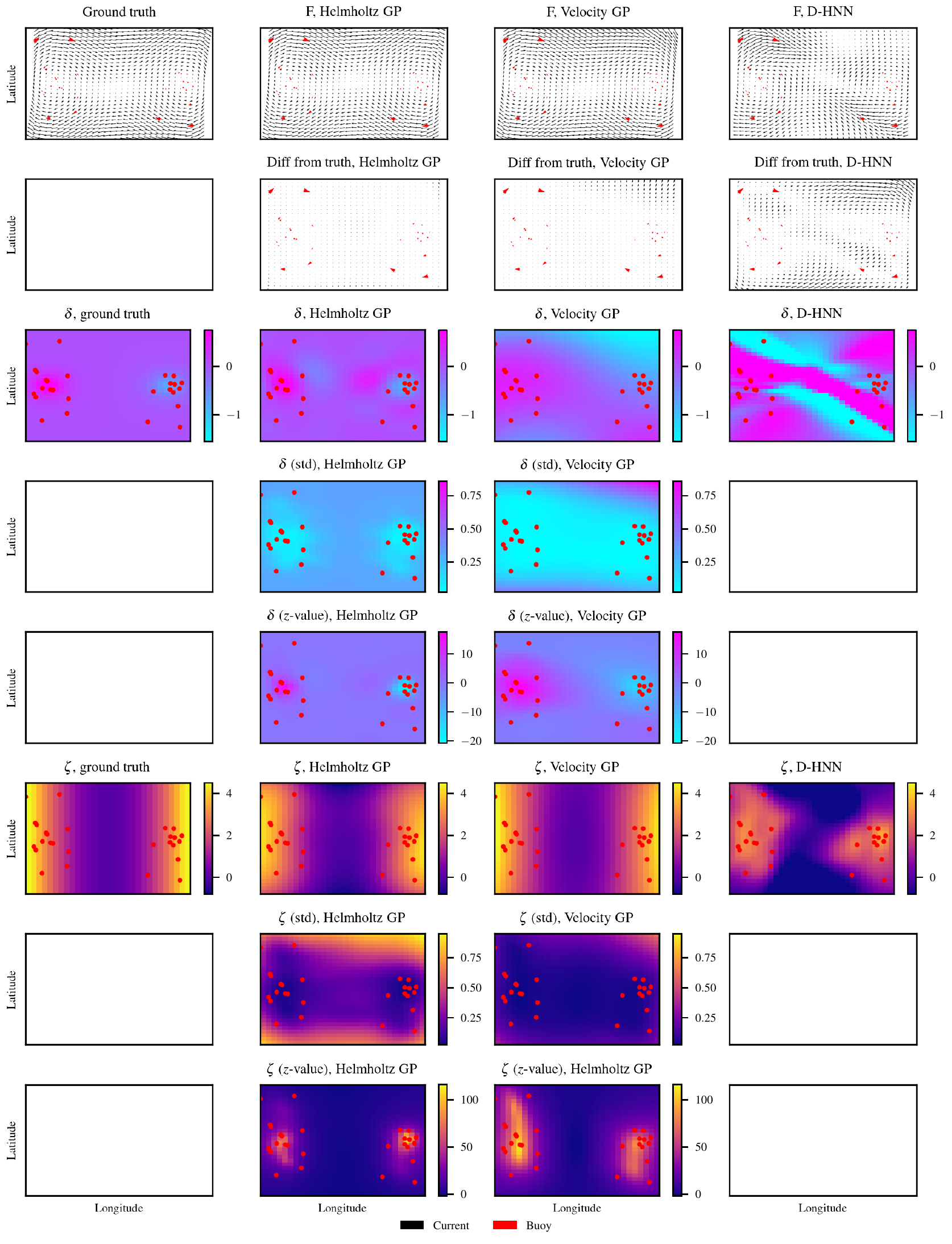}
    \caption{Duffing with medium area of divergence. First column: ground truths. Second column: $\sehelmgp$ results. Third column: $\sevelgp$ results. Fourth column: D-HNN results.}
    \label{fig:midduffing}
\end{figure*}
\clearpage}

\afterpage{\begin{figure*}[!h]
    \centering
    \includegraphics[width=0.96\textwidth]{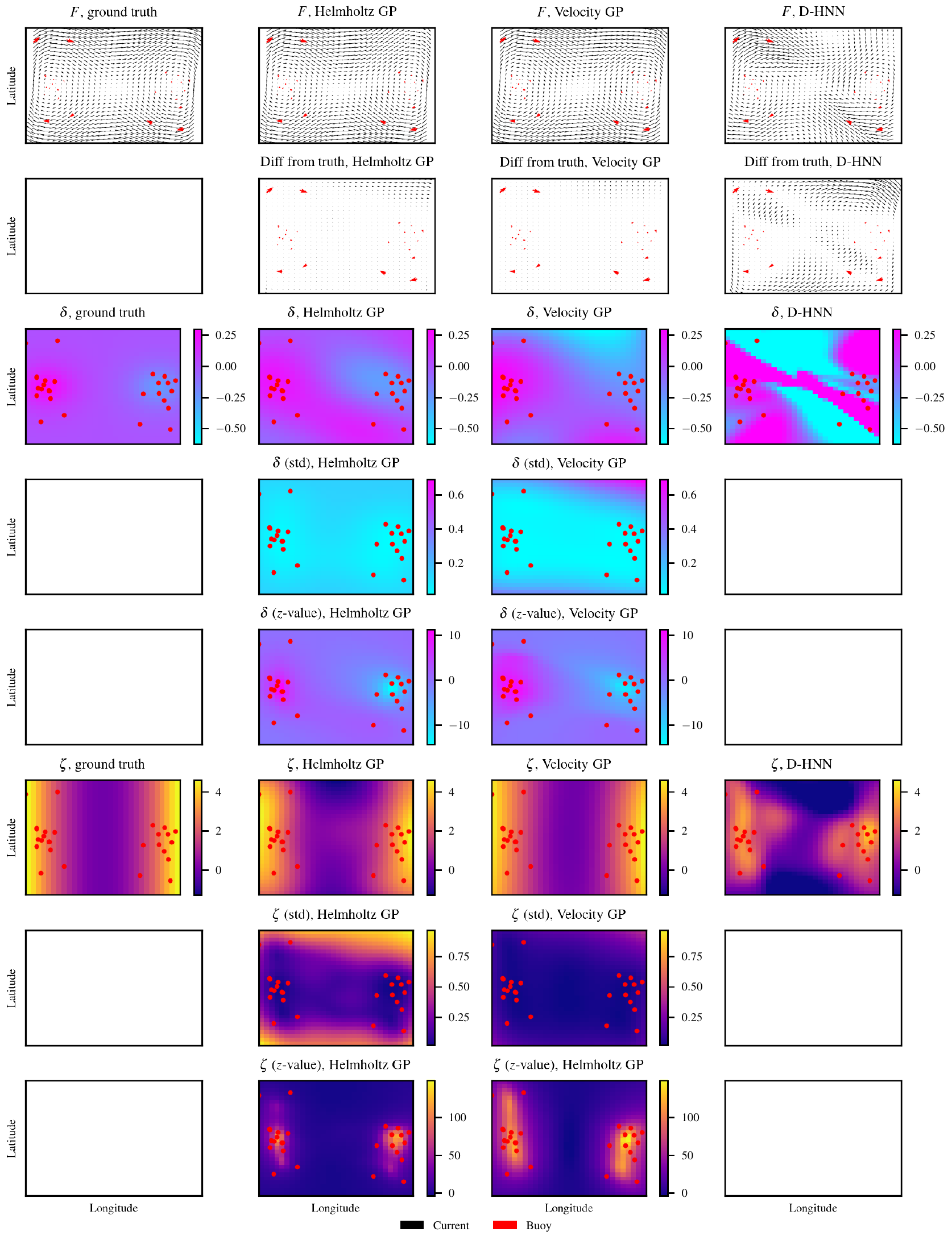}
    \caption{Duffing with big area of divergence. First column: ground truths. Second column: $\sehelmgp$ results. Third column: $\sevelgp$ results. Fourth column: D-HNN results.}
    \label{fig:bigduffing}
\end{figure*}
\clearpage}

\newpage

\subsection{Real-world data 1: LASER}\label{app:experiments-laser}

The LAgrangian Submesoscale ExpeRiment, or LASER \citep{novelli2017biodegradable}, was performed in the Gulf of Mexico in January-February 2016. Around 10 million data points were retrieved from more than 1000 near-surface biodegradable CODE-type ocean drifters (drogued at a depth of one meter) tracked in real-time using SPOT GPS units. These data were then preprocessed as described in \citet{yaremchuk2014filtering}. Finally, since satellite data can have errors and positions of buoys sometimes jump an unrealistic amount, oceanographers removed some bad points that were visible by eye. The preprocessed data are available at \url{https://data.gulfresearchinitiative.org/data/R4.x265.237:0001} \citep{dasaroLASERdata2017}.  In our analysis, we use locations and velocities of buoys as they appear in this dataset. 

The main goal of the experiment was to obtain data to understand the submesoscale ocean currents in the open ocean environment near the DeSoto Canyon, as well as how oil or other pollutants might be transported via these currents. In our analysis, we consider a subsample of the LASER data, in an area where the oceanographers expect a convergent front to be (from visual inspection of drifter data). This particular structure in the ocean happens when there are two different masses of water that collide and cause the formation of an area where water sinks. This behavior could happen when two water masses with different temperatures and/or salinities meet, or when water masses from different directions go towards the same area, such as the meeting of warm equatorial water and cold polar water. These fronts are very important for understanding ocean circulation and weather patterns, and can also be a source of nutrients for marine life. To study this structure, we consider two experiments: in the first one, we run our model on a small subset of buoys from this region, collapsing the time dimension and downsampling the observations. To confirm our finds, we then run our models on a dataset that contains more buoys and observations, still from that region.

\subsubsection{LASER, convergent front, sparse}
In this analysis, we consider 19 buoys, observed every fifteen minutes over a two hour time horizon. By downsampling by a factor of 3 and collapsing the time dimension, we obtain 55 observations. In these data, oceanographers expect to see a clear convergent front in the left region of the spatial domain.  

\paragraph{Model fitting.} The optimization routine is exactly the same that we do for the simulated experiments: gradient-based Adam algorithm until convergence or a sufficient amount of iterations has elapsed. For the initial hyperparameters, we have tried various alternatives, and found out that the predictions do not change significantly. Hence, for coherence, we stick to the usual initialization done for synthetic data, i.e.,  $\ell_{\Phi} = 1, \sigma_{\Phi}=1, \ell_{\Psi} = 2.7, \sigma_{\Psi} = 0.368, \obsvar = 0.135$ for the $\sehelmgp$, and $\ell_{1} = 1, \sigmaFu=1, \ell_{2} = 2.7, \sigmaFv = 0.368, \obsvar = 0.135$ for the $\sevelgp$. The optimal hyperparameters obtained are: $\ell_{\Phi} = 1.6032, \sigma_{\Phi}=0.0496, \ell_{\Psi} = 13.3272, \sigma_{\Psi} = 1.6392, \obsvar = 0.0232$ for the $\sehelmgp$, and $\ell_{1} = 8.3149, \sigmaFu=0.1384, \ell_{2} = 2.7183, \sigmaFv = 0.1318, \obsvar = 0.0276$

\paragraph{Results.} We show the results in \cref{fig:lasersparse}. The top row shows the predictions for the three models. As before, red arrows are the observed buoy data. 
The black arrows show the current posterior means at test locations. The test locations are 400 points evenly sparse on a 20 x 20 grid that covers the full range of latitude and longitude of our buoys' observations. The three models produce very similar results: a quasi-constant flow towards the south-west area of the region. There is a slight difference in prediction for the region where buoys seem to converge ($\sevelgp$ and D-HNN do not predict different current around there, $\sehelmgp$ predicts a more converging behavior). 

This difference is clear when we look at the posterior divergence plots, in the second row. Our model predicts a negative divergence area (in light-blue) in the area where the oceanographers expect a convergent front. On the contrary, the $\sevelgp$ predicts no divergence on the whole spatial domain. This is a very important difference, showing how our model can perform better in recovering this very important property of the ocean. Note that this same intuition is confirmed if we look at the fourth row, where we have z-value plots for both models: the z-values for the $\sehelmgp$ around the expected convergent front are strongly negative, meaning that the divergence there is significantly non-zero, as desired. 

For the vorticity, we just have very small values, almost zero, for both models. Unfortunately, there is no oceanographic knowledge to predict the vorticity far away from the observed drifter traces, and therefore we can not conclude anything related to this point.

\subsubsection{LASER, convergent front, full}\label{app:experiments-laser-full}
To further validate the result on the divergence, we consider the same buoys floating over a nine hour time horizon, downsampled by a factor of 3, obtaining 240 observations. We fit our models by performing the usual optimization routine, and we plot the results in \cref{fig:laserfull}. 

In the top row we show the prediction results. For all the models, the predictions around the buoy agree almost perfectly with predictions from the sparse experiment for the $\sehelmgp$; further away models, are more conservative and closer to the prior. The divergence plots in the second row are of the most interest. The prediction according to $\sevelgp$ changes remarkably relative to the past experiment. Now it matches closely the Helmholtz result, and both methods detect the convergent front. This result shows the strength of our model in being more data efficient, a very desirable property for a GP model. 

\afterpage{
\begin{figure*}[!h]
    \centering
    \includegraphics[width=0.98\textwidth]{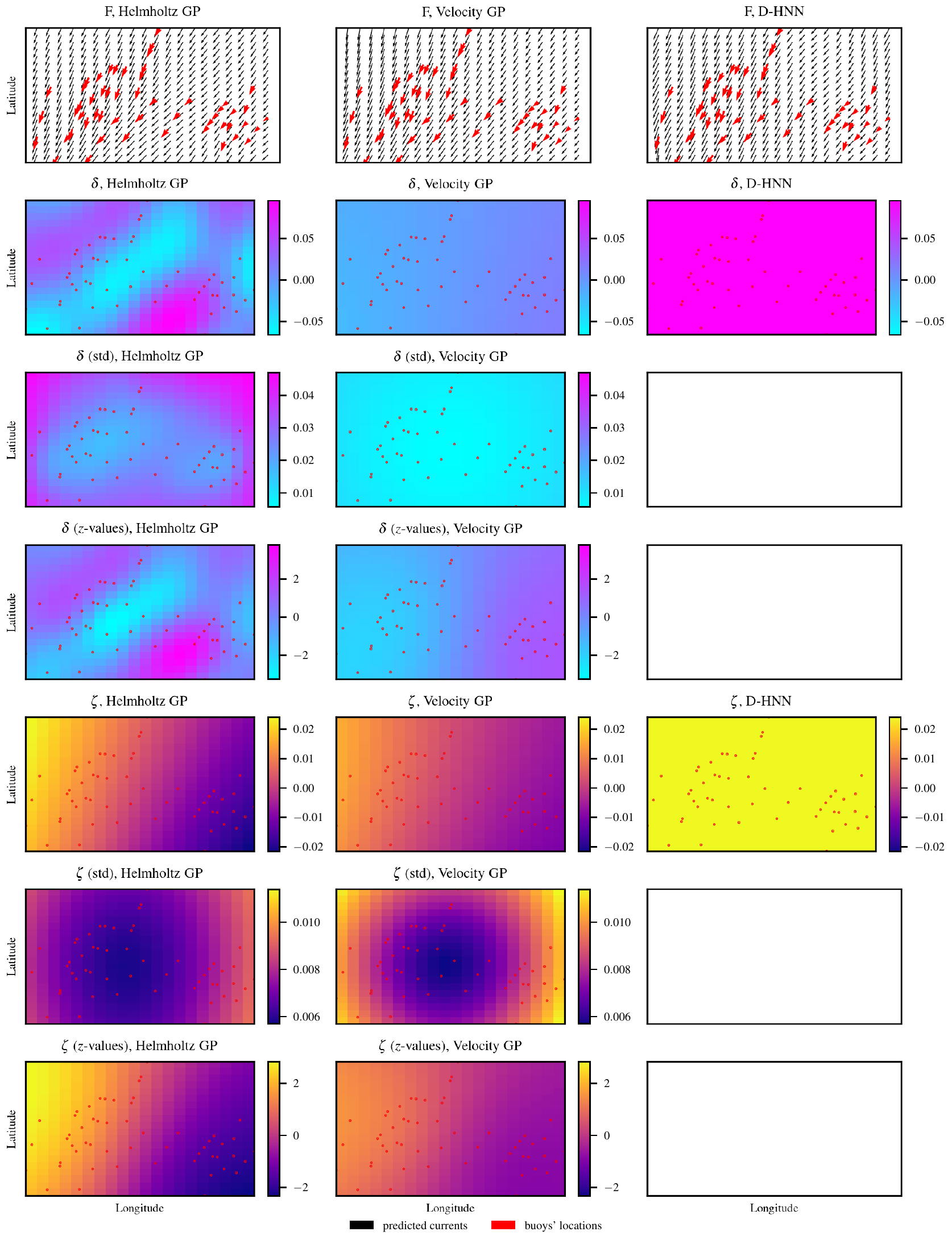}
    \caption{LASER sparse. First column: $\sehelmgp$ results. Second column: $\sevelgp$ results. Third column: D-HNN results.}
    \label{fig:lasersparse}
\end{figure*}
\clearpage}

\afterpage{
\begin{figure*}[!h]
    \centering
    \includegraphics[width=0.98\textwidth]{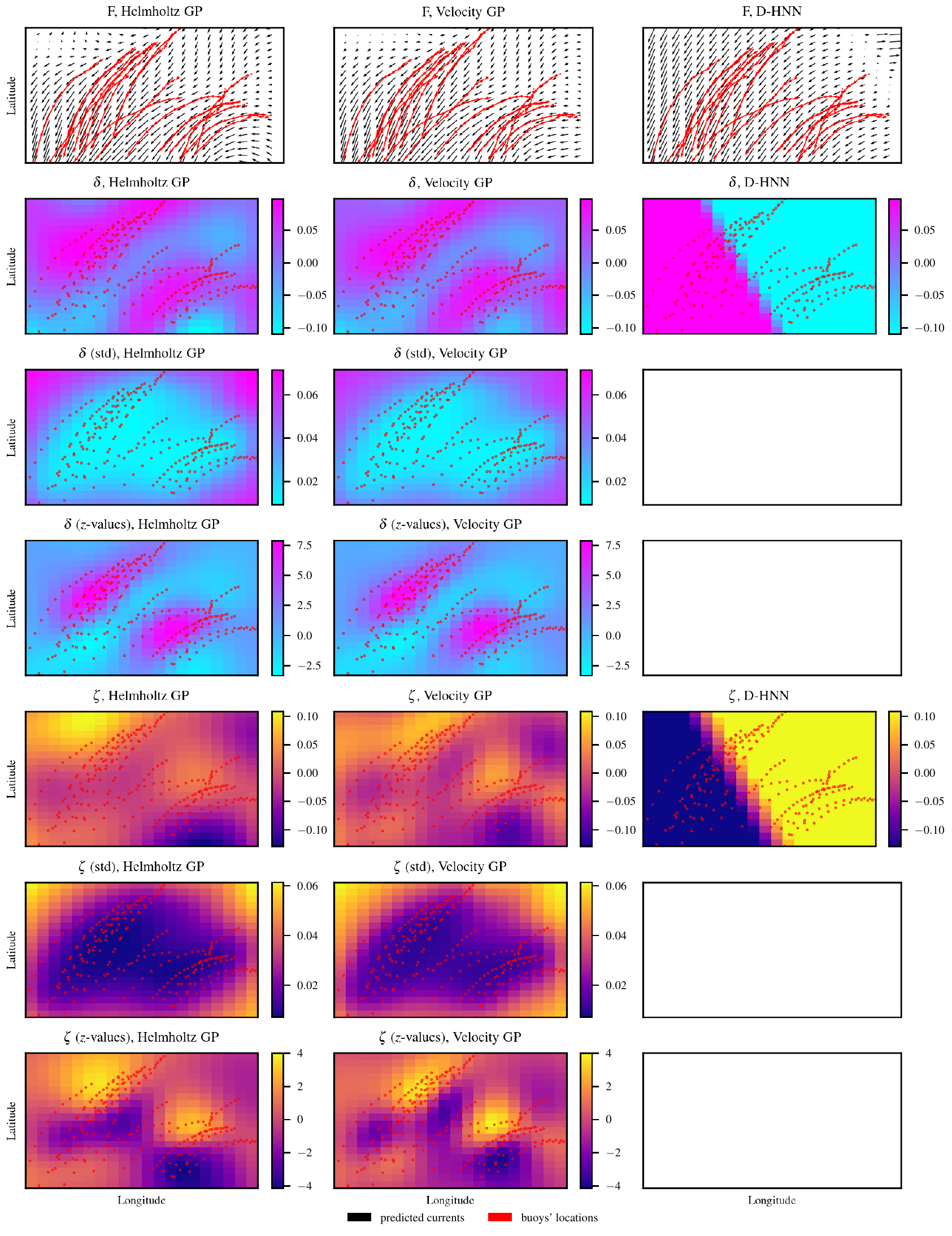}
    \caption{LASER complete. First column: $\sehelmgp$ results. Second column: $\sevelgp$ results. Third column: D-HNN results.}
    \label{fig:laserfull}
\end{figure*}
\clearpage}

\newpage

\subsection{Real-world data 2: GLAD}\label{app:experiments-glad}
The Grand Lagrangian Deployment (GLAD) experiment \citep{ozgokmen2012carthe} is another experiment conducted in the northern Gulf of Mexico in July 2012. More than 300 custom-made buoys (of the same type as in the LASER experiment) were deployed near the Deepwater Horizon site and Louisiana coast. This experiment was originally intended to help advance research in understanding the spread and dispersion of oil after the Deepwater Horizon tragedy. Researchers have been using this dataset to study interactions among ocean flows, the levels of influence on transport that large and small flows have, and the size of oil spread at which large flows dominate. Since the GLAD experiment was conducted in the summer time with a shallow 20-meter surface mixed layer for the buoys, the wind has a very strong impact on the trajectories, creating a lot of oscillations. These oscillations are due to a balance of forces due to wind forcing and Earth’s rotation, and get amplified during summer time. Filtering these oscillations is a very complicated task, so this wind-induced motions represent a true problem for buoys that are used for measuring oceanographic parameters. Note that we do not see these issues with the LASER data, because that was a winter experiment, where the surface layer is 100-meter deep and devoid of these oscillations.

\paragraph{Model fitting.} To deal with this issue, we consider a limited subset of our dataset. We take drifter traces of 12 buoys, observed hourly over a four days time horizon. We collapse the time dimension and downsample these traces by a factor 50, obtaining 85 observations. In terms of optimization routine, we follow very similarly what done in all the other experiments. The only difference is that here different hyperparameter optimization led to different prediction plots for some combinations. In our final results, we decided to stick to the hyperparameter initialization for which both the $\sehelmgp$ and the $\sevelgp$ results were visually more appealing. These are $\ell_{\Phi} = 12.18, \sigma_{\Phi}=0.135, \ell_{\Psi} = 7.4, \sigma_{\Psi} = 3, \obsvar = 0.135$ for the $\sehelmgp$, $\ell_{1} = 2.7, \sigmaFu=1, \ell_{2} = 2.7, \sigmaFv = 1, \obsvar = 0.135$ for the $\sevelgp$. 

The optimal hyperparameters obtained after the optimization routine are  $\ell_{\Phi} = 45.6840, \sigma_{\Phi}=0.0362, \ell_{\Psi} = 80.1871, \sigma_{\Psi} = 13.5514, \obsvar = 0.1715$ for the $\sehelmgp$, and $\ell_{1} = 72.5835, \sigmaFu=0.2622, \ell_{2} = 2.7183, \sigmaFv = 0.1354, \obsvar = 0.1739$ for the $\sevelgp$. 

\paragraph{Results.} In these data, we expect to see a continuous current with no sharp deviations (i.e., lengthscale of variation is long), with few smaller vortices distributed across the region. Unfortunately, here there is no explicit divergence structure that oceanographers expect, so any conclusion from the divergence and vorticity plots is difficult to verify. We show the results of the experiments in \cref{fig:gladsparse}. We have the predictions in the first row. As before, red arrows are the observed buoy data. The black arrows show the current posterior means at test locations. First of all, the D-HNN model makes physically implausible predictions, likely due to the sparse nature of the data on a large domain. For the GP models, both prediction plots look reasonable, but there are two regions of interest showing important issues with the $\sevelgp$. Consider the bottom right corner. Despite evidence of a strong current making a u-turn, the standard approach shows an abrupt drop in current away from observed data. Our method, on the contrary, predicts a strong current connecting across drifters, in accordance with the continuity of currents (the idea that when a fluid is in motion, it must move in such a way that mass is conserved). This behavior is very problematic. Consider then the top-left corner. Flow behavior around the observations suggests that there might be a vortex in that region. The standard approach shows none. With the $\sehelmgp$, instead, we can see the expected vortex between the two lines of current. 

To further prove our point, we increase the number of observations to 1200, by decreasing the downsampling factor, and we re-fit the two models with the same optimization routine. The velocity prediction results are included in the first row of \cref{fig:gladfull}. Here we can see that our model starts being affected by the oscillations in the data, predicting currents with shorter lengthscale of variation. But also it is still able to reconstruct a continuous current, also far away from the observations, with some vortices with shorter length scale. For the $\sevelgp$, the discontinuity issues increase significantly, and the model is still unable to detect vortices. These are two strong motivations to believe the $\sehelmgp$ provides a better alternative for this task. The prediction of the D-HNN remains poor. 

In terms of divergence and vorticity reconstruction on the sparse dataset, the $\sehelmgp$ predicts very small divergence almost everywhere, and vorticity coherent with the buoys trajectories. The $\sevelgp$, instead, predicts a reasonable vorticity field, but the divergence shows irregular patterns that look more suspicious. See the second and third blocks in \cref{fig:gladsparse} for a visual comparison.
By looking at the data, we can see how there are regions on the left where buoys observations seem to be more affected by the oscillations. The $\sevelgp$ is more influenced by this noise than our model, and hence predicts divergence areas around the buoys. This claim can be validated by looking at the plots when the dataset size increases. See the second and third blocks in \cref{fig:gladfull}. Here, both models seem to be affected more by the oscillations, but the $\sehelmgp$ still predicts divergence closer to zero, whereas the $\sevelgp$ predicts divergence areas around each conglomerate of buoys in the region. Therefore, we can conclude that our model is at least as good as the $\sevelgp$. Note that we cannot say anything stronger, because there is no expert knowledge suggesting that the $\sehelmgp$ behavior is the expected one.

\afterpage{
\begin{figure*}[!h]
    \centering
    \includegraphics[width=0.96\textwidth]{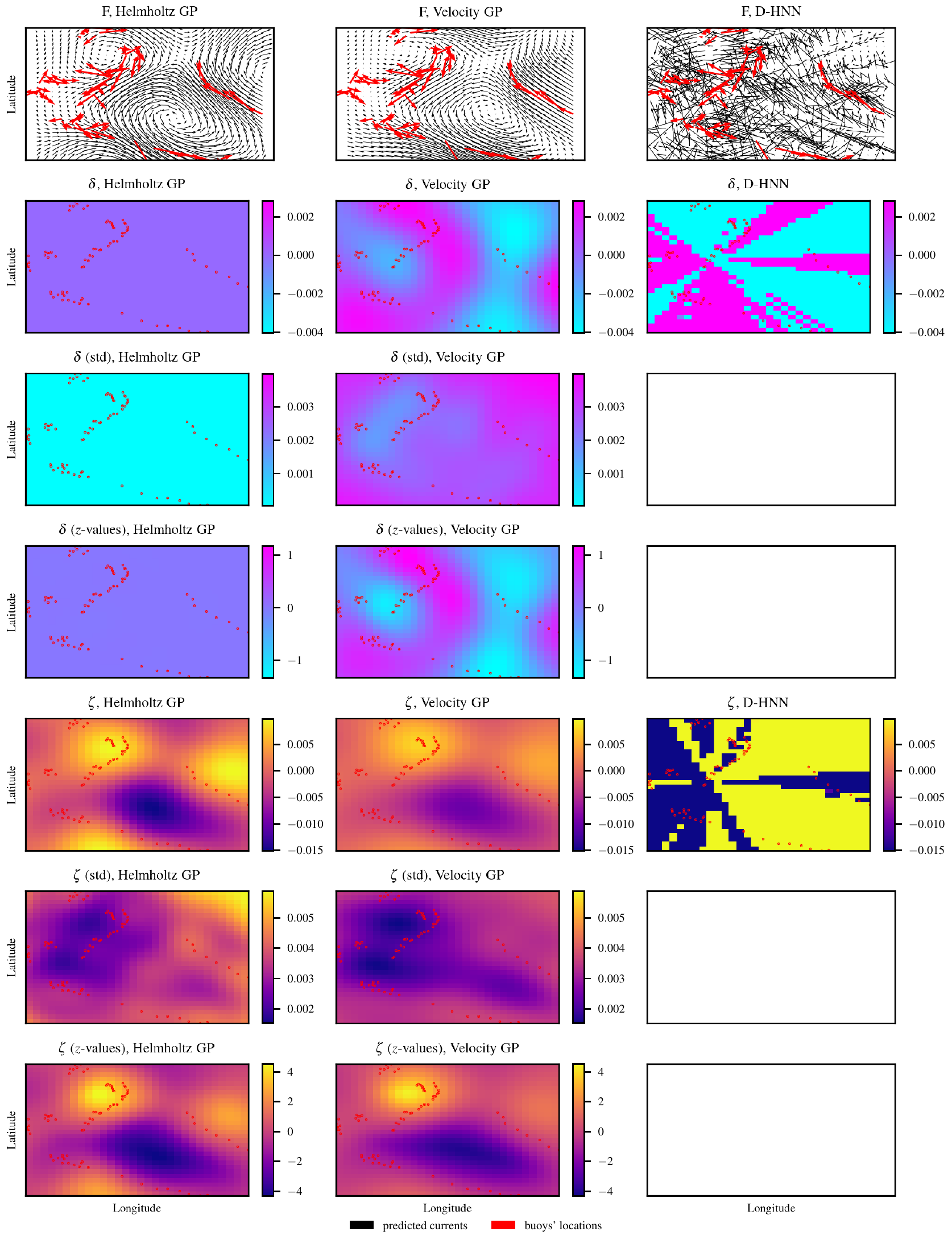}
    \caption{GLAD sparse. First column: $\sehelmgp$ results. Second column: $\sevelgp$ results. Third column: D-HNN results.}
    \label{fig:gladsparse}
\end{figure*}
\clearpage}
\afterpage{
\begin{figure*}[!h]
    \centering
    \includegraphics[width=0.96\textwidth]{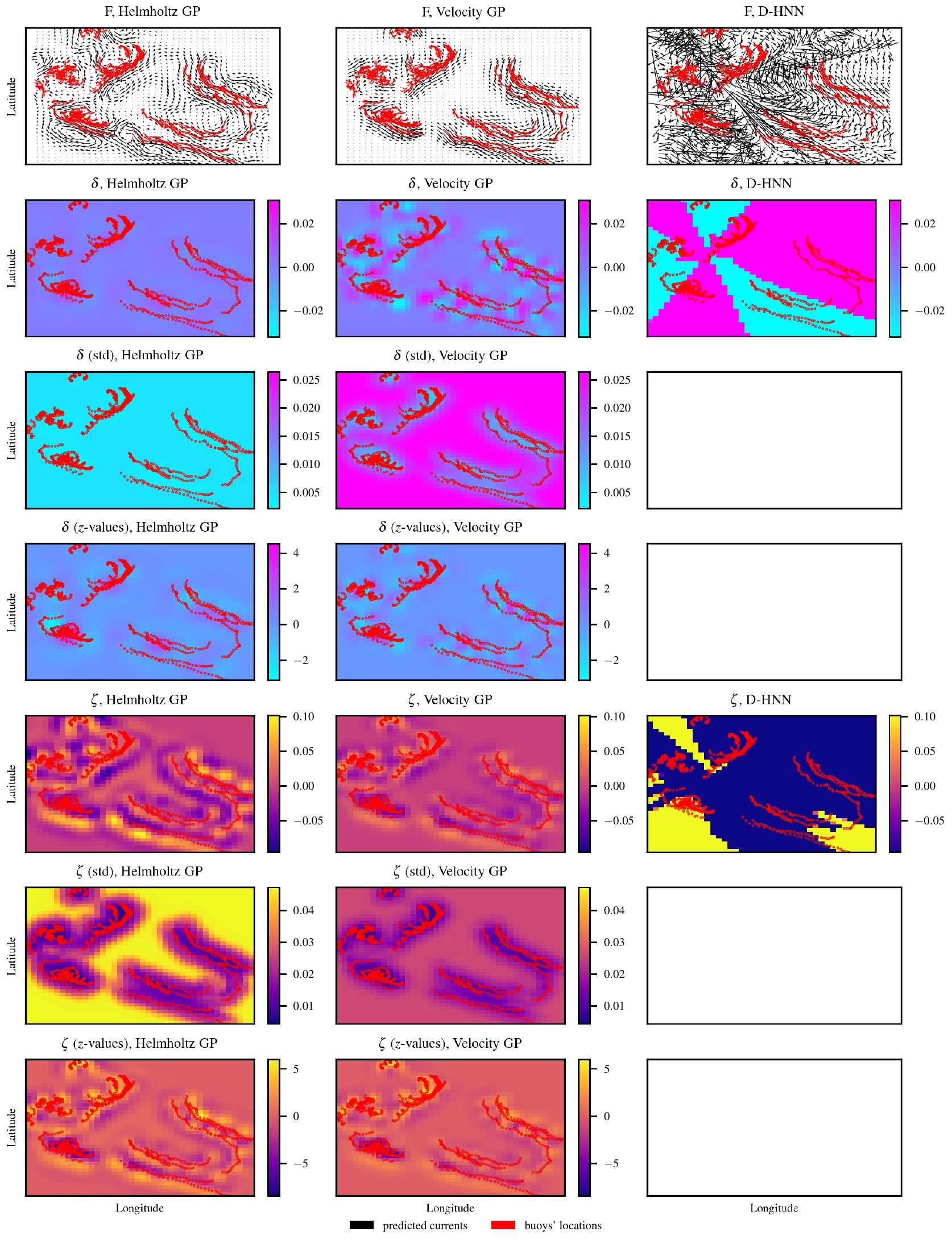}
    \caption{GLAD complete. First column: $\sehelmgp$ results. Second column: $\sevelgp$ results. Third column: D-HNN results.}
    \label{fig:gladfull}
\end{figure*}
\clearpage}

\end{document}